%% file: heo.tex
\documentclass[11pt]{article}
\usepackage[colorlinks=true,
            citecolor=black,
            linkcolor=black,
            urlcolor=black]{hyperref}
\usepackage{fullpage}
\usepackage{latexsym,amsmath,amsfonts,amssymb,stmaryrd,amsthm}
\usepackage{url}
\usepackage{proof}
\usepackage[square,authoryear,semicolon]{natbib}
\usepackage{cleveref}
\usepackage[all]{xy}

\theoremstyle{definition}
\newtheorem{definition}{Definition}

\theoremstyle{remark}
\newtheorem{remark}[definition]{Remark}

\theoremstyle{plain}
\newtheorem{lemma}[definition]{Lemma}
\newtheorem{theorem}[definition]{Theorem}
\crefname{lemma}{Lemma}{Lemmas} % always capitalize
\crefname{theorem}{Theorem}{Theorems} % always capitalize
\crefname{definition}{Definition}{Definitions} % always capitalize
\crefname{section}{Section}{Sections} % always capitalize

\newcommand{\G}{\ensuremath{\Gamma}}
\newcommand{\g}{\ensuremath{\gamma}}
\newcommand{\e}{\ensuremath{\varepsilon}}
\newcommand{\eb}{\overline{\e}}
\newcommand{\eq}{\ensuremath{\mathbin{\doteq}}}
\newcommand{\fresh}{\ \#\ }
\newcommand{\fd}[1]{\ensuremath{\mathsf{FD}(#1)}}

\newcommand{\oft}[2]{#1\mathbin{:}#2}

\newcommand{\dimm}{\mathsf{dim}}

\newcommand{\subst}[3]{\ensuremath{#1 [#2 / #3]}}
\newcommand{\dsubst}[3]{\ensuremath{#1 \langle{#2}/{#3}\rangle}}
\newcommand{\tsubst}[2]{\ensuremath{{#1}{#2}}}
\newcommand{\tsext}[3]{\ensuremath{#1 [#2 / #3]}}

\newcommand{\arr}[2]{\ensuremath{#1 \to #2}}

\newcommand{\lam}[2]{\ensuremath{\lambda{#1}.{#2}}}
\newcommand{\app}[2]{\ensuremath{\mathsf{app}({#1},{#2})}}

\newcommand{\prd}[2]{\ensuremath{#1 \times #2}}

\newcommand{\pair}[2]{\ensuremath{\langle #1,#2\rangle}}
\newcommand{\fst}[1]{\ensuremath{\mathsf{fst}(#1)}}
\newcommand{\snd}[1]{\ensuremath{\mathsf{snd}(#1)}}

\newcommand{\bool}{\ensuremath{\mathsf{bool}}}
\newcommand{\notb}[1]{\ensuremath{\mathsf{not}_{#1}}}
\newcommand{\notf}[1]{\ensuremath{\mathsf{not}(#1)}}
\newcommand{\notel}[2]{\ensuremath{\mathsf{notel}_{#1}(#2)}}
\newcommand{\true}{\ensuremath{\mathsf{true}}}
\newcommand{\false}{\ensuremath{\mathsf{false}}}
\newcommand{\ifsym}{\ensuremath{\mathsf{if}}}
\newcommand{\ifb}[4]{\ensuremath{\ifsym_{#1}(#2;#3,#4)}}

\newcommand{\C}{\ensuremath{\mathbb{S}^1}}
\newcommand{\base}{\ensuremath{\mathsf{base}}}
\newcommand{\lp}[1]{\ensuremath{\mathsf{loop}_{#1}}}
\newcommand{\Celimsym}{\ensuremath{\mathbb{S}^1\mathsf{\text{-}elim}}}
\newcommand{\Celim}[4]{\ensuremath{\Celimsym_{#1}({#2};{#3},{#4})}}

\newcommand{\coesym}{\ensuremath{\mathsf{coe}}}
\newcommand{\coe}[4]{\ensuremath{\mathsf{coe}_{#1}^{#2 \rightsquigarrow #3}(#4)}}
\newcommand{\coegeneric}[1]{\coe{#1}{r}{r'}{M}}

\newcommand{\hcomsym}{\ensuremath{\mathsf{hcom}}}
\newcommand{\hcom}[6]{\ensuremath{%
\hcomsym_{#2}^{#1}(#3\rightsquigarrow #4,#5;#6)}}
\newcommand{\hcomgeneric}[2]{\hcom{#1}{#2}{r}{r'}{M}{y.N^0,y.N^1}}

\newcommand{\steps}{\ensuremath{\longmapsto}}
\newcommand{\evals}{\ensuremath{\Downarrow}}
\newcommand{\isval}[1]{\ensuremath{#1\ \mathsf{val}}}
\newcommand{\isvaltab}[1]{\ensuremath{#1 & \ \mathsf{val}}}

\newcommand{\wftm}[2][\Psi]{\ensuremath{{#2}\ \mathsf{tm}\ [#1]}}
\newcommand{\wfval}[2][\Psi]{\ensuremath{#2\ \mathsf{val}\ [#1]}}
\newcommand{\wfdim}[2][\Psi]{\ensuremath{#2\ \dimm\ [#1]}}

\newcommand{\msubsts}[3]{\ensuremath{{#2} : {#1} \to {#3}}}
\newcommand{\msubst}[1]{\msubsts{\Psi'}{#1}{\Psi}}
\newcommand{\td}[2]{\ensuremath{{#1}{#2}}}

\newcommand{\per}[2][\Psi]{\ensuremath{{#2}\ \mathsf{per}\ [#1]}}
\newcommand{\vinper}[4][\Psi]{\ensuremath{{#2} \approx^{#1}_{#4} {#3}}}
\newcommand{\vinperfour}[6][\Psi]{\ensuremath{{#2} \approx^{#1}_{#6} {#3}
  \approx^{#1}_{#6} {#4} \approx^{#1}_{#6} {#5}}}
\newcommand{\inpersym}{\sim}
\newcommand{\inper}[4][\Psi]{\ensuremath{{#2} \sim^{#1}_{#4} {#3}}}
\newcommand{\inpertab}[4][\Psi]{\ensuremath{{#2} & \sim^{#1}_{#4} {#3}}}
\newcommand{\inperfour}[6][\Psi]{\ensuremath{{#2} \sim^{#1}_{#6} {#3}
  \sim^{#1}_{#6} {#4} \sim^{#1}_{#6} {#5}}}
\newcommand{\eqper}[3][\Psi]{\ensuremath{{#2} \approx^{#1} {#3}}}
\newcommand{\eqperfour}[5][\Psi]{\ensuremath{{#2} \approx^{#1} {#3} \approx^{#1}
{#4} \approx^{#1} {#5}}}
\newcommand{\cpretype}[2][\Psi]{\ensuremath{{#2}\ \mathsf{pretype}\ [#1]}}
\newcommand{\ceqpretype}[3][\Psi]{\ensuremath{{#2}\eq {#3} \ \mathsf{pretype}\ [#1]}}
\newcommand{\wfctx}[2][\Psi]{\ensuremath{{#2}\ \mathsf{ctx}\ [#1]}}
\newcommand{\cwftype}[2][\Psi]{\ensuremath{{#2}\ \mathsf{type}\ [#1]}}
\newcommand{\ceqtype}[3][\Psi]{\ensuremath{{#2} \eq {#3}\ \mathsf{type}\ [#1]}}
\newcommand{\coftype}[3][\Psi]{\ensuremath{{#2} \in {#3}\ [#1]}}
\newcommand{\oftype}[4][\Psi]{\ensuremath{{#2} \gg {#3} \in {#4}\ [#1]}}
\newcommand{\eqtm}[5][\Psi]{\ensuremath{{#2} \gg {#3} \eq {#4} \in {#5}\ [#1]}}
\newcommand{\ceqtm}[4][\Psi]{\ensuremath{{#2} \eq {#3} \in {#4}\ [#1]}}
\newcommand{\ceqtmtab}[4][\Psi]{\ensuremath{{#2} \eq\ & {#3} \in {#4}\ [#1]}}

\title{Computational Higher Type Theory I:\\Abstract Cubical Realizability}

\author{Carlo Angiuli\thanks{\texttt{cangiuli@cs.cmu.edu}}\\Carnegie Mellon University
  \and Robert Harper\thanks{\texttt{rwh@cs.cmu.edu}}\\Carnegie Mellon University
  \and Todd Wilson\thanks{\texttt{twilson@csufresno.edu}}\\California State University Fresno}

\date{April, 2016}

\begin{document}

\maketitle{}

\begin{abstract}
  Brouwer's constructivist foundations of mathematics is based on an
  intuitively meaningful notion of computation shared by all
  mathematicians.  Martin-L\"{o}f's \emph{meaning explanations} for
  constructive type theory define the concept of a type in terms of
  computation.  Briefly, a type is a complete (closed) program that
  evaluates to a \emph{canonical type} whose members are complete
  programs that evaluate to \emph{canonical elements} of that type.
  The explanation is extended to incomplete (open) programs by
  \emph{functionality}: types and elements must respect equality in
  their free variables.  Equality is evidence-free---two types or
  elements are at most equal---and equal things are implicitly
  interchangeable in all contexts.

  \emph{Higher-dimensional type theory} extends type theory to account
  for \emph{identifications} of types and elements.  An identification
  witnesses that two types or elements are explicitly interchangeable
  in all contexts by an explicit transport, or coercion, operation.
  There must be sufficiently many identifications, which is ensured by
  imposing a generalized form of the \emph{Kan condition} from
  homotopy theory.  Here we provide a Martin-L\"{o}f-style meaning
  explanation of simple higher-dimensional type theory based on a
  programming language that includes Kan-like constructs witnessing
  the computational meaning of the higher structure of types.  The
  treatment includes an example of a higher inductive type (namely,
  the 1-dimensional sphere) and an example of Voevodsky's
  \emph{univalence} principle, which identifies equivalent types.

  The main result is a \emph{computational canonicity theorem} that
  validates the computational interpretation: a closed boolean
  expression must always evaluate to a boolean value, even in the
  presence of higher-dimensional structure.  This provides the first
  fully computational formulation of higher-dimensional type theory.
\end{abstract}
%Moreover, two types are \emph{exactly equal} when they evaluate to
%exactly equal canonical types, and two elements are exactly equal when
%they evaluate to exactly equal canonical members.  

\input{overview}

\pagebreak % looks less bad
\subsection*{Acknowledgements}

Apart from the overarching influence of Martin-L\"{o}f and Constable,
the main antecedents to this work are the two-dimensional type theory
given by~\citet{lh2dtt}, the uniform Kan cubical model of homotopy
type theory given by~\citet{bch} and the subsequent cubical type
theories given by~\citet{cohen2016cubical}\footnote{Also discussed in
  a number of informal notes by Coquand circulated on
  \url{homotopytypetheory@googlegroups.com}.}  and
by~\citet{licata2014cubical}.  We are greatly indebted to Marc Bezem,
Evan Cavallo, Kuen-Bang Hou (Favonia), Simon Huber, Dan Licata, and Ed
Morehouse who have provided much helpful feedback on this work.

The first two authors gratefully acknowledge the support of the Air
Force Office of Scientific Research through MURI grant
FA9550-15-1-0053.  Any opinions, findings and conclusions or
recommendations expressed in this material are those of the authors
and do not necessarily reflect the views of the AFOSR.  The third
author gratefully acknowledges California State University, Fresno for
supporting his sabbatical semester, and Carnegie Mellon University for
making possible his visit.

\clearpage
\input{opsem}
\newpage
\input{meanings}
\newpage
\input{types}

\newpage
\input{summary}
\newpage

\bibliographystyle{plainnat}
\bibliography{heo}

\end{document}

%%% Local Variables:
%%% mode: latex
%%% TeX-master: t
%%% End:

%% file: overview.tex
\section{Introduction}

The goal of this work is to develop a computation-based account of
higher-dimensional type theory for which canonicity at observable
types is true by construction.  Types are considered as descriptions
of the computational behavior of terms, rather than as formal syntax
to which meaning is attached separately.  Types are structured as
collections of terms of each finite dimension.  At dimension zero the
terms of a type are its ordinary members; at higher dimension terms
are lines between terms of the next lower dimension.  The terms of
each dimension satisfy coherence conditions ensuring that the terms
may be seen as abstract cubes.  Each line is to be interpreted as an
identification of two cubes in that it provides evidence for their
exchangeability in all contexts.  It is required that there be
sufficiently many lines that this interpretation is tenable.  For
example, lines must be reversible and closed under concatenation, so
that the identifications present the structure of a pre-groupoid.
Moreover, there must be further lines witnessing the unit, inverse,
and associativity laws of concatention, the structure of an
$\infty$-groupoid.

In this paper we give a ``meaning explanation'' of a computational
higher type theory in the style of Martin-L\"{o}f and of Constable and
Allen, \textit{et
  al.}~\citep{cmcp,martin1984intuitionistic,constableetalnuprl,allen2006innovations}.
Such an explanation starts with a dimension-stratified collection of
terms endowed with a deterministic operational semantics defining what
it means to evaluate closed terms of any dimension to canonical form.
The dimension of a term is the finite set of dimension names it contains; these
dimension names may be thought of as variables ranging over an abstract
interval, in which case terms may be thought of as tracing out lines in a type.
The end points, $0$ and $1$, of the interval may be substituted to obtain the
end points of such lines.  Dimension names may be substituted for one
another without restriction, allowing dimensions to be renamed, identified, or
duplicated.  The semantics of types is given by specifying, at each dimension,
when canonical elements are equal, when general elements are equal, and when
these definitions capture the structure of an $\infty$-groupoid, namely
when they are \emph{cubical} and satisfy the \emph{uniform Kan
condition}~\citep{bch}.
%The former implies that the structure is closed under substitution, and the
%latter that certain closure conditions are witnessed by programs.

For the sake of clarity, we illustrate this method for a simple type
theory with higher inductive types, one line between types given by an
equivalence, and closed under function and product types.  The main
technical result is the following \emph{canonicity theorem} for closed
terms of boolean type:
\begin{center}
  \textit{A closed term of boolean type of dimension zero has a unique value,
  which is either true or false.}
\end{center}
To our knowledge this is the first result of this kind for higher-dimensional
type theory.

In a follow-on paper we intend to extend our results to consider
type-indexed families of types, and in particular a type of
identifications of the members of a type.  Consideration of families
requires no new semantic machinery, merely the introduction of more lines
between types at each dimension, whose semantics are
already accounted for in the basic setup.  Adding identification types
requires a mild generalization of the Kan operations, but otherwise
presents no new difficulties.

The main remaining question is whether this framework can be extended
to account for Voevodsky's univalence axiom~\citep{hottbook}.  There
is by now strong evidence that it can be given computational meaning
(chiefly, the recent work by~\citet{cohen2016cubical} and ongoing work
by~\citet{hubercanonicity} on proving canonicity for it), but to do so
may require further generalization of the Kan operations.
The treatment of negation as a type identification given here is a
special case of the much more general concept of glueing introduced
by~\citet{cohen2016cubical}.

\section{Overview}

The most basic concept is that of a \emph{dimension name}, which may be thought
of as a formal variable ranging over an abstract interval. A \emph{dimension
context} is a finite set $\Psi$ of dimension names. The judgment $\wfdim{r}$
states that $r$ is a well-formed dimension term relative to $\Psi$, that is, $r$
is either $0$, $1$, or a dimension name $x\in\Psi$. If $\wfdim[\Psi,x]{r'}$ and
$\wfdim{r}$, then $\wfdim{\dsubst{r'}{r}{x}}$ is the result of replacing
occurrences of $x$ in $r'$ by $r$ (that is, $r$ when $r' = x$ and $r'$
otherwise).
A \emph{dimension substitution}
$\msubsts{\Psi'}{\psi}{\Psi}$ is a mapping assigning to each dimension name in
$\Psi$ a dimension term well-formed in $\Psi'$.  Dimension substitutions provide
a structural interpretation of dimension names as variables ranging over an
abstract interval.  A length-$n$ dimension context may then be seen as an
abstract $n$-dimensional cube thought of as an $n$-fold product of intervals.

The collection of terms includes the standard forms of expression (including
type expressions) of type theory, as well as expressions containing dimension
terms, chiefly the Kan operations. The judgment $\wftm[\Psi]{E}$ means that $E$
is closed with respect to term variables and its dimension subterms are
well-formed in $\Psi$.  If $\wftm[\Psi]{E}$ and $\msubsts{\Psi'}{\psi}{\Psi}$,
then $\wftm[\Psi']{\td{E}{\psi}}$ is the term resulting from substituting
dimension names $x\in\Psi$ by $\psi(x)$ in $E$.%
\footnote{When dimension binders occur in $E$, dimension substitution is defined
only up to renaming of bound dimension names, and is understood to avoid capture.}
We call $\td{E}{\psi}$ a \emph{cubical aspect} of $E$ because it geometrically
represents a $\Psi'$-cube obtained by performing face, degeneracy, and diagonal
cubical operations to the $\Psi$-cube $E$.
Terms are equipped with a deterministic operational semantics given by the
judgments $\isval{E}$, stating that $E$ is a value, and $E\evals V$, stating
that $E$ evaluates to $V$.  The evaluation relation is oblivious to dimension
and hence is not indexed by $\Psi$. We use various capitalized meta-variables,
including $M$, $N$, $A$, $B$, to stand for terms.

The main judgment forms of computational higher type theory are
the \emph{exact equalities}:
\begin{enumerate}
\item Exact type equality: $\ceqtype[\Psi]{A}{B}$.
\item Exact term equality: $\ceqtm[\Psi]{M}{N}{A}$.
\end{enumerate}
Exact equality is extensional, rather than intensional.  Two special
forms of judgment are derived from these:
\begin{enumerate}
\item Type formation: $\cwftype[\Psi]{A}$, which means
  $\ceqtype[\Psi]{A}{A}$.
\item Term formation: $\coftype[\Psi]{M}{A}$, which means
  $\ceqtm[\Psi]{M}{M}{A}$.
\end{enumerate}

The type formation judgment states that $A$ is a type of dimension
specified by $\Psi$.  When $\Psi$ is empty, $A$ is a type in the
familiar sense.  When $\Psi$ is non-empty, say $\Psi=(\Psi',x)$ with
$x\notin\Psi'$, then $A$ is a \emph{type line} connecting end points
$\dsubst{A}{0}{x}$ and $\dsubst{A}{1}{x}$, both of which are types of dimension
$\Psi'$.  Similarly, the membership judgment states that $M$ is a \emph{line in
the type line} $A$.  When $\Psi$ is empty, this means that $M$ inhabits $A$ in
the usual sense.  When $\Psi$ is $(\Psi',x)$ and $A$ is \emph{homogeneous} in
$x$ (in that $x$ does not occur in $A$), the membership judgment means that $M$
is an ordinary $x$-line in $A$ between $\dsubst{M}{0}{x}$ and
$\dsubst{M}{1}{x}$.  When $A$ depends on $x$, $A$ is a \emph{heterogeneous type
line} underwriting the \emph{coercion}, or \emph{transport}, of members of
$\dsubst{A}{0}{x}$ to $\dsubst{A}{1}{x}$; a member $M$ of $A$ can be thought
of as a homogeneous $x$-line in $\dsubst{A}{1}{x}$ between $\dsubst{M}{1}{x}$
and the coercion in $x$ along $A$ of $\dsubst{M}{0}{x}$.

The meanings of these judgments are given in terms of several subsidiary
concepts.  First, we designate certain values of each dimension as naming
partial equivalence relations (PERs) on values of the same dimension, and
specify that two such values are related when they name the same PER.  We write
$\eqper[\Psi]{A}{B}$ to mean that values $A$ and $B$ name the same PER on
values, and write $\vinper[\Psi]{M}{N}{A}$ to mean that $M$ and $N$ are
equivalent values according to the PER named by the value $A$.

Second, we define when $A$ and $B$ name \emph{equal pretypes}, and, for a
pretype $A$, when $M$ and $N$ are \emph{equal terms in $A$}.  These concepts
extend PER equality and membership from values to closed terms by
evaluation, so that equal pretypes evaluate to values naming the same PER, and
equal members of a pretype evaluate to equivalent values in the corresponding
PER.  Their precise definition is somewhat subtle and includes conditions
ensuring that pretypes and members of pretypes have coherent cubical aspects in
a sense to be made precise in \cref{sec:meanings} below. In particular, all the
cubical aspects of a pretype are themselves pretypes---if $\cpretype{A}$, then
for all $\msubst{\psi}$, $\cpretype{\td{A}{\psi}}$.  Similarly, if
$\cpretype{A}$ and $\msubst{\psi}$ and $\coftype{M}{A}$, then
$\coftype[\Psi']{\td{M}{\psi}}{\td{A}{\psi}}$.

%Membership and equality are cumulative in that every $\Psi$ type and
%$\Psi$ member is implicitly also a $\Psi'$ type or member when
%$\Psi\subseteq\Psi'$.

A pretype $A$ is \emph{cubical} when $\vinper{M}{N}{A_0}$ implies
$\ceqtm{M}{N}{\td{A}{\psi}}$ for all $\msubst{\psi}$, where
$\td{A}{\psi}\evals A_0$.  This condition states that the values of a
pretype must be full members of that pretype, and hence have
coherent aspects.  A pretype $A$ is
\emph{uniformly Kan} whenever it supports heterogeneous coercion as
described above, and is, moreover, closed under a homogeneous
composition operation.  Finally, a pretype is a \emph{type} if it is both
cubical and Kan.  The exact formulation of the Kan condition is
particular to our setting, but is broadly in line with the
formulations given by~\citet{bch,licata2014cubical,cohen2016cubical}.  The
general idea is to ensure that type lines can be operationalized as
coercions, and that there are sufficiently many lines to support their
interpretation as identifications.

With these definitions in hand, we then define an illustrative
collection of simple types.  Specifically, we consider two higher
inductive types, namely the booleans ($\bool$) and the circle ($\C$),
negation ($\notb{x}$) as a line between $\bool$ and itself given by
negation, and cartesian products ($\prd{A}{B}$) and function spaces
($\arr{A}{B}$).  We show that
the standard rules for these types are true under the definitions of
these types.  In particular, their terms satisfy the characteristic
exact equalities associated with these types.

The basic judgments of the type theory are extended to open terms
(those with free term variables) by \emph{judgments-in-context}, or
\emph{sequents}, of the form $\eqtm[\Psi]{\G}{M}{N}{A}$, where $\G$ is
a finite sequence $\oft{a_1}{A_1},\dots,\oft{a_n}{A_n}$ with
$n\geq 0$.  The meaning of such a judgment is given by
\emph{functionality}, which states that equal terms of the types $A_i$
are sent to equal terms of type $A$.  This is enough to
ensure that standard equational reasoning principles are valid for
open terms, in particular that equals may be silently replaced by
equals.

It is important to note that the entailment expressed by a sequent is
\emph{not} a derivability judgment in the sense of formal type theory,
which is concerned with \emph{formal proofs} of propositions (viewed
as elements of types), but rather expresses an intuitionistically true
entailment witnessed by a computation mapping evidence for the
hypotheses into evidence for the conclusion.  There is therefore no
reason to expect, much less demand, that such entailments are
decidable; rather, they are expressions of truths that must be
witnessed with computable evidence.

%%% Local Variables:
%%% mode: latex
%%% TeX-master: "heo.tex"
%%% End:

%% file: opsem.tex
\section{Programming language}
\label{sec:opsem}

The programming language itself has two sorts, dimensions and terms,
and binders for both sorts.  Terms are an ordinary untyped lambda calculus
with constructors; dimensions are either dimension
constants ($0$ or $1$) or one of countably many dimension names
($x,y,\dots$) behaving like nominal constants~\citep{pittsnominal}.
Dimension terms
occur at specific positions in some terms; for example, $\lp{r}$ is a term
for any dimension term $r$.  The operational semantics is
defined on terms that are closed with respect to term variables but may
contain free dimension names.

Dimension names represent generic elements of an abstract interval whose end
points are notated $0$ and $1$. While one may sensibly substitute any dimension
term for a dimension name, terms are \emph{not} to be understood solely in terms
of their dimensionally-closed instances (namely, their end points). Rather, a
term's dependence on dimension names is to be understood generically;
geometrically, one might imagine additional unnamed points in the interior of
the abstract interval.

The language features two terms that are specific to higher type
theory.  The first, called \emph{coercion}, has the form
$\coegeneric{x.A}$, where $x.A$ is a type line, $r$ is the
\emph{starting} dimension and $r'$ is the \emph{ending} dimension.
Coercion transports a term $M$ from $\dsubst{A}{r}{x}$ to
$\dsubst{A}{r'}{x}$ using the type line $x.A$ as a guide.
Coercion from $r$ to itself has no effect, up to exact
equality.  Coercion from $0$ to $1$ or vice versa is transport, which
applies one direction of the equivalence induced by the type line.
Coercion from $0$ or $1$ to a dimension name $y$ creates a $y$-line in 
$\dsubst{A}{y}{x}$, and coercion from $y$ to $0$ or $1$ 
yields a line between one end point of the input $y$-line and
the transport of the opposite end point.
Finally, coercion from one dimension name to another reorients the
line from one dimension to another.

The second, called \emph{homogeneous Kan composition}, has the form
$\hcomgeneric{r_1}{A}$, where $r_1$ is the \emph{extent}, $r$ is the
\emph{starting} dimension, and $r'$ is the \emph{ending} dimension.
The term $M$ is called the \emph{cap}, and the terms $N^0$ and $N^1$
form the \emph{tube} of the composition.%
\footnote{If $r_1=x$ and $x$ occurs in $N^\e$, then the tube sides are actually
$\dsubst{N^\e}{\e}{x}$ for $\e=0,1$, representing that $x=\e$ on the $N^\e$ side
of the composition problem. In the present description, we assume that $x$ does
not occur in $N^\e$; the precise typing rules for Kan composition are given in
\cref{def:kan}.}
This composition is well-typed when the starting side of each $N^\e$ coincides
(up to exact equality) with the $\e$ side of the cap. When $r_1$ is a dimension
name $x$, the composition results in an $x$-line, called the \emph{composite},
whose $\e$ sides coincide with the ending sides of each $N^\e$. The composite is
easily visualized when $r=0$ and $r'=1$:

\[
\renewcommand{\objectstyle}{\scriptstyle}
\renewcommand{\labelstyle}{\textstyle}
\xymatrix@=0.75em{
  {} \ar[d] \ar[r] & x \\ y
}
\qquad
\xymatrix@C=8em@R=3em{%\ar @{} [dr] |{=}
  \bullet \ar[d]_{N^0} \ar[r]^{M} &
  \bullet \ar[d]^{N^1} \\
  \dsubst{N^0}{1}{y} \ar@{-->}[r]_{\hcomsym^x_A(0\rightsquigarrow 1)} &
  \dsubst{N^1}{1}{y}}
\]
The case of $r=1$ and $r'=0$ is symmetric, swapping the roles of the cap and the
composite.

When the starting dimension is $r=0$ (or, analogously, $r=1$) and the ending
dimension is $r'=y$, where $y$ does not occur in $M$, the Kan composition yields
the interior of the $x,y$-square depicted above, called the \emph{filler}. One
may think of this composition as sweeping out that square by sliding the cap from
$y=0$ to any point in the $y$ dimension, much in the manner of opening a window
shade. The filler is simultaneously an $x$-line identifying the two tube sides
with each other, and a $y$-line identifying the cap with the composite.  

When $r=y$ and $r'$ is $0$ or $1$, the composition may be visualized as closing
a window shade, starting in the ``middle'' and heading towards the roll at one
end or the other.
When both $r$ and $r'$ are dimension names, the result is harder to visualize,
and is best understood formally, as is also the case where $r=0$ and $r'=y$ but
$y$ does occur in $M$.

Finally, there are two cases in which the composition scenario trivializes.
When $r=r'=0$ or $r=r'=1$, the composition is the cap itself, intuitively
because the window shade does not move from its starting position at the cap.
When $r_1$ is $0$ (or $1$), rather than a dimension name, the composition is
simply $\dsubst{N^0}{r'}{y}$ (or $\dsubst{N^1}{r'}{y}$), because the composition
has no extent beyond that end point. These two cases are important because they
ensure, respectively, that the $y$ and $x$ end points of the $x,y$-filler are as
depicted above.

\subsection{Terms}

\[\begin{aligned}
M &:=
% types
\arr{A}{B} \mid
\prd{A}{B} \mid
\bool \mid
\notb{r} \mid
\C \\&\mid
% terms
\lam{a}{M} \mid
\app{M}{N} \mid
\pair{M}{N} \mid
\fst{M} \mid
\snd{M} \\&\mid
\true \mid
\false \mid
\ifb{A}{M}{N_1}{N_2} \mid
\notel{r}{M} \\&\mid
\base \mid
\lp{r} \mid
\Celim{A}{M}{N_1}{x.N_2} \\&\mid
\coegeneric{x.A} \mid
\hcomgeneric{r_1}{A}
\end{aligned}\]

We use capital letters like $M$, $N$, and $A$ to denote terms, $r$,
$r'$, $r_1$ to denote dimension terms, $x$ to denote dimension names, $\e$
to denote dimension constants ($0$ or $1$), and $\eb$ to denote the
opposite dimension constant of $\e$.  We write $x.-$ for dimension
binders, $a.-$ for term binders, and $\fd{M}$ for the set of dimension
names free in $M$.  Dimension substitution
$\dsubst{M}{r}{x}$ and term substitution $\subst{M}{N}{a}$ are defined
in the usual way.  We write $\notf{M}$ as shorthand for the term
$\ifb{\bool}{M}{\false}{\true}$.

\begin{remark}
In a follow-on paper, we will generalize $\hcomsym$ to allow for
$n\geq 1$ pairs of tubes, as follows:
\[
\hcomsym_{A}^{r_1,\dots,r_n}(r\rightsquigarrow r',M;
y.N^0_1,y.N^1_1,\dots,y.N^0_n,y.N^1_n)
\]
This $\hcomsym$ operation (and a corresponding generalization of the Kan
conditions in \cref{sec:meanings}) is needed to define identification types.
\end{remark}

\subsection{Operational semantics}

The following describes a deterministic weak head reduction evaluation
strategy for closed terms in the form of a transition system with two
judgments:
\begin{enumerate}
\item $\isval{E}$, stating that $E$ is a \emph{value}, or
  \emph{canonical form}.
\item $E\steps E'$, stating that $E$ takes \emph{one step of
    evaluation} to $E'$.
\end{enumerate}
These judgments are defined so that if $\isval{E}$, then
$E\not\steps$, but the converse need not be the case.  As usual, we
write $E\steps^* E'$ to mean that $E$ transitions to $E'$ in zero
or more steps.  We say $E$ evaluates to $V$, written $E \evals V$, when
$E\steps^* V$ and $\isval{V}$.

Most of the evaluation rules are standard, and evaluate only principal arguments
of elimination forms. The principal arguments of $\hcomsym$ and $\coesym$ are
their type subscripts, whose head constructors determine how those terms
evaluate. In the present system, the only non-canonical type subscript of
interest is $\notb{\e}$.

Determinacy is a strong condition that implies that a term has at most
one value.
\begin{lemma}[Determinacy]
  If $M\steps M_1$ and $M\steps M_2$, then $M_1 = M_2$.
\end{lemma}

Stability states that evaluation does not introduce any new dimension names.
\begin{lemma}[Stability]
  If $M\steps M'$, then $\fd{M'}\subseteq\fd{M}$.
\end{lemma}

\paragraph{Types}

\[
\infer
  {\notb{\e} \steps \bool}
  {}
\]
\[
\infer
  {\isval{\arr{A}{B}}}
  {}
\qquad
\infer
  {\isval{\prd{A}{B}}}
  {}
\qquad
\infer
  {\isval{\bool}}
  {}
\qquad
\infer
  {\isval{\notb{x}}}
  {}
\qquad
\infer
  {\isval{\C}}
  {}
\]

\paragraph{Hcom/coe}

\[
\infer
  {\coe{x.A}{r}{r'}{M} \steps \coe{x.A'}{r}{r'}{M}}
  {A\steps A'}
\]

\[
\infer
  {\hcomgeneric{r_1}{A} \steps \hcomgeneric{r_1}{A'}}
  {A\steps A'}
\]

\paragraph{Function types}

\[
\infer
  {\app{M}{N} \steps \app{M'}{N}}
  {M \steps M'}
\qquad
\infer
  {\app{\lam{a}{M}}{N} \steps \subst{M}{N}{a}}
  {}
\qquad
\infer
  {\isval{\lam{a}{M}}}
  {}
\]

\[
\infer
  {\hcomgeneric{r_1}{\arr{A}{B}} \steps
   \lam{a}{\hcom{r_1}{B}{r}{r'}{\app{M}{a}}{y.\app{N^0}{a},y.\app{N^1}{a}}}}
  {}
\]

\[
\infer
  {\coe{x.\arr{A}{B}}{r}{r'}{M} \steps
   \lam{a}{\coe{x.B}{r}{r'}{\app{M}{\coe{x.A}{r'}{r}{a}}}}}
  {}
\]

\paragraph{Product types}

\[
\infer
  {\fst{M} \steps \fst{M'}}
  {M \steps M'}
\qquad
\infer
  {\snd{M} \steps \snd{M'}}
  {M \steps M'}
\qquad
\infer
  {\isval{\pair{M}{N}}}
  {}
\]

\[
\infer
  {\fst{\pair{M}{N}} \steps M}
  {}
\qquad
\infer
  {\snd{\pair{M}{N}} \steps N}
  {}
\]

\[
\infer
  {\begin{array}{c}
   \hcomgeneric{r_1}{\prd{A}{B}} \\
   \steps \\
   \pair{\hcom{r_1}{A}{r}{r'}{\fst{M}}{y.\fst{N^0},y.\fst{N^1}}}{\hcom{r_1}{B}{r}{r'}{\snd{M}}{y.\snd{N^0},y.\snd{N^1}}}
   \end{array}}
  {}
\]~

\[
\infer
  {\coe{x.\prd{A}{B}}{r}{r'}{M} \steps
   \pair{\coe{x.A}{r}{r'}{\fst{M}}}
        {\coe{x.B}{r}{r'}{\snd{M}}}}
  {}
\]

\paragraph{Booleans}

\[
\infer
  {\hcomgeneric{\e}{\bool} \steps \dsubst{N^\e}{r'}{y}}
  {}
\qquad
\infer
  {\hcomgeneric{x}{\bool} \steps M}
  {r = r'}
\]

\[
\infer
  {\isval{\true}}
  {}
\qquad
\infer
  {\isval{\false}}
  {}
\qquad
\infer
  {\isval{\hcomgeneric{x}{\bool}}}
  {r \neq r'}
\]

\[
\infer
  {\ifb{A}{M}{T}{F} \steps \ifb{A}{M'}{T}{F}}
  {M \steps M'}
\quad
\infer
  {\ifb{A}{\true}{T}{F} \steps T}
  {}
\quad
\infer
  {\ifb{A}{\false}{T}{F} \steps F}
  {}
\]

\[
\infer
  {\begin{array}{c}
   \ifb{A}{\hcomgeneric{x}{\bool}}{T}{F} \\
   \steps \\
   \hcom{x}{A}{r}{r'}{\ifb{A}{M}{T}{F}}
    {y.\ifb{A}{N^0}{T}{F},y.\ifb{A}{N^1}{T}{F}}
   \end{array}}
  {r \neq r'}
\]~

\[
\infer
  {\coe{x.\bool}{r}{r'}{M} \steps M}
  {}
\]

\paragraph{Circle}

\[
\infer
  {\hcomgeneric{\e}{\C} \steps \dsubst{N^\e}{r'}{y}}
  {}
\qquad
\infer
  {\hcomgeneric{x}{\C} \steps M}
  {r = r'}
\]

\[
\infer
  {\lp{\e} \steps \base}
  {}
\qquad
\infer
  {\isval{\base}}
  {}
\qquad
\infer
  {\isval{\lp{x}}}
  {}
\qquad
\infer
  {\isval{\hcomgeneric{x}{\C}}}
  {r \neq r'}
\]

\[
\infer
  {\Celim{A}{M}{P}{z.L} \steps \Celim{A}{M'}{P}{z.L}}
  {M \steps M'}
\]

\[
\infer
  {\Celim{A}{\base}{P}{z.L} \steps P}
  {}
\qquad
\infer
  {\Celim{A}{\lp{w}}{P}{z.L} \steps \dsubst{L}{w}{z}}
  {}
\]

\[
\infer
  {\begin{array}{c}
   \Celim{A}{\hcomgeneric{x}{\C}}{P}{z.L} \\
   \steps \\
   \hcom{x}{A}{r}{r'}{\Celim{A}{M}{P}{z.L}}
    {y.\Celim{A}{N^0}{P}{z.L},y.\Celim{A}{N^1}{P}{z.L}}
   \end{array}}
  {r \neq r'}
\]~

\[
\infer
  {\coe{x.\C}{r}{r'}{M} \steps M}
  {}
\]

\paragraph{Not}

\[
\infer
  {\isval{\notel{x}{M}}}
  {}
\qquad
\infer
  {\notel{0}{M} \steps \notf{M}}
  {}
\qquad
\infer
  {\notel{1}{M} \steps M}
  {}
\]

\[
\infer
  {\coe{x.\notb{x}}{\e}{\eb}{M} \steps \notf{M}}
  {}
\qquad
\infer
  {\coe{x.\notb{x}}{\e}{\e}{M} \steps M}
  {}
\]

\[
\infer
  {\coe{x.\notb{x}}{0}{x}{M} \steps \notel{x}{\notf{M}}}
  {}
\qquad
\infer
  {\coe{x.\notb{x}}{1}{x}{M} \steps \notel{x}{M}}
  {}
\]

\[
\infer
  {\coe{x.\notb{x}}{x}{r}{M} \steps \coe{x.\notb{x}}{x}{r}{M'}}
  {M\steps M'}
\qquad
\infer
  {\coe{x.\notb{x}}{x}{r}{\notel{x}{M}} \steps \notel{r}{M}}
  {}
\qquad
\infer
  {\coe{x.\notb{y}}{r}{r'}{M} \steps M}
  {x \neq y}
\]

\[
\infer
  {\begin{array}{c}
   \hcomgeneric{r_1}{\notb{w}} \\
   \steps \\
   \notel{w}{\hcom{r_1}{\bool}{r}{r'}{\coe{x.\notb{x}}{w}{1}{M}}{y.\coe{x.\notb{x}}{w}{1}{N^0},y.\coe{x.\notb{x}}{w}{1}{N^1}}}
   \end{array}}
  {}
\]~

%%% Local Variables:
%%% mode: latex
%%% TeX-master: "heo"
%%% End:

%% file: meanings.tex
\section{Meaning explanations}
\label{sec:meanings}

\begin{definition}\label{def:wftm}
We say $\wftm{M}$ when $M$ is a term with no free term variables, and
$\fd{M}\subseteq\Psi$.
\end{definition}

\begin{remark}\label{def:wfval}
  We write $\wfval{M}$ when $\wftm{M}$ and $\isval{M}$. Being a value
  does not depend on the choice of $\Psi$, so whenever
  $\wfval{M}$ and $\fd{M}\subseteq\Psi'$, we also have
  $\wfval[\Psi']{M}$.
\end{remark}

\begin{definition}
A total dimension substitution $\msubst{\psi}$ assigns to each dimension name in
$\Psi$ either $0$, $1$, or a dimension name in $\Psi'$. It follows that if
$\wftm{M}$ then $\wftm[\Psi']{\td{M}{\psi}}$.
\end{definition}

Some $\wfval{A}$ are taken to name types; to these, we associate partial
equivalence relations over values in $\Psi$. PERs are a convenient way of
describing sets equipped with an equivalence relation; elements of the
corresponding set are the values that are related to themselves.

The presuppositions of a judgment are the facts that must be true before one
can even sensibly state that judgment. For example, in \cref{def:eqper} below
we define a judgment on $A$ and $B$, presupposing that $A$ and $B$ are
associated to PERs, which asserts that those PERs are equal; this condition is
not meaningful unless these PERs exist.

\begin{definition}\label{def:per}
We say $\per{A}$, presupposing $\wfval{A}$, when we have associated to $\Psi$
and $A$ a symmetric and transitive relation $\vinper--{A}$ on terms $M$ such
that $\wfval{M}$.
\end{definition}

\begin{remark}\label{def:inper}
We write $\inper{M}{N}{A}$ when $\per{A}$, $M\evals M_0$, $N\evals N_0$, and
$\vinper{M_0}{N_0}{A}$.
\end{remark}

\begin{definition}\label{def:eqper}
We say $\eqper{A}{B}$, presupposing $\per{A}$ and $\per{B}$, when
for all $\wfval{M}$ and $\wfval{N}$,
$\vinper{M}{N}{A}$ if and only if $\vinper{M}{N}{B}$.
\end{definition}

\begin{remark}
The above definition of $\eqper{A}{B}$ yields an extensional notion of (pre)type
equality; it is also possible to define $\eqper{A}{B}$ inductively on the
structure of $A$ and $B$ in order to obtain an intensional (pre)type equality.
In either case it is essential that if $\per{A}$ then $\eqper{A}{A}$.
\end{remark}

Approximately, a term $A$ is a pretype in $\Psi$ when $\td{A}{\psi}$
evaluates to the name of a PER in $\Psi'$ for every $\msubst{\psi}$.
A term $M$ is an
element of a (pre)type when every $\td{M}{\psi}$ evaluates to an
element of the corresponding PER.  We also demand that pretypes and
their elements have \emph{coherent aspects}, a technical condition
implying that dimension substitutions can be taken simultaneously or
sequentially, before or after evaluating a term, without affecting the
outcome, up to PER equality.
(In our postfix notation for dimension substitutions, $\td{A}{\psi_1\psi_2}$
means $\td{(\td{A}{\psi_1})}{\psi_2}$.)

\begin{definition}\label{def:ceqpretype}
We say $\ceqpretype{A}{B}$, presupposing $\wftm{A}$ and $\wftm{B}$, when
for any $\msubsts{\Psi_1}{\psi_1}{\Psi}$ and
$\msubsts{\Psi_2}{\psi_2}{\Psi_1}$,
\begin{enumerate}
\item
$\td{A}{\psi_1}\evals A_1$, 
$\td{A_1}{\psi_2}\evals A_2$, 
$\td{A}{\psi_1\psi_2}\evals A_{12}$, 
$\per[\Psi_2]{A_2}$, $\per[\Psi_2]{A_{12}}$, 
\item
$\td{B}{\psi_1}\evals B_1$, 
$\td{B_1}{\psi_2}\evals B_2$, 
$\td{B}{\psi_1\psi_2}\evals B_{12}$, 
$\per[\Psi_2]{B_2}$, $\per[\Psi_2]{B_{12}}$, and
\item
$\eqperfour[\Psi_2]{A_2}{A_{12}}{B_2}{B_{12}}$.
\end{enumerate}
\end{definition}

\begin{remark}\label{def:cpretype}
We write $\cpretype{A}$ when $\ceqpretype{A}{A}$.
\end{remark}

\begin{definition}\label{def:ceqtm}
We say $\ceqtm{M}{N}{B}$, presupposing $\cpretype{B}$,
$\wftm{M}$, and $\wftm{N}$, when for any
$\msubsts{\Psi_1}{\psi_1}{\Psi}$ and $\msubsts{\Psi_2}{\psi_2}{\Psi_1}$,
\begin{enumerate}
\item
$\td{M}{\psi_1}\evals M_1$, 
$\td{M_1}{\psi_2}\evals M_2$, 
$\td{M}{\psi_1\psi_2}\evals M_{12}$, 
\item
$\td{N}{\psi_1}\evals N_1$, 
$\td{N_1}{\psi_2}\evals N_2$, 
$\td{N}{\psi_1\psi_2}\evals N_{12}$, and
\item
$\vinperfour[\Psi_2]{M_2}{M_{12}}{N_2}{N_{12}}{B_{12}}$,
where $\td{B}{\psi_1\psi_2}\evals B_{12}$.
\end{enumerate}
\end{definition}

A valid term context $\G$ in $\Psi$ is either the empty context $\cdot$, or a
sequence $\oft{a_1}{A_1},\dots,\oft{a_n}{A_n}$ of distinct term variables $a_i$
each paired with a pretype $A_i$ in $\Psi$.

\begin{definition}\label{def:wfctx}
We say $\wfctx{\cdot}$ always, and $\wfctx{\G,\oft aA}$ whenever
$\wfctx{\G}$ and $\cpretype{A}$.
\end{definition}

\begin{definition}\label{def:eqtm}
We say $\eqtm{\G}{M}{M'}{B}$, presupposing $\wfctx{\G}$ and $\cpretype{B}$, when
\begin{enumerate}
\item %$\eqtm{\cdot}{M}{M'}{B}$
$\G = \cdot$ and $\ceqtm{M}{M'}{B}$, or
\item %$\eqtm{\G,\oft aA}{M}{M'}{B}$
$\G = (\G',\oft aA)$ and for any $\msubst{\psi}$ and
$\eqtm[\Psi']{\td{\G'}{\psi}}{N}{N'}{\td{A}{\psi}}$, that
$\eqtm[\Psi']{\td{\G'}{\psi}}{\subst{\td{M}{\psi}}{N}{a}}{\subst{\td{M'}{\psi}}{N'}{a}}{\td{B}{\psi}}$.
\end{enumerate}
\end{definition}

In the notation $\eqtm{\G}{M}{M'}{B}$, one should read the dimension index
$[\Psi]$ as extending across the entire sequent, as it specifies the starting
dimension at which to consider $\G$.
To make sense of the second clause of \cref{def:eqtm},
notice that for any $\msubst{\psi}$,
if $\cpretype{A}$, then $\cpretype[\Psi']{\td{A}{\psi}}$;
and if $\wfctx{\G}$, then $\wfctx[\Psi']{\td{\G}{\psi}}$,
where $\td{\G}{\psi}$ applies $\psi$ to every pretype in $\G$.

\begin{remark}\label{def:oftype}
We write $\coftype{M}{A}$ when $\cpretype{A}$ and $\ceqtm{M}{M}{A}$, and
$\oftype{\G}{M}{A}$ when $\cpretype{A}$ and $\eqtm{\G}{M}{M}{A}$.
\end{remark}

\begin{remark}
The equality judgments $\ceqtm{M}{N}{A}$ and $\eqtm{\G}{M}{N}{A}$ are symmetric
and transitive. Therefore,
if $\ceqtm{M}{N}{A}$ then $\coftype{M}{A}$ and $\coftype{N}{A}$, and
if $\eqtm{\G}{M}{N}{A}$ then $\oftype{\G}{M}{A}$ and $\oftype{\G}{N}{A}$.
\end{remark}

If no terms in our programming language contained dimension subterms, then for
all $M$, we would have $M=\td{M}{\psi}$. The above meaning explanations would
therefore collapse into:
$\cpretype{A}$ whenever $A\evals A_0$ and $\per{A_0}$;
and $\ceqtm{M}{N}{A}$ whenever $\inper{M}{N}{A_0}$ where $A\evals A_0$.
These are precisely the ordinary meaning explanations for computational type
theory.

Finally, a type is a pretype which is both \emph{cubical} (meaning that its PERs
are functorially indexed by the cube category) and \emph{Kan} (meaning that all
its instances validate the $\hcomsym$ and $\coesym$ rules).

\begin{definition}\label{def:cubical}
We say $\cpretype{A}$ is \emph{cubical} if for any $\msubst{\psi}$ and
$\vinper[\Psi']{M}{N}{A_0}$ (where $\td{A}{\psi}\evals A_0$) then
$\ceqtm[\Psi']{M}{N}{\td{A}{\psi}}$.
\end{definition}

\begin{definition}\label{def:kan}
We say $\cpretype{A}$ is \emph{Kan} if the following four conditions hold:
\begin{enumerate}
\item % hcom exists
For any $\msubsts{(\Psi',x)}{\psi}{\Psi}$, if
\begin{enumerate}
\item \ceqtm[\Psi',x]{M}{O}{\td{A}{\psi}},
\item \ceqtm[\Psi',y]{\dsubst{N^\e}{\e}{x}}{\dsubst{P^\e}{\e}{x}}{\dsubst{\td{A}{\psi}}{\e}{x}} for
$\e=0,1$, and
\item \ceqtm[\Psi']{\dsubst{\dsubst{N^\e}{r}{y}}{\e}{x}}{\dsubst{M}{\e}{x}}{\dsubst{\td{A}{\psi}}{\e}{x}} for $\e=0,1$,
\end{enumerate}
then $\ceqtm[\Psi',x]{\hcomgeneric{x}{\td{A}{\psi}}}{\hcom{x}{\td{A}{\psi}}{r}{r'}{O}{y.P^0,y.P^1}}{\td{A}{\psi}}$.

\item % hcom when r = r'
For any $\msubsts{(\Psi',x)}{\psi}{\Psi}$, if
\begin{enumerate}
\item \coftype[\Psi',x]{M}{\td{A}{\psi}},
\item \coftype[\Psi',y]{\dsubst{N^\e}{\e}{x}}{\dsubst{\td{A}{\psi}}{\e}{x}} for
$\e=0,1$, and
\item \ceqtm[\Psi']{\dsubst{\dsubst{N^\e}{r}{y}}{\e}{x}}{\dsubst{M}{\e}{x}}{\dsubst{\td{A}{\psi}}{\e}{x}} for $\e=0,1$,
\end{enumerate}
then $\ceqtm[\Psi',x]{\hcom{x}{\td{A}{\psi}}{r}{r}{M}{y.N^0,y.N^1}}{M}{\td{A}{\psi}}$.

\item % hcom when r_1 = \e
For any $\msubst{\psi}$, if
\begin{enumerate}
\item \coftype[\Psi']{M}{\td{A}{\psi}},
\item \coftype[\Psi',y]{N^\e}{\td{A}{\psi}}, and
\item
\ceqtm[\Psi']{\dsubst{N^\e}{r}{y}}{M}{\td{A}{\psi}},
\end{enumerate}
then $\ceqtm[\Psi']{\hcomgeneric{\e}{\td{A}{\psi}}}{\dsubst{N^\e}{r'}{y}}{\td{A}{\psi}}$.

\item % coe
For any $\msubsts{(\Psi',x)}{\psi}{\Psi}$, if
$\ceqtm[\Psi']{M}{N}{\dsubst{\td{A}{\psi}}{r}{x}}$, then
$\ceqtm[\Psi']{{\coe{x.\td{A}{\psi}}{r}{r'}{M}}}{{\coe{x.\td{A}{\psi}}{r}{r'}{N}}}{\dsubst{\td{A}{\psi}}{r'}{x}}$.
\end{enumerate}
\end{definition}

\begin{remark}\label{def:tubesubst}
  We always substitute $\e$ for $x$ in the tube face premises of each
  of these conditions, reflecting that $x=0$ and $x=1$ in the
  respective side of the tube.
\end{remark}

\begin{definition}\label{def:cwftype}
We say $\cwftype{A}$, presupposing $\cpretype{A}$, if $A$ is cubical and Kan.
\end{definition}

\subsection{Basic lemmas}

We prove some basic results about our core judgments before proceeding.

\begin{lemma}
For any $\msubsts{\Psi'}{\psi}{\Psi}$,
\begin{enumerate}
\item if $\cwftype{A}$ then $\cwftype[\Psi']{\td{A}{\psi}}$;
\item if $\ceqtm{M}{N}{A}$ then
$\ceqtm[\Psi']{\td{M}{\psi}}{\td{N}{\psi}}{\td{A}{\psi}}$; and
\item if $\eqtm{\G}{M}{N}{A}$, then
$\eqtm[\Psi']{\td{\G}{\psi}}{\td{M}{\psi}}{\td{N}{\psi}}{\td{A}{\psi}}$.
\end{enumerate}
\end{lemma}
\begin{proof}
We have already observed that if $\cpretype{A}$ then
$\cpretype[\Psi']{\td{A}{\psi}}$. A type is a Kan cubical pretype; that being a
type is closed under dimension substitution follows from the fact that the
Kan and cubical conditions are as well. Exact equality is closed under dimension
substitution also essentially because its definition quantifies over all
dimension substitutions.

The proof of the third proposition uses induction on the length
of $\G$. If $\G$ is empty, then the result follows immediately from
the previous one. If $\G = (\G',\oft aA)$, then we know
$\eqtm{\G',\oft{a}{A}}{M}{M'}{B}$, and want to show
$\eqtm[\Psi']{\td{\G}{\psi},\oft{a}{\td{A}{\psi}}}
{\td{M}{\psi}}{\td{M'}{\psi}}{\td{B}{\psi}}$.
Expanding definitions, this means we must show that for any
$\msubsts{\Psi''}{\psi'}{\Psi'}$ and
$\eqtm[\Psi'']{\td{\G}{\psi\psi'}}{N}{N'}{\td{A}{\psi\psi'}}$, we have
$\eqtm[\Psi'']{\td{\G}{\psi\psi'}} {\subst{\td{M}{\psi\psi'}}{N}{a}}
{\subst{\td{M'}{\psi\psi'}}{N'}{a}} {\td{B}{\psi\psi'}}$.
But this follows directly from the definition of
$\eqtm{\G',\oft{a}{A}}{M}{M'}{B}$.
\end{proof}

\begin{lemma}\label{lem:ceqpretype-ceqtm}
If $\ceqpretype{A}{B}$ and $\ceqtm{M}{N}{A}$ then $\ceqtm{M}{N}{B}$.
\end{lemma}
\begin{proof}
For any $\msubsts{\Psi_1}{\psi_1}{\Psi}$ and $\msubsts{\Psi_2}{\psi_2}{\Psi_1}$,
by the first hypothesis we have that
$\td{A}{\psi_1\psi_2}\evals A_{12}$, $\td{B}{\psi_1\psi_2}\evals B_{12}$, and
$\eqper[\Psi_2]{A_{12}}{B_{12}}$;
by the second hypothesis, we have that
$\vinperfour[\Psi_2]{M_2}{M_{12}}{N_2}{N_{12}}{A_{12}}$.
But this implies
$\vinperfour[\Psi_2]{M_2}{M_{12}}{N_2}{N_{12}}{B_{12}}$.
\end{proof}

The definition of the open judgment $\eqtm{\G}{M}{M'}{B}$ involves, for each
$\oft{a_i}{A_i}$ in $\G$, substituting into $M,M'$ a dimension substitution
$\msubst{\psi}$ and a pair of equal terms
$\eqtm[\Psi']{\td{\G}{\psi}}{N_i}{N_i'}{\td{A_i}{\psi}}$.
We prove this is equivalent to performing a single dimension substitution $\psi$
and a pair of simultaneous term substitutions $\g,\g'$ for all of $\G$,
whose components $\g(a_i),\g'(a_i)$ are equal in each $\td{A_i}{\psi}$.
We write $()$ for an empty simultaneous term substitution, and
$\tsext{\g}{N}{a}$ for the extension of a substitution $\g$ sending $a$ to $N$.
Then we say $\ceqtm{\g}{\g'}{\G}$ when $\g$ and $\g'$ are substitutions for all
of $\G$ whose components are equal: 

\begin{definition}\label{def:ceq-tsubst}
We say
\begin{enumerate}
\item
$\ceqtm{()}{()}{\cdot}$ always.
\item
$\ceqtm{\tsext{\g}{N}{a}}{\tsext{\g'}{N'}{a}}{(\G,\oft aA)}$,
presupposing $\wfctx{\G}$ and $\cpretype{A}$, when $\ceqtm{\g}{\g'}{\G}$ and
$\ceqtm{N}{N'}{A}$.
\end{enumerate}
\end{definition}

\begin{lemma}\label{lem:eqtm-tsubst}
  The open equation $\eqtm{\G}{M}{M'}{B}$ is true iff for any
  $\msubsts{\Psi'}{\psi}{\Psi}$ and any
  $\ceqtm[\Psi']{\g}{\g'}{\td{\G}{\psi}}$, we have
  $\ceqtm[\Psi']{\tsubst{\td{M}{\psi}}{\g}}{\tsubst{\td{M'}{\psi}}{\g'}}{\td{A}{\psi}}$.
\end{lemma}
\begin{proof}
  Simultaneously, by induction on the length of $\G$.  When $\G$ is
  empty, the result is immediate.  Otherwise let $\G=\G_1,\oft aA$.

  Suppose that $\eqtm{\G_1,\oft aA}{M}{M'}{B}$, and consider any
  $\msubsts{\Psi'}{\psi}{\Psi}$ and
  $\ceqtm[\Psi']{\g}{\g'}{\td{\G_1}{\psi},\oft{a}{\td{A}{\psi}}}$.
  Then $\g=\tsext{\g_1}{N}{a}$, $\g'=\tsext{\g_1'}{N'}{a}$,
  $\ceqtm[\Psi']{\g_1}{\g_1'}{\td{\G_1}{\psi}}$, and
  $\ceqtm[\Psi']{N}{N'}{\td{A}{\psi}}$.
  Because $N$ and $N'$ are closed, and all dimension substitution instances of 
  $\ceqtm[\Psi']{N}{N'}{\td{A}{\psi}}$ are true, by the reverse induction
  hypothesis we have $\eqtm[\Psi']{\td{\G_1}{\psi}}{N}{N'}{\td{A}{\psi}}$.
  By instantiating the definition of our hypothesis at $N,N'$ we get
  $\eqtm[\Psi']{\td{\G_1}{\psi}}{\subst{\td{M}{\psi}}{N}{a}}{\subst{\td{M'}{\psi}}{N'}{a}}{\td{B}{\psi}}$.
  But by the forward induction hypothesis, this gives us
  $\ceqtm[\Psi']{\tsubst{\subst{\td{M}{\psi}}{N}{a}}{\g_1}}{\tsubst{\subst{\td{M'}{\psi}}{N'}{a}}{\g_1'}}{\td{B}{\psi}}$,
  which is to say
  $\ceqtm[\Psi']{\tsubst{\td{M}{\psi}}{\g}}{\tsubst{\td{M'}{\psi}}{\g'}}{\td{B}{\psi}}$,
  as required.

  Conversely, suppose that for all $\msubsts{\Psi'}{\psi}{\Psi}$, if
  $\ceqtm[\Psi']{\g}{\g'}{\td{\G}{\psi}}$, then
  $\ceqtm[\Psi']{\tsubst{\td{M}{\psi}}{\g}}{\tsubst{\td{M'}{\psi}}{\g'}}{\td{B}{\psi}}$.
  To show $\eqtm[\Psi]{\G}{M}{M'}{B}$, suppose that
  $\msubsts{\Psi'}{\psi}{\Psi}$ and
  $\eqtm[\Psi']{\td{\G_1}{\psi}}{N}{N'}{\td{A}{\psi}}$; it suffices to
  show that
  $\eqtm[\Psi']{\td{\G_1}{\psi}}{\subst{\td{M}{\psi}}{N}{a}}{\subst{\td{M'}{\psi}}{N'}{a}}{\td{B}{\psi}}$.
  By the forward induction hypothesis we have that for all
  $\ceqtm[\Psi']{\g_1}{\g_1'}{\td{\G_1}{\psi}}$,
  $\ceqtm[\Psi']{\tsubst{N}{\g_1}}{\tsubst{N'}{\g_1'}}{\td{A}{\psi}}$.
  Let $\g=\g_1[\tsubst{N}{\g_1}/a]$ and $\g'=\g_1'[\tsubst{N'}{\g_1'}/a]$, so
  that $\ceqtm[\Psi']{\g}{\g'}{\td{\G}{\psi}}$. By assumption we have
  $\ceqtm[\Psi']{\tsubst{\td{M}{\psi}}{\g}}{\tsubst{\td{M'}{\psi}}{\g'}}{\td{B}{\psi}}$.
  But $\tsubst{\td{M}{\psi}}{\g}$ is
  $\tsubst{\subst{\td{M}{\psi}}{N}{a}}{\g_1}$ and
  $\tsubst{\td{M'}{\psi}}{\g'}$ is
  $\tsubst{\subst{\td{M'}{\psi}}{N'}{a}}{\g_1'}$,
  so by the reverse induction hypothesis we have
  $\eqtm[\Psi']{\td{\G_1}{\psi}}{\subst{\td{M}{\psi}}{N}{a}}{\subst{\td{M'}{\psi}}{N'}{a}}{\td{B}{\psi}}$.
\end{proof}

%%% Local Variables:
%%% mode: latex
%%% TeX-master: "heo"
%%% End:

%% file: types.tex
\section{Types}
\label{sec:types}

In this section, we will define various types by defining their PERs, verifying
they are pretypes, proving their introduction, elimination, and computation
rules, and then proving that they are cubical and Kan. A handful of lemmas will
be useful throughout this process:

\begin{lemma}[Head expansion]
If $\ceqtm{M'}{N}{A}$ and for all $\msubst{\psi}$,
$\td{M}{\psi} \steps^* \td{M'}{\psi}$, then $\ceqtm{M}{N}{A}$.
\end{lemma}
\begin{proof}
For any $\msubsts{\Psi_1}{\psi_1}{\Psi}$ and $\msubsts{\Psi_2}{\psi_2}{\Psi_1}$,
we know $\vinperfour[\Psi_2]{M'_2}{M'_{12}}{N_2}{N_{12}}{A_{12}}$. Therefore it
suffices to show $\vinper[\Psi_2]{M'_{12}}{M_{12}}{A_{12}}$ and
$\vinper[\Psi_2]{M'_2}{M_2}{A_{12}}$.
The former is true because
$\td{M}{\psi_1\psi_2} \steps^* \td{M'}{\psi_1\psi_2} \evals M'_{12}$ and
$\vinper[\Psi_2]{M'_{12}}{M'_{12}}{A_{12}}$. The latter is true because
$\td{M}{\psi_1} \steps^* \td{M'}{\psi_1} \evals M'_1$,
$\td{M'_1}{\psi_2} \evals M'_2$, and
$\vinper[\Psi_2]{M'_2}{M'_2}{A_{12}}$.
\end{proof}

A special case of this lemma is that if $\coftype{M'}{A}$ then $\coftype{M}{A}$.

\begin{lemma}\label{lem:coftype-ceqtm}
If $\coftype{M}{A}$, $\coftype{N}{A}$, and for all $\msubst{\psi}$,
$\inper[\Psi']{\td{M}{\psi}}{\td{N}{\psi}}{A'}$ where $\td{A}{\psi}\evals A'$,
then $\ceqtm{M}{N}{A}$.
\end{lemma}
\begin{proof}
For all $\msubsts{\Psi_1}{\psi_1}{\Psi}$ and $\msubsts{\Psi_2}{\psi_2}{\Psi_1}$,
by $\coftype{M}{A}$ we have $\td{M}{\psi_1}\evals M_1$ and
$\inper[\Psi_2]{\td{M_1}{\psi_2}}{\td{M}{\psi_1\psi_2}}{A_{12}}$, and
by $\coftype{N}{A}$ we have $\td{N}{\psi_1}\evals N_1$ and
$\inper[\Psi_2]{\td{N_1}{\psi_2}}{\td{N}{\psi_1\psi_2}}{A_{12}}$. Therefore it
suffices to show
$\inper[\Psi_2]{\td{M}{\psi_1\psi_2}}{\td{N}{\psi_1\psi_2}}{A_{12}}$, which
follows from our assumption at $\psi = \psi_1\psi_2$.
\end{proof}

%%%%%%%%%%%%%%%%%%%%%%%%%%%%%%%%%%%%%%%%%%%%%%%%%%%%%%%%%%%%%%%%%%%%%%%%%%%%%%%%
\subsection{Booleans}

We will define $\bool$ as a higher inductive type, meaning that we freely add
Kan composites as higher cells, rather than specifying that all its higher cells
are exactly $\true$ or $\false$. We do this to demonstrate the
robustness of our canonicity theorem and our treatment of $\notb{x}$, but in
practice it may be convenient to have (instead or in addition) a type of
``strict booleans.''

We define the relation $\vinper[-]{-}{-}{\bool}$ as the least relation closed
under:
\begin{enumerate}
\item $\vinper{\true}{\true}{\bool}$,
\item $\vinper{\false}{\false}{\bool}$, and
\item $\vinper[\Psi,x]{\hcomgeneric{x}{\bool}}{\hcom{x}{\bool}{r}{r'}{O}{y.P^0,y.P^1}}{\bool}$ whenever $r\neq r'$,
\begin{enumerate}
\item \ceqtm[\Psi,x]{M}{O}{\bool},
\item \ceqtm[\Psi,y]{\dsubst{N^\e}{\e}{x}}{\dsubst{P^\e}{\e}{x}}{\bool} for
$\e=0,1$, and
\item \ceqtm[\Psi]{\dsubst{\dsubst{N^\e}{r}{y}}{\e}{x}}{\dsubst{M}{\e}{x}}{\bool}
for $\e=0,1$.
\end{enumerate}
\end{enumerate}
Note that this relation is symmetric because the first two premises of the third
case ensure that
\ceqtm[\Psi]{\dsubst{M}{\e}{x}}{\dsubst{O}{\e}{x}}{\bool}
and
\ceqtm[\Psi]{\dsubst{\dsubst{N^\e}{r}{y}}{\e}{x}}{\dsubst{\dsubst{P^\e}{r}{y}}{\e}{x}}{\bool},
so
\ceqtm[\Psi]{\dsubst{\dsubst{P^\e}{r}{y}}{\e}{x}}{\dsubst{O}{\e}{x}}{\bool}.

The self-references in this definition can be seen by unrolling the definition
of (for example) $\ceqtm[\Psi,x]{M}{O}{\bool}$, which is:
for any $\msubsts{\Psi_1}{\psi_1}{(\Psi,x)}$ and
$\msubsts{\Psi_2}{\psi_2}{\Psi_1}$, 
\begin{enumerate}
\item
$\td{M}{\psi_1}\evals M_1$, 
$\td{M_1}{\psi_2}\evals M_2$, 
$\td{M}{\psi_1\psi_2}\evals M_{12}$,
\item
$\td{O}{\psi_1}\evals O_1$, 
$\td{O_1}{\psi_2}\evals O_2$, 
$\td{O}{\psi_1\psi_2}\evals O_{12}$, such that
\item
$\vinperfour[\Psi_2]{M_2}{M_{12}}{O_2}{O_{12}}{\bool}$.
\end{enumerate}

\paragraph{Pretype}
$\cpretype{\bool}$.

For all $\psi_1,\psi_2$,
$\td{\bool}{\psi_1}\evals\bool$, 
$\td{\bool}{\psi_2}\evals\bool$, 
$\td{\bool}{\psi_1\psi_2}\evals\bool$, 
and $\per[\Psi_2]{\bool}$.

\paragraph{Introduction}
$\coftype{\true}{\bool}$ and $\coftype{\false}{\bool}$.

For all $\msubsts{\Psi_1}{\psi_1}{\Psi}$
and $\msubsts{\Psi_2}{\psi_2}{\Psi_1}$,
$\td{\true}{\psi_1}\evals\true$,
$\td{\true}{\psi_2}\evals\true$,
$\td{\true}{\psi_1\psi_2}\evals\true$,
and $\vinper[\Psi_2]{\true}{\true}{\bool}$.
The $\false$ case is analogous.

\paragraph{Elimination}
If $\ceqtm{M}{N}{\bool}$, $\cwftype{A}$, $\coftype{T}{A}$, and
$\coftype{F}{A}$, then $\ceqtm{\ifb{A}{M}{T}{F}}{\ifb{A}{N}{T}{F}}{A}$.

Our proof of the elimination rule (that $\ifsym$ respects $\eq$ up to $\eq$)
requires us to know that $\ifsym$ respects $\inpersym$ up to $\inpersym$.
We first prove that the elimination rule holds for those booleans on whose
aspects $\ifsym$ respects $\inpersym$. We then use this fact to prove that
$\ifsym$ \emph{always} respects $\inpersym$, and therefore that the elimination
rule holds for all booleans.

\begin{definition}
For $\ceqtm{M}{N}{\bool}$, $\cwftype{A}$, $\coftype{T}{A}$, and
$\coftype{F}{A}$, we say that \emph{$\ifsym$ is coherent on values} if
for any $\msubsts{\Psi_1}{\psi_1}{\Psi}$ and $\msubsts{\Psi_2}{\psi_2}{\Psi_1}$,
\[\begin{aligned}
\inpertab[\Psi_2]
{\ifb{\td{A}{\psi_1\psi_2}}{M_{12}}{\td{T}{\psi_1\psi_2}}{\td{F}{\psi_1\psi_2}}}
{\ifb{\td{A}{\psi_1\psi_2}}{M_2}{\td{T}{\psi_1\psi_2}}{\td{F}{\psi_1\psi_2}}}
{A_{12}} \\
\inper[\Psi_2]{}{}{A_{12}}
\inpertab[\Psi_2]
{\ifb{\td{A}{\psi_1\psi_2}}{N_{12}}{\td{T}{\psi_1\psi_2}}{\td{F}{\psi_1\psi_2}}}
{\ifb{\td{A}{\psi_1\psi_2}}{N_2}{\td{T}{\psi_1\psi_2}}{\td{F}{\psi_1\psi_2}}}
{A_{12}}
\end{aligned}\]
where $M_{12},M_2,N_{12},N_2$ are the coherent aspects of $M,N$.
\end{definition}

\begin{lemma}\label{lem:if-ceqtm}
If $\ceqtm{M}{N}{\bool}$, $\cwftype{A}$, $\coftype{T}{A}$, $\coftype{F}{A}$, 
and $\ifsym$ is coherent on values for these parameters, then
$\ceqtm{\ifb{A}{M}{T}{F}}{\ifb{A}{N}{T}{F}}{A}$.
\end{lemma}
\begin{proof}
Here we work through the proof for the unary case (if $\coftype{M}{\bool}$ then
$\coftype{\ifb{A}{M}{T}{F}}{A}$); the binary case follows by repeating the
argument for $N$. Let $I = \ifb{A}{M}{T}{F}$. We need to show that for all
$\msubsts{\Psi_1}{\psi_1}{\Psi}$ and
$\msubsts{\Psi_2}{\psi_2}{\Psi_1}$,
$\td{I}{\psi_1} \evals I_1$, $\td{I_1}{\psi_2}\evals I_2$,
$\td{I}{\psi_1\psi_2} \evals I_{12}$, and
$\vinper[\Psi_2]{I_2}{I_{12}}{A_{12}}$
where $\td{A}{\psi_1\psi_2} \evals A_{12}$.

Expanding $\coftype{M}{\bool}$, we get
$\td{M}{\psi_1}\evals M_1$,
$\td{M_1}{\psi_2}\evals M_2$,
$\td{M}{\psi_1\psi_2}\evals M_{12}$,
$\vinper[\Psi_1]{M_1}{M_1}{\bool}$, and 
$\vinper[\Psi_2]{M_2}{M_{12}}{\bool}$.
By the operational semantics for $\ifsym$,
$\ifb{\td{A}{\psi_1}}{\td{M}{\psi_1}}{\td{T}{\psi_1}}{\td{F}{\psi_1}} \steps^*
\ifb{\td{A}{\psi_1}}{M_1}{\td{T}{\psi_1}}{\td{F}{\psi_1}}$.
This term's next step depends on $M_1$, which we determine by induction on
$\vinper[\Psi_1]{M_1}{M_1}{\bool}$:

\begin{enumerate}
\item $\vinper[\Psi_1]{\true}{\true}{\bool}$.

Then $M_1 = \true$, and 
$\ifb{\td{A}{\psi_1}}{\true}{\td{T}{\psi_1}}{\td{F}{\psi_1}} \steps
\td{T}{\psi_1}$. By $\coftype{T}{A}$ we have
$\td{T}{\psi_1}\evals T_1$,
$\td{T_1}{\psi_2}\evals T_2$,
$\td{T}{\psi_1\psi_2}\evals T_{12}$, and
$\vinper[\Psi_2]{T_2}{T_{12}}{A_{12}}$
where $\td{A}{\psi_1\psi_2} \evals A_{12}$, so $I_1 = T_1$ and $I_2 = T_2$.
To determine $I_{12}$, notice that
$\ifb{\td{A}{\psi_1\psi_2}}{\td{M}{\psi_1\psi_2}}{\td{T}{\psi_1\psi_2}}{\td{F}{\psi_1\psi_2}}
\steps^*
\ifb{\td{A}{\psi_1\psi_2}}{M_{12}}{\td{T}{\psi_1\psi_2}}{\td{F}{\psi_1\psi_2}}$.
Since $\td{\true}{\psi_2} \evals M_2 = \true$ and
$\vinper[\Psi_2]{M_2}{M_{12}}{\bool}$, $M_{12} = \true$ also. Then
$\ifb{\td{A}{\psi_1\psi_2}}{\true}{\td{T}{\psi_1\psi_2}}{\td{F}{\psi_1\psi_2}}
\steps \td{T}{\psi_1\psi_2} \evals T_{12}$. Therefore $I_{12} = T_{12}$ and
$\vinper[\Psi_2]{I_2}{I_{12}}{A_{12}}$.

\item $\vinper[\Psi_1]{\false}{\false}{\bool}$.

Same as previous case but with $I_1 = F_1$, $I_2 = F_2$, and $I_{12} = F_{12}$.

\item $\vinper[\Psi_1',x]{\hcom{x}{\bool}{r}{r'}{M'}{y.N^0,y.N^1}}{\hcom{x}{\bool}{r}{r'}{M'}{y.N^0,y.N^1}}{\bool}$ where 
$\Psi_1 = (\Psi_1',x)$,
$r\neq r'$,
$\coftype[\Psi_1',x]{M'}{\bool}$,
$\coftype[\Psi_1',y]{\dsubst{N^\e}{\e}{x}}{\bool}$ for $\e=0,1$, and
$\ceqtm[\Psi_1']{\dsubst{\dsubst{N^\e}{r}{y}}{\e}{x}}{\dsubst{M'}{\e}{x}}{\bool}$ for $\e=0,1$.

Then
$\ifb{\td{A}{\psi_1}}{\hcom{x}{\bool}{r}{r'}{M'}{y.N^0,y.N^1}}{\td{T}{\psi_1}}{\td{F}{\psi_1}} \steps H$ where
\[
H := \hcom{x}{\td{A}{\psi_1}}{r}{r'}
{\ifb{\td{A}{\psi_1}}{M'}{\td{T}{\psi_1}}{\td{F}{\psi_1}}}
{y.\ifb{\td{A}{\psi_1}}{N^0}{\td{T}{\psi_1}}{\td{F}{\psi_1}},
y.\ifb{\td{A}{\psi_1}}{N^1}{\td{T}{\psi_1}}{\td{F}{\psi_1}}}.
\]
We will show $\coftype[\Psi_1]{H}{\td{A}{\psi_1}}$, which implies
$\td{I}{\psi_1}\evals I_1$,
$\td{I_1}{\psi_2}\evals I_2$,
$\td{H}{\psi_2}\evals H_2$, and
$\vinper[\Psi_2]{H_2}{I_2}{A_{12}}$.
(It does not give us information about $I_{12}$, because $\td{I}{\psi_1\psi_2}$
might not step to $H$.)

Since $\cwftype{A}$, we know $A$ is Kan;
by the first Kan condition, it suffices to show 
\begin{enumerate}
\item $\coftype[\Psi_1',x]{\ifb{\td{A}{\psi_1}}{M'}{\td{T}{\psi_1}}{\td{F}{\psi_1}}}{\td{A}{\psi_1}}$,
\item $\coftype[\Psi_1',y]{\ifb{\dsubst{\td{A}{\psi_1}}{\e}{x}}{\dsubst{N^\e}{\e}{x}}{\dsubst{\td{T}{\psi_1}}{\e}{x}}{\dsubst{\td{F}{\psi_1}}{\e}{x}}}{\dsubst{\td{A}{\psi_1}}{\e}{x}}$ for $\e=0,1$,
and
\item $\ceqtm[\Psi_1',y]{\ifb{\dsubst{\td{A}{\psi_1}}{\e}{x}}{\dsubst{\dsubst{N^\e}{r}{y}}{\e}{x}}{\dsubst{\td{T}{\psi_1}}{\e}{x}}{\dsubst{\td{F}{\psi_1}}{\e}{x}}}
{\ifb{\dsubst{\td{A}{\psi_1}}{\e}{x}}{\dsubst{M'}{\e}{x}}{\dsubst{\td{T}{\psi_1}}{\e}{x}}{\dsubst{\td{F}{\psi_1}}{\e}{x}}}
{\dsubst{\td{A}{\psi_1}}{\e}{x}}$ for $\e=0,1$.
\end{enumerate}
All three can be obtained by applying the inductive hypothesis to the
typing and equality information we extracted from this case of
$\vinper[\Psi_1',x]{M_1}{M_1}{\bool}$. (Note that $x$ can occur in
$\td{T}{\psi_1}$ and $\td{F}{\psi_1}$, but $y$ cannot; and that we need the
binary version of the inductive hypothesis in order to derive the adjacency
condition.)

Taking stock, we now know $I_1$ and $I_2$, and must show
$\vinper[\Psi_2]{I_2}{I_{12}}{A_{12}}$ where $\td{I}{\psi_1\psi_2} \evals
I_{12}$. Since $\vinper[\Psi_2]{I_2}{H_2}{A_{12}}$, it suffices to show
$\vinper[\Psi_2]{H_2}{I_{12}}{A_{12}}$. But $\td{I}{\psi_1\psi_2} \steps^*
\ifb{\td{A}{\psi_1\psi_2}}{M_{12}}{\td{T}{\psi_1\psi_2}}{\td{F}{\psi_1\psi_2}}$
where $\td{M}{\psi_1\psi_2}\evals M_{12}$, so $I_{12}$ is determined by
$M_{12}$, which is in turn determined by $\td{M_1}{\psi_2}\evals M_2$ because
$\vinper[\Psi_2]{M_2}{M_{12}}{\bool}$. Therefore we proceed by considering the
three possible ways $\td{M_1}{\psi_2} = 
\hcom{\td{x}{\psi_2}}{\bool}{\td{r}{\psi_2}}{\td{r'}{\psi_2}}{\td{M'}{\psi_2}}{y.\td{N^0}{\psi_2},y.\td{N^1}{\psi_2}}$
can evaluate.

\begin{enumerate}
\item 
$\td{M_1}{\psi_2} \steps \dsubst{\td{N^\e}{\psi_2}}{\td{r'}{\psi_2}}{y}$
because $\td{x}{\psi_2} = \e$.

Since $\td{x}{\psi_2} = \e$,
$\dsubst{\td{\dsubst{N^\e}{\e}{x}}{\psi_2}}{\td{r'}{\psi_2}}{y} =
\dsubst{\td{N^\e}{\psi_2}}{\td{r'}{\psi_2}}{y}$. Thus by the typing premise of
$\vinper[\Psi_1',x]{M_1}{M_1}{\bool}$,
$\dsubst{\td{N^\e}{\psi_2}}{\td{r'}{\psi_2}}{y} \evals N^\e_2 = M_2$.
By the third Kan condition of $A$, $\td{x}{\psi_2} = \e$, and $y
\fresh \td{T}{\psi_1\psi_2},\td{F}{\psi_1\psi_2}$, we have
\[
\ceqtm[\Psi_2]{\td{H}{\psi_2}}{
\ifb{\td{A}{\psi_1\psi_2}}{\dsubst{\td{N^\e}{\psi_2}}{\td{r'}{\psi_2}}{y}}{\td{T}{\psi_1\psi_2}}{\td{F}{\psi_1\psi_2}}}{\td{A}{\psi_1\psi_2}}
\]
and so
\[\begin{aligned}
\inpertab[\Psi_2]{H_2}
{\ifb{\td{A}{\psi_1\psi_2}}{\dsubst{\td{N^\e}{\psi_2}}{\td{r'}{\psi_2}}{y}}{\td{T}{\psi_1\psi_2}}{\td{F}{\psi_1\psi_2}}}{A_{12}} \\
\inpertab[\Psi_2]{}
{\ifb{\td{A}{\psi_1\psi_2}}{N^\e_2}{\td{T}{\psi_1\psi_2}}{\td{F}{\psi_1\psi_2}}}{A_{12}} \\
\inpertab[\Psi_2]{}
{\ifb{\td{A}{\psi_1\psi_2}}{M_{12}}{\td{T}{\psi_1\psi_2}}{\td{F}{\psi_1\psi_2}}}{A_{12}} \\
\inpertab[\Psi_2]{}{\td{I}{\psi_1\psi_2}}{A_{12}} \\
\end{aligned}\]
where the middle step uses that $\ifsym$ is coherent on values
(for $\vinper[\Psi_2]{N^\e_2 = M_2}{M_{12}}{\bool}$).

\item
$\td{M_1}{\psi_2} \steps \td{M'}{\psi_2}$
because $\td{x}{\psi_2} = x'$ and $\td{r}{\psi_2} = \td{r'}{\psi_2}$.

By the typing premise of 
$\vinper[\Psi_1',x]{M_1}{M_1}{\bool}$, $\td{M'}{\psi_2}\evals M'_2 =
M_2$. By the second Kan condition of $A$ and $\td{r}{\psi_2} =
\td{r'}{\psi_2}$, 
\[
\ceqtm[\Psi_2]{\td{H}{\psi_2}}{\ifb{\td{A}{\psi_1}}{\td{M'}{\psi_2}}{\td{T}{\psi_1\psi_2}}{\td{F}{\psi_1\psi_2}}}{\td{A}{\psi_1\psi_2}}
\]
and so we conclude $\vinper[\Psi_2]{H_2}{I_{12}}{A_{12}}$ as in the previous
case, again relying on the assumption that $\ifsym$ is coherent on values.

\item
$\isval{\td{M_1}{\psi_2}}$
because $\td{x}{\psi_2} = x'$ and $\td{r}{\psi_2} \neq \td{r'}{\psi_2}$.

Then $\Psi_2 = (\Psi_2',x')$, $\td{M_1}{\psi_2} = M_2$, and by
$\vinper[\Psi_2]{\td{M_1}{\psi_2}}{M_{12}}{\bool}$, we know
\[
M_{12} = \hcom{x'}{\bool}{\td{r}{\psi_2}}{\td{r'}{\psi_2}}{O}{y.P^0,y.P^1}
\]
where
$\ceqtm[\Psi_2',x']{\td{M'}{\psi_2}}{O}{\bool}$ and
$\ceqtm[\Psi_2',y]{\dsubst{\td{N^\e}{\psi_2}}{\e}{x'}}{\dsubst{P^\e}{\e}{x'}}{\bool}$
for $\e=0,1$.
Then $\td{I}{\psi_1\psi_2} \steps^*$
\[
\hcom{x'}{\td{A}{\psi_1\psi_2}}{\td{r}{\psi_2}}{\td{r'}{\psi_2}}
{\ifb{\td{A}{\psi_1\psi_2}}{O}{\td{T}{\psi_1\psi_2}}{\td{F}{\psi_1\psi_2}}}
{y.\ifb{\td{A}{\psi_1\psi_2}}{P^0}{\td{T}{\psi_1\psi_2}}{\td{F}{\psi_1\psi_2}},
\dots}
\]
Call this term $H'$. By the inductive hypothesis applied to
$\ceqtm[\Psi_2',x']{\td{M'}{\psi_2}}{O}{\bool}$,
\[
\ceqtm[\Psi_2',x']
{\ifb{\td{A}{\psi_1\psi_2}}{\td{M'}{\psi_2}}{\td{T}{\psi_1\psi_2}}{\td{F}{\psi_1\psi_2}}}
{\ifb{\td{A}{\psi_1\psi_2}}{O}{\td{T}{\psi_1\psi_2}}{\td{F}{\psi_1\psi_2}}}
{\td{A}{\psi_1\psi_2}}
\]
and similarly for the other components of the $\hcomsym$. By the first Kan
condition of $A$, these equations imply
$\ceqtm[\Psi_2]{H'}{\td{H}{\psi_2}}{\td{A}{\psi_1\psi_2}}$,
so in particular $\inper[\Psi_2]{H'}{\td{H}{\psi_2}}{A_{12}}$ and thus
$\vinper[\Psi_2]{I_{12}}{H_2}{A_{12}}$.
\qedhere
\end{enumerate}
\end{enumerate}
\end{proof}

\begin{lemma}\label{lem:if-vinper}
If $\vinper{M}{N}{\bool}$, $\cwftype{A}$, $\coftype{T}{A}$, and
$\coftype{F}{A}$,
then $\inper{\ifb{A}{M}{T}{F}}{\ifb{A}{N}{T}{F}}{A_0}$ where $A\evals A_0$.
\end{lemma}
\begin{proof}
By induction on $\vinper{M}{N}{\bool}$.
\begin{enumerate}
\item $\vinper{\true}{\true}{\bool}$.

Then $\ifb{A}{\true}{T}{F} \steps T \evals T_0$, and $\vinper{T_0}{T_0}{A_0}$.

\item $\vinper{\false}{\false}{\bool}$.

Then $\ifb{A}{\false}{T}{F} \steps F \evals F_0$, and $\vinper{F_0}{F_0}{A_0}$.

\item
$\vinper[\Psi',x]{\hcom{x}{\bool}{r}{r'}{M'}{y.N^0,y.N^1}}{\hcom{x}{\bool}{r}{r'}{O}{y.P^0,y.P^1}}{\bool}$
where $\Psi=(\Psi',x)$, $r\neq r'$,
$\ceqtm[\Psi',x]{M'}{O}{\bool}$,
$\ceqtm[\Psi',y]{\dsubst{N^\e}{\e}{x}}{\dsubst{P^\e}{\e}{x}}{\bool}$ for
$\e=0,1$, and
$\ceqtm[\Psi']{\dsubst{\dsubst{N^\e}{r}{y}}{\e}{x}}{\dsubst{M'}{\e}{x}}{\bool}$
for $\e=0,1$.

Then $\ifb{A}{M}{T}{F} \steps
\hcom{x}{A}{r}{r'}{\ifb{A}{M'}{T}{F}}
{y.\ifb{A}{N^0}{T}{F},y.\ifb{A}{N^1}{T}{F}}$
and similarly for $\ifb{A}{N}{T}{F}$. To show the resulting $\hcomsym$s are
$\inper{-}{-}{A_0}$, it suffices to show that they are 
$\ceqtm{-}{-}{A}$. We appeal to the first Kan condition of $A$, which applies
when
\begin{enumerate}
\item $\ceqtm[\Psi',x]{\ifb{A}{M'}{T}{F}}{\ifb{A}{O}{T}{F}}{A}$,
\item
$\ceqtm[\Psi',y]
{\dsubst{(\ifb{A}{N^\e}{T}{F})}{\e}{x}}
{\dsubst{(\ifb{A}{P^\e}{T}{F})}{\e}{x}}
{\dsubst{A}{\e}{x}}$ for $\e=0,1$, and
\item
$\ceqtm[\Psi']
{\dsubst{(\ifb{A}{\dsubst{N^\e}{r}{y}}{T}{F})}{\e}{x}}
{\dsubst{(\ifb{A}{M'}{T}{F})}{\e}{x}}
{\dsubst{A}{\e}{x}}$ for $\e=0,1$.
\end{enumerate}
We establish these equalities by appealing to \cref{lem:if-ceqtm} at
$\ceqtm[\Psi',x]{M'}{O}{\bool}$, etc., whose
$\inper[\Psi_2]
{\ifb{\td{A}{\psi_1\psi_2}}{-}{\td{T}{\psi_1\psi_2}}{\td{F}{\psi_1\psi_2}}}
{\ifb{\td{A}{\psi_1\psi_2}}{-}{\td{T}{\psi_1\psi_2}}{\td{F}{\psi_1\psi_2}}}
{A_{12}}$ hypothesis we establish with the current lemma's
inductive hypothesis.
\qedhere
\end{enumerate}
\end{proof}

Finally, the elimination rule for $\bool$ follows directly from
\cref{lem:if-ceqtm,lem:if-vinper}, because \cref{lem:if-vinper} implies that
$\ifsym$ is always coherent on values.

\paragraph{Computation}
If $\cwftype{A}$, $\coftype{T}{A}$, and $\coftype{F}{A}$, then
$\ceqtm{\ifb{A}{\true}{T}{F}}{T}{A}$ and
$\ceqtm{\ifb{A}{\false}{T}{F}}{F}{A}$.

For all $\psi$, 
$\ifb{\td{A}{\psi}}{\true}{\td{T}{\psi}}{\td{F}{\psi}} \steps \td{T}{\psi}$, so
the former follows from head expansion and $\ceqtm{T}{T}{A}$. The latter case is
analogous.

\paragraph{Kan}
Show $\cpretype{\bool}$ is Kan.

We will once again prove only the unary version of the first condition, in order
to lessen the notational burden; the binary version follows by the same
argument. Show that for any $\Psi'$, if
\begin{enumerate}
\item \coftype[\Psi',x]{M}{\bool},
\item \coftype[\Psi',y]{\dsubst{N^\e}{\e}{x}}{\bool} for $\e=0,1$, and
\item \ceqtm[\Psi']{\dsubst{\dsubst{N^\e}{r}{y}}{\e}{x}}{\dsubst{M}{\e}{x}}{\bool} for $\e=0,1$,
\end{enumerate}
then $\coftype[\Psi',x]{\hcomgeneric{x}{\bool}}{\bool}$.
That is, for any $\msubsts{\Psi_1}{\psi_1}{(\Psi',x)}$ and
$\msubsts{\Psi_2}{\psi_2}{\Psi_1}$, $\td{\hcomsym}{\psi_1}\evals H_1$,
$\td{H_1}{\psi_2}\evals H_2$, $\td{\hcomsym}{\psi_1\psi_2}\evals H_{12}$, and
$\vinper[\Psi_2]{H_{12}}{H_2}{\bool}$. We prove this by case-analyzing how
$x,r,r'$ are affected by $\psi_1$ and $\psi_1\psi_2$.

\begin{enumerate}
\item $\td{x}{\psi_1}=\e$. (Therefore $\td{x}{\psi_1\psi_2}=\e$ also.)
Then 
\[\begin{aligned}
\hcom{\e}{\bool}{\td{r}{\psi_1}}{\td{r'}{\psi_1}}{\td{M}{\psi_1}}{y.\td{N^0}{\psi_1},y.\td{N^1}{\psi_1}}
&\steps \dsubst{\td{N^\e}{\psi_1}}{\td{r'}{\psi_1}}{y} \\
\hcom{\e}{\bool}{\td{r}{\psi_1\psi_2}}{\td{r'}{\psi_1\psi_2}}{\td{M}{\psi_1\psi_2}}{y.\td{N^0}{\psi_1\psi_2},y.\td{N^1}{\psi_1\psi_2}}
&\steps \dsubst{\td{N^\e}{\psi_1\psi_2}}{\td{r'}{\psi_1\psi_2}}{y}
\end{aligned}\]
By $\coftype[\Psi',y]{\dsubst{N^\e}{\e}{x}}{\bool}$ and
$\td{\dsubst{N^\e}{\e}{x}}{\psi_1} = \td{N^\e}{\psi_1}$,
we know that
$\dsubst{\td{N^\e}{\psi_1}}{\td{r'}{\psi_1}}{y} \evals N^\e_1$,
$\td{N^\e_1}{\psi_2} \evals N^\e_2$,
$\dsubst{\td{N^\e}{\psi_1\psi_2}}{\td{r'}{\psi_1\psi_2}}{y} 
= \td{\dsubst{\td{N^\e}{\psi_1}}{\td{r'}{\psi_1}}{y}}{\psi_2}
\evals N^\e_{12}$,
and $\vinper[\Psi_2]{N^\e_{12}}{N^\e_2}{\bool}$.

\item $\td{x}{\psi_1}=x'$, $\td{r}{\psi_1}=\td{r'}{\psi_1}$, and
$\td{x}{\psi_1\psi_2}=\e$. Then
\[\begin{aligned}
\hcom{x'}{\bool}{\td{r}{\psi_1}}{\td{r'}{\psi_1}}{\td{M}{\psi_1}}{y.\td{N^0}{\psi_1},y.\td{N^1}{\psi_1}}
&\steps \td{M}{\psi_1} \\
\hcom{\e}{\bool}{\td{r}{\psi_1\psi_2}}{\td{r'}{\psi_1\psi_2}}{\td{M}{\psi_1\psi_2}}{y.\td{N^0}{\psi_1\psi_2},y.\td{N^1}{\psi_1\psi_2}}
&\steps \dsubst{\td{N^\e}{\psi_1\psi_2}}{\td{r'}{\psi_1\psi_2}}{y}
\end{aligned}\]
By $\coftype[\Psi',x]{M}{\bool}$ we know $\td{M}{\psi_1}\evals M_1$,
$\td{M_1}{\psi_2}\evals M_2$, $\td{M}{\psi_1\psi_2}\evals M_{12}$, and
$\vinper[\Psi_2]{M_{12}}{M_2}{\bool}$, so by transitivity it suffices to show
$\inper[\Psi_2]{\td{M}{\psi_1\psi_2}}{\dsubst{\td{N^\e}{\psi_1\psi_2}}{\td{r'}{\psi_1\psi_2}}{y}}{\bool}$.
By $\td{x}{\psi_1\psi_2}=\e$, we have
$\td{M}{\psi_1\psi_2} = \td{\dsubst{M}{\e}{x}}{\psi_1\psi_2}$.
By $\td{x}{\psi_1\psi_2}=\e$ and $\td{r}{\psi_1}=\td{r'}{\psi_1}$, we have
$\dsubst{\td{N^\e}{\psi_1\psi_2}}{\td{r'}{\psi_1\psi_2}}{y} =
 \dsubst{\td{N^\e}{\psi_1\psi_2}}{\td{r}{\psi_1\psi_2}}{y} =
 \dsubst{\td{\dsubst{N^\e}{\e}{x}}{\psi_1\psi_2}}{\td{r}{\psi_1\psi_2}}{y} =
 \td{\dsubst{\dsubst{N^\e}{r}{y}}{\e}{x}}{\psi_1\psi_2}$.
Therefore our desired equation follows directly from
$\ceqtm[\Psi']{\dsubst{\dsubst{N^\e}{r}{y}}{\e}{x}}{\dsubst{M}{\e}{x}}{\bool}$
under $\psi_1\psi_2$.

\item $\td{x}{\psi_1}=x'$, $\td{r}{\psi_1}=\td{r'}{\psi_1}$, and
$\td{x}{\psi_1\psi_2}=x''$. (Therefore
$\td{r}{\psi_1\psi_2}=\td{r'}{\psi_1\psi_2}$.) Then
\[\begin{aligned}
\hcom{x'}{\bool}{\td{r}{\psi_1}}{\td{r'}{\psi_1}}{\td{M}{\psi_1}}{y.\td{N^0}{\psi_1},y.\td{N^1}{\psi_1}}
&\steps \td{M}{\psi_1} \\
\hcom{x''}{\bool}{\td{r}{\psi_1\psi_2}}{\td{r'}{\psi_1\psi_2}}{\td{M}{\psi_1\psi_2}}{y.\td{N^0}{\psi_1\psi_2},y.\td{N^1}{\psi_1\psi_2}}
&\steps \td{M}{\psi_1\psi_2}
\end{aligned}\]
By $\coftype[\Psi',x]{M}{\bool}$ we know $\td{M}{\psi_1}\evals M_1$,
$\td{M_1}{\psi_2}\evals M_2$, $\td{M}{\psi_1\psi_2}\evals M_{12}$, and
$\vinper[\Psi_2]{M_{12}}{M_2}{\bool}$.

\item $\td{x}{\psi_1}=x'$, $\td{r}{\psi_1}\neq\td{r'}{\psi_1}$, and
$\td{x}{\psi_1\psi_2}=\e$. Then
\[\begin{aligned}
\isvaltab{\hcom{x'}{\bool}{\td{r}{\psi_1}}{\td{r'}{\psi_1}}{\td{M}{\psi_1}}{y.\td{N^0}{\psi_1},y.\td{N^1}{\psi_1}}} \\
\hcom{\e}{\bool}{\td{r}{\psi_1\psi_2}}{\td{r'}{\psi_1\psi_2}}{\td{M}{\psi_1\psi_2}}{y.\td{N^0}{\psi_1\psi_2},y.\td{N^1}{\psi_1\psi_2}}
&\steps \dsubst{\td{N^\e}{\psi_1\psi_2}}{\td{r'}{\psi_1\psi_2}}{y}
\end{aligned}\]
In this case $H_1 = \td{\hcomsym}{\psi_1}$, so
$\td{H_1}{\psi_2} = \td{\hcomsym}{\psi_1\psi_2}$ and we must show
$\dsubst{\td{N^\e}{\psi_1\psi_2}}{\td{r'}{\psi_1\psi_2}}{y} \evals N^\e_{12}$
and $\vinper[\Psi_2]{N^\e_{12}}{N^\e_{12}}{\bool}$.
This follows from
$\coftype[\Psi',y]{\dsubst{N^\e}{\e}{x}}{\bool}$ and
$\dsubst{\td{N^\e}{\psi_1\psi_2}}{\td{r'}{\psi_1\psi_2}}{y} =
 \dsubst{\td{\dsubst{N^\e}{\e}{x}}{\psi_1\psi_2}}{\td{r'}{\psi_1\psi_2}}{y}$.

\item $\td{x}{\psi_1}=x'$, $\td{r}{\psi_1}\neq\td{r'}{\psi_1}$,
$\td{x}{\psi_1\psi_2}=x''$, and $\td{r}{\psi_1\psi_2} = \td{r'}{\psi_1\psi_2}$.
Then
\[\begin{aligned}
\isvaltab{\hcom{x'}{\bool}{\td{r}{\psi_1}}{\td{r'}{\psi_1}}{\td{M}{\psi_1}}{y.\td{N^0}{\psi_1},y.\td{N^1}{\psi_1}}} \\
\hcom{x''}{\bool}{\td{r}{\psi_1\psi_2}}{\td{r'}{\psi_1\psi_2}}{\td{M}{\psi_1\psi_2}}{y.\td{N^0}{\psi_1\psi_2},y.\td{N^1}{\psi_1\psi_2}}
&\steps \td{M}{\psi_1\psi_2}
\end{aligned}\]
Once again, $H_1 = \td{\hcomsym}{\psi_1}$ and $\td{H_1}{\psi_2} = 
\td{\hcomsym}{\psi_1\psi_2}$, so we must show
$\td{M}{\psi_1\psi_2} \evals M_{12}$ and
$\vinper[\Psi_2]{M_{12}}{M_{12}}{\bool}$, which follows from
$\coftype[\Psi',x]{M}{\bool}$.

\item $\td{x}{\psi_1}=x'$, $\td{r}{\psi_1}\neq\td{r'}{\psi_1}$,
$\td{x}{\psi_1\psi_2}=x''$, and $\td{r}{\psi_1\psi_2}\neq\td{r'}{\psi_1\psi_2}$.
Then
\[\begin{aligned}
\isvaltab{\hcom{x'}{\bool}{\td{r}{\psi_1}}{\td{r'}{\psi_1}}{\td{M}{\psi_1}}{y.\td{N^0}{\psi_1},y.\td{N^1}{\psi_1}}} \\
\isvaltab{\hcom{x''}{\bool}{\td{r}{\psi_1\psi_2}}{\td{r'}{\psi_1\psi_2}}{\td{M}{\psi_1\psi_2}}{y.\td{N^0}{\psi_1\psi_2},y.\td{N^1}{\psi_1\psi_2}}}
\end{aligned}\]
Because $\td{H_1}{\psi_2} = \td{\hcomsym}{\psi_1\psi_2} = H_{12}$, we must show
that $\vinper[\Psi_2]{H_{12}}{H_{12}}{\bool}$. Let $\Psi_2 = \Psi_2',x''$. Then
$\vinper[\Psi_2',x'']{H_{12}}{H_{12}}{\bool}$ because
$\td{r}{\psi_1\psi_2}\neq\td{r'}{\psi_1\psi_2}$ and
$\coftype[\Psi_2',x'']{\td{M}{\psi_1\psi_2}}{\bool}$; by
$\td{\dsubst{N^\e}{\e}{x}}{\psi_1\psi_2} = 
 \dsubst{\td{N^\e}{\psi_1\psi_2}}{\e}{x''}$ we have that
$\coftype[\Psi_2',y]{\dsubst{\td{N}{\psi_1\psi_2}}{\e}{x''}}{\bool}$; and by
$\td{\dsubst{M}{\e}{x}}{\psi_1\psi_2} = 
 \dsubst{\td{M}{\psi_1\psi_2}}{\e}{x''}$ and 
$\td{\dsubst{\dsubst{N^\e}{r}{y}}{\e}{x}}{\psi_1\psi_2} = 
 \dsubst{\dsubst{\td{N^\e}{\psi_1\psi_2}}{\td{r}{\psi_1\psi_2}}{y}}{\e}{x''}$
 we have that
$\ceqtm[\Psi_2']{\dsubst{\dsubst{\td{N^\e}{\psi_1\psi_2}}{\td{r}{\psi_1\psi_2}}{y}}{\e}{x''}}{\dsubst{\td{M}{\psi_1\psi_2}}{\e}{x''}}{\bool}$.
\end{enumerate}

The second Kan condition asserts that when $r = r'$, the Kan
composition is equal to its ``cap'': for any $\Psi'$, if
\begin{enumerate}
\item \coftype[\Psi',x]{M}{\bool},
\item \coftype[\Psi',y]{\dsubst{N^\e}{\e}{x}}{\bool} for
$\e=0,1$, and
\item \ceqtm[\Psi']{\dsubst{\dsubst{N^\e}{r}{y}}{\e}{x}}{\dsubst{M}{\e}{x}}{\bool} for $\e=0,1$,
\end{enumerate}
then $\ceqtm[\Psi',x]{\hcom{x}{\bool}{r}{r}{M}{y.N^0,y.N^1}}{M}{\bool}$.
Recall that establishing such an equation requires showing that both sides have
coherent aspects, and moreover, those aspects are $\vinper[\Psi_2]{-}{-}{\bool}$
to the aspects of the other side. The first Kan condition establishes that the
left-hand side has coherent aspects, and the first hypothesis of this theorem
establishes the same for $M$, so it suffices to show that the two sides under
$\psi_1\psi_2$ are $\inper[\Psi_2]{-}{-}{\bool}$.

\begin{enumerate}
\item $\td{x}{\psi_1\psi_2}=\e$. Then
\[
\hcom{\e}{\bool}{\td{r}{\psi_1\psi_2}}{\td{r}{\psi_1\psi_2}}{\td{M}{\psi_1\psi_2}}{y.\td{N^0}{\psi_1\psi_2},y.\td{N^1}{\psi_1\psi_2}}
\steps \dsubst{\td{N^\e}{\psi_1\psi_2}}{\td{r}{\psi_1\psi_2}}{y}
\]
Since
$\dsubst{\td{N^\e}{\psi_1\psi_2}}{\td{r}{\psi_1\psi_2}}{y}
= \dsubst{\td{\dsubst{N^\e}{\e}{x}}{\psi_1\psi_2}}{\td{\dsubst{r}{\e}{x}}{\psi_1\psi_2}}{y}
= \td{\dsubst{\dsubst{N^\e}{r}{y}}{\e}{x}}{\psi_1\psi_2}$
and $\td{M}{\psi_1\psi_2} = \td{\dsubst{M}{\e}{x}}{\psi_1\psi_2}$, 
the adjacency assumption yields the desired equation:
$\inper[\Psi_2]
{\dsubst{\td{N^\e}{\psi_1\psi_2}}{\td{r}{\psi_1\psi_2}}{y}}
{\td{M}{\psi_1\psi_2}}{\bool}$.

\item $\td{x}{\psi_1\psi_2}=x'$. Then
\[
\hcom{x'}{\bool}{\td{r}{\psi_1\psi_2}}{\td{r}{\psi_1\psi_2}}{\td{M}{\psi_1\psi_2}}{y.\td{N^0}{\psi_1\psi_2},y.\td{N^1}{\psi_1\psi_2}}
\steps \td{M}{\psi_1\psi_2}
\]
and what we want to show, 
$\inper[\Psi_2]{\td{M}{\psi_1\psi_2}}{\td{M}{\psi_1\psi_2}}{\bool}$, follows
immediately from $\coftype[\Psi',x]{M}{\bool}$.
\end{enumerate}

The third Kan condition asserts that when $r_1 = \e$, the Kan
composition is equal to the $\e$ ``tube face'': for any $\Psi'$, if
\begin{enumerate}
\item \coftype[\Psi']{M}{\bool},
\item \coftype[\Psi',y]{N^\e}{\bool}, and
\item
\ceqtm[\Psi']{\dsubst{N^\e}{r}{y}}{M}{\bool},
\end{enumerate}
then $\ceqtm[\Psi']{\hcomgeneric{\e}{\bool}}{\dsubst{N^\e}{r'}{y}}{\bool}$.
This term always steps to a tube face, so we appeal to head expansion: for all
$\msubsts{\Psi''}{\psi}{\Psi'}$, 
\[
\hcom{\e}{\bool}{\td{r}{\psi}}{\td{r'}{\psi}}{\td{M}{\psi}}{y.\td{N^0}{\psi},y.\td{N^1}{\psi}}
\steps \dsubst{\td{N^\e}{\psi}}{\td{r'}{\psi}}{y}
= \td{\dsubst{N^\e}{r'}{y}}{\psi}
\]
and $\coftype[\Psi']{\dsubst{N^\e}{r'}{y}}{\bool}$, so
$\ceqtm[\Psi']{\hcomgeneric{\e}{\bool}}{\dsubst{N^\e}{r'}{y}}{\bool}$.

The fourth Kan condition asserts that one can coerce across the
type $\bool$: for any $\Psi'$, if
$\ceqtm[\Psi']{M}{N}{\bool}$, then
$\ceqtm[\Psi']{{\coe{x.\bool}{r}{r'}{M}}}{{\coe{x.\bool}{r}{r'}{N}}}{\bool}$.
But for any $\msubsts{\Psi''}{\psi}{\Psi'}$,
$\coe{x.\bool}{\td{r}{\psi}}{\td{r'}{\psi}}{\td{M}{\psi}}
\steps \td{M}{\psi}$ so by head expansion,
$\coe{x.\bool}{r}{r'}{M} \eq M
\eq N \eq \coe{x.\bool}{r}{r'}{N}$.

\paragraph{Cubical}
Show for any $\Psi'$ and $\vinper[\Psi']{M}{N}{\bool}$,
$\ceqtm[\Psi']{M}{N}{\bool}$.

We consider each case of $\vinper[\Psi']{M}{N}{\bool}$.
For $\true$ and $\false$, this follows from the introduction rules already
proven. For $\hcomgeneric{x}{\bool}$, this follows from the first Kan condition
of $\bool$, again already proven.

%%%%%%%%%%%%%%%%%%%%%%%%%%%%%%%%%%%%%%%%%%%%%%%%%%%%%%%%%%%%%%%%%%%%%%%%%%%%%%%%
\subsection{Circle}

Our definition of $\C$ is very similar to that of $\bool$, because we defined
$\bool$ as a higher inductive type (with no path constructors). We omit proofs
that proceed identically to those for $\bool$.

We define the relation $\vinper[-]{-}{-}{\C}$ as the least relation closed
under:
\begin{enumerate}
\item $\vinper{\base}{\base}{\C}$,
\item $\vinper[\Psi,x]{\lp{x}}{\lp{x}}{\C}$, and
\item $\vinper[\Psi,x]{\hcomgeneric{x}{\C}}{\hcom{x}{\C}{r}{r'}{O}{y.P^0,y.P^1}}{\C}$ whenever $r\neq r'$,
\begin{enumerate}
\item \ceqtm[\Psi,x]{M}{O}{\C},
\item \ceqtm[\Psi,y]{\dsubst{N^\e}{\e}{x}}{\dsubst{P^\e}{\e}{x}}{\C} for
$\e=0,1$, and
\item \ceqtm[\Psi]{\dsubst{\dsubst{N^\e}{r}{y}}{\e}{x}}{\dsubst{M}{\e}{x}}{\C}
for $\e=0,1$.
\end{enumerate}
\end{enumerate}

\paragraph{Pretype}
$\cpretype{\C}$.

For all $\psi_1,\psi_2$,
$\td{\C}{\psi_1}\evals\C$, 
$\td{\C}{\psi_2}\evals\C$, 
$\td{\C}{\psi_1\psi_2}\evals\C$, 
and $\per[\Psi_2]{\C}$.

\paragraph{Introduction}
$\coftype{\base}{\C}$, $\coftype{\lp{r}}{\C}$, and $\ceqtm{\lp{\e}}{\base}{\C}$.

\begin{enumerate}
\item
For all $\msubsts{\Psi_1}{\psi_1}{\Psi}$
and $\msubsts{\Psi_2}{\psi_2}{\Psi_1}$,
$\td{\base}{\psi_1}\evals\base$,
$\td{\base}{\psi_2}\evals\base$,
$\td{\base}{\psi_1\psi_2}\evals\base$,
and $\vinper[\Psi_2]{\base}{\base}{\C}$.

\item
For all $\msubsts{\Psi_1}{\psi_1}{\Psi}$
and $\msubsts{\Psi_2}{\psi_2}{\Psi_1}$,
we case on $\td{r}{\psi_1}$ and $\td{r}{\psi_1\psi_2}$:
\begin{enumerate}
\item $\td{r}{\psi_1} = \e$. (Therefore $\td{r}{\psi_1\psi_2} = \e$ also.)

Then $\lp{\e}\evals\base$,
$\td{\base}{\psi_2}\evals\base$,
$\lp{\e}\evals\base$,
and $\vinper[\Psi_2]{\base}{\base}{\C}$.

\item $\td{r}{\psi_1} = x$ and $\td{x}{\psi_2} = \e$.

Then $\lp{x}\evals\lp{x}$,
$\lp{\e}\evals\base$,
$\lp{\e}\evals\base$,
and $\vinper[\Psi_2]{\base}{\base}{\C}$.

\item $\td{r}{\psi_1} = x$ and $\td{x}{\psi_2} = x'$.
(Therefore $\Psi_2 = (\Psi_2',x')$.)

Then $\lp{x}\evals\lp{x}$,
$\lp{x'}\evals\lp{x'}$,
$\lp{x'}\evals\lp{x'}$,
and $\vinper[\Psi_2',x']{\lp{x'}}{\lp{x'}}{\C}$.
\end{enumerate}

\item By head expansion and the first introduction rule, since for all $\psi$,
$\lp{\td{\e}{\psi}}\steps\td{\base}{\psi}$.
\end{enumerate}

\paragraph{Elimination}
If $\ceqtm{M}{N}{\C}$,
$\cwftype{A}$,
$\coftype{P}{A}$,
$\coftype[\Psi,x]{L}{A}$,
and $\ceqtm{\dsubst{L}{\e}{x}}{P}{A}$ for $\e=0,1$,
then $\ceqtm{\Celim{A}{M}{P}{x.L}}{\Celim{A}{N}{P}{x.L}}{A}$.

We use essentially the same proof as for the elimination rule for booleans;
see \cref{lem:if-ceqtm,lem:if-vinper} for full details.

\begin{lemma}\label{lem:Celim-ceqtm}
If $\ceqtm{M}{N}{\C}$, $\cwftype{A}$, $\coftype{P}{A}$,
$\coftype[\Psi,x]{L}{A}$, $\ceqtm{\dsubst{L}{\e}{x}}{P}{A}$ for $\e=0,1$,
and $\Celimsym$ is coherent on values for these parameters, then
$\ceqtm{\Celim{A}{M}{P}{x.L}}{\Celim{A}{N}{P}{x.L}}{A}$.
\end{lemma}
\begin{proof}
Here we work through the proof for the unary case.
Let $E = \Celim{A}{M}{P}{x.L}$. We need to show that $E$ has confluent aspects
for all $\msubsts{\Psi_1}{\psi_1}{\Psi}$ and
$\msubsts{\Psi_2}{\psi_2}{\Psi_1}$.
We know
$\Celim{\td{A}{\psi_1}}{\td{M}{\psi_1}}{\td{P}{\psi_1}}{x.\td{L}{\psi_1}}
\steps^* \Celim{\td{A}{\psi_1}}{M_1}{\td{P}{\psi_1}}{x.\td{L}{\psi_1}}$
where $\td{M}{\psi_1} \evals M_1$. This term's next step depends on $M_1$, which
we determine by induction on $\vinper[\Psi_1]{M_1}{M_1}{\C}$. 
The $\hcomsym$ case is identical to that of \cref{lem:if-ceqtm}; the $\base$
case follows the pattern of the $\true$ case. Hence we consider only the
$\lp{y}$ case.

Then $M_1 = \lp{y}$, and 
$\Celim{\td{A}{\psi_1}}{M_1}{\td{P}{\psi_1}}{x.\td{L}{\psi_1}} \steps
\dsubst{\td{L}{\psi_1}}{y}{x}$.
By $\coftype[\Psi,x]{L}{A}$,
$\dsubst{\td{L}{\psi_1}}{y}{x} \evals L_1$ and $\td{L_1}{\psi_2} \evals L_2$,
and thus $E_1 = L_1$ and $E_2 = L_2$. To determine $E_{12}$ we case on
$\td{y}{\psi_2}$:
\begin{enumerate}
\item If $\td{y}{\psi_2} = \e$ then
$M_{12} = \base$, and
$\Celim{\td{A}{\psi_1\psi_2}}{\base}{\td{P}{\psi_1\psi_2}}{x.\td{L}{\psi_1\psi_2}}
\steps \td{P}{\psi_1\psi_2}$.
We obtain
$\inper[\Psi_2]{\td{P}{\psi_1\psi_2}}{\td{L_1}{\psi_2}}{A_{12}}$ because by
$\coftype[\Psi,x]{L}{A}$ we know
$\inper[\Psi_2]
{\td{L_1}{\psi_2}}
{\td{\dsubst{\td{L}{\psi_1}}{y}{x}}{\psi_2} = 
\td{\dsubst{L}{\e}{x}}{\psi_1\psi_2}}{A_{12}}$, and by
$\ceqtm{\dsubst{L}{\e}{x}}{P}{A}$ we know
$\inper[\Psi_2]{\td{P}{\psi_1\psi_2}}
{\td{\dsubst{L}{\e}{x}}{\psi_1\psi_2}}{A_{12}}$.

\item If $\td{y}{\psi_2} = y'$ then
$M_{12} = \lp{y'}$, and
$\Celim{\td{A}{\psi_1\psi_2}}{\lp{y'}}{\td{P}{\psi_1\psi_2}}{x.\td{L}{\psi_1\psi_2}}
\steps \dsubst{\td{L}{\psi_1\psi_2}}{y'}{x}$.
By $\coftype[\Psi,x]{L}{A}$ we know
$\inper[\Psi_2]
{\td{L_1}{\psi_2}}
{\td{\dsubst{\td{L}{\psi_1}}{y}{x}}{\psi_2} = 
\dsubst{\td{L}{\psi_1\psi_2}}{y'}{x}}{A_{12}}$ as needed.
\qedhere
\end{enumerate}
\end{proof}

\begin{lemma}\label{lem:Celim-vinper}
If $\vinper{M}{N}{\C}$, $\cwftype{A}$, 
$\coftype{P}{A}$,
$\coftype[\Psi,x]{L}{A}$,
and $\ceqtm{\dsubst{L}{\e}{x}}{P}{A}$ for $\e=0,1$, then
$\inper{\Celim{A}{M}{P}{x.L}}{\Celim{A}{N}{P}{x.L}}{A_0}$ where $A\evals A_0$.
\end{lemma}
\begin{proof}
By induction on $\vinper{M}{N}{\C}$. The $\hcomsym$ case is identical to that
of \cref{lem:if-vinper} and requires \cref{lem:Celim-ceqtm}; the $\base$ case
follows the pattern of the $\true$ case. We consider only the $\lp{y}$ case.

If $M = N = \lp{y}$ then
$\Celim{A}{\lp{y}}{P}{x.L} \steps \dsubst{L}{y}{x} \evals L_0$, and
$\vinper{L_0}{L_0}{A_0}$.
\end{proof}

Finally, the elimination rule for $\C$ follows directly from
\cref{lem:Celim-ceqtm,lem:Celim-vinper}, because \cref{lem:Celim-vinper} implies
that $\Celimsym$ is always coherent on values.

\paragraph{Computation}
If $\cwftype{A}$,
$\coftype{P}{A}$,
$\coftype[\Psi,x]{L}{A}$,
and $\ceqtm{\dsubst{L}{\e}{x}}{P}{A}$ for $\e=0,1$,
then $\ceqtm{\Celim{A}{\base}{P}{x.L}}{P}{A}$
and $\ceqtm{\Celim{A}{\lp{r}}{P}{x.L}}{\dsubst{L}{r}{x}}{A}$.

For all $\psi$, 
$\Celim{\td{A}{\psi}}{\base}{\td{P}{\psi}}{x.\td{L}{\psi}} \steps \td{P}{\psi}$,
so the first computation rule follows from head expansion and $\ceqtm{P}{P}{A}$.

The second computation rule requires a case analysis of how 
$E := \Celim{A}{\lp{r}}{P}{x.L}$ and $\lp{r}$
evaluate under $\msubsts{\Psi_1}{\psi_1}{\Psi}$ and
$\msubsts{\Psi_2}{\psi_2}{\Psi_1}$.
Notice that $\td{E}{\psi_1} \steps^* \Celim{A}{M_1}{P}{x.L}$
where $\lp{\td{r}{\psi_1}}\evals M_1$. Hence we case on $\td{r}{\psi_1}$:
\begin{enumerate}
\item $\td{r}{\psi_1} = \e$.
Then $M_1 = \base$, so $\td{E}{\psi_1} \steps^* \td{P}{\psi_1} \evals P_1$ and
$\td{P_1}{\psi_2}\evals P_2$ by $\coftype{P}{A}$.
But $\inper[\Psi_2]{\td{M_1}{\psi_2} = \base}{\td{M}{\psi_1\psi_2}}{\C}$, so
$\td{M}{\psi_1\psi_2} \evals \base$ and $\td{E}{\psi_1\psi_2} \steps^*
\td{P}{\psi_1\psi_2} \evals P_{12}$ where
$\vinper[\Psi_2]{P_{12}}{P_2}{A_{12}}$.

\item $\td{r}{\psi_1} = w$.
Then $M_1 = \lp{w}$, so $\td{E}{\psi_1} \steps^* \dsubst{\td{L}{\psi_1}}{w}{x}
\evals L_1$.
We know $\inper[\Psi_2]{\td{L_1}{\psi_2}}
{\td{\dsubst{\td{L}{\psi_1}}{w}{x}}{\psi_2}}{A_{12}}$, and want to show
$\inper[\Psi_2]{\td{L_1}{\psi_2}}
{\td{E}{\psi_1\psi_2}}{A_{12}}$.
We proceed by casing on $\td{r}{\psi_1\psi_2}$:
\begin{enumerate}
\item $\td{r}{\psi_1\psi_2} = \e$.
Then by $\inper[\Psi_2]{\lp{\td{w}{\psi_2}}}{\td{M}{\psi_1\psi_2}}{\C}$,
$\td{M}{\psi_1\psi_2}\evals\base$ and 
$\td{E}{\psi_1\psi_2} \steps^* \td{P}{\psi_1\psi_2}$.
But $\td{\dsubst{\td{L}{\psi_1}}{w}{x}}{\psi_2}
= \td{\dsubst{L}{\e}{x}}{\psi_1\psi_2}$ and
$\inper[\Psi_2]{\td{\dsubst{L}{\e}{x}}{\psi_1\psi_2}}{\td{P}{\psi_1\psi_2}}{A_{12}}$
by the hypothesis $\ceqtm{\dsubst{L}{\e}{x}}{P}{A}$, so the result follows by
transitivity.

\item $\td{r}{\psi_1\psi_2} = w'$.
By $\inper[\Psi_2]{\lp{\td{w}{\psi_2}}}{\td{M}{\psi_1\psi_2}}{\C}$,
$\td{M}{\psi_1\psi_2}\evals\lp{w'}$ and 
$\td{E}{\psi_1\psi_2} \steps^* \dsubst{\td{L}{\psi_1\psi_2}}{w'}{x}$.
But $\td{\dsubst{\td{L}{\psi_1}}{w}{x}}{\psi_2}
= \dsubst{\td{L}{\psi_1\psi_2}}{w'}{x}$
so the result again follows by transitivity.
\end{enumerate}
\end{enumerate}

\paragraph{Kan}
Show $\cpretype{\C}$ is Kan.

This proof is identical to the proof that $\cpretype{\bool}$ is Kan, because the
relevant portions of the operational semantics and the definition of
$\vinper[-]{-}{-}{\C}$ are identical.

\paragraph{Cubical}
Show for any $\Psi'$ and $\vinper[\Psi']{M}{N}{\C}$,
$\ceqtm[\Psi']{M}{N}{\C}$.

We consider each case of $\vinper[\Psi']{M}{N}{\C}$.
For $\base$ and $\lp{x}$, this follows from the introduction rules already
proven. For $\hcomgeneric{x}{\C}$, this follows from the first Kan condition for
$\C$, again already proven.

%%%%%%%%%%%%%%%%%%%%%%%%%%%%%%%%%%%%%%%%%%%%%%%%%%%%%%%%%%%%%%%%%%%%%%%%%%%%%%%%
\subsection{Products}

When $\cwftype{A}$ and $\cwftype{B}$ we define $\per{\prd{A}{B}}$ as follows:
\[ \vinper{\pair{M}{N}}{\pair{M'}{N'}}{\prd{A}{B}} \]
when $\ceqtm{M}{M'}{A}$ and $\ceqtm{N}{N'}{B}$.

\paragraph{Pretype}
If $\cwftype{A}$ and $\cwftype{B}$ then $\cpretype{\prd{A}{B}}$.

For any $\msubsts{\Psi_1}{\psi_1}{\Psi}$ and $\msubsts{\Psi_2}{\psi_2}{\Psi_1}$,
$\wfval[\Psi_1]{\prd{\td{A}{\psi_1}}{\td{B}{\psi_1}}}$ and
$\wfval[\Psi_2]{\prd{\td{A}{\psi_1\psi_2}}{\td{B}{\psi_1\psi_2}}}$. Since
$\cwftype[\Psi_2]{\td{A}{\psi_1\psi_2}}$ and
$\cwftype[\Psi_2]{\td{B}{\psi_1\psi_2}}$, we have
$\per[\Psi_2]{\prd{\td{A}{\psi_1\psi_2}}{\td{B}{\psi_1\psi_2}}}$.

\paragraph{Introduction}
If $\cwftype{A}$, $\cwftype{B}$,
$\ceqtm{M}{M'}{A}$, and $\ceqtm{N}{N'}{B}$, then
$\ceqtm{\pair{M}{N}}{\pair{M'}{N'}}{\prd{A}{B}}$.

Since for any $\msubst{\psi}$,
$\wfval[\Psi']{\pair{\td{M}{\psi}}{\td{N}{\psi}}}$, 
each side has coherent aspects up to syntactic equality. 
Thus it suffices to show
$\vinper[\Psi_2]
{\pair{\td{M}{\psi_1\psi_2}}{\td{N}{\psi_1\psi_2}}}
{\pair{\td{M'}{\psi_1\psi_2}}{\td{N'}{\psi_1\psi_2}}}
{\prd{\td{A}{\psi_1\psi_2}}{\td{B}{\psi_1\psi_2}}}$.
But this is true because
$\ceqtm[\Psi_2]{\td{M}{\psi_1\psi_2}}{\td{M'}{\psi_1\psi_2}}
{\td{A}{\psi_1\psi_2}}$ and similarly for $N$.

\paragraph{Elimination}
If $\cwftype{A}$, $\cwftype{B}$, and
$\ceqtm{P}{P'}{\prd{A}{B}}$, then
$\ceqtm{\fst{P}}{\fst{P'}}{A}$ and
$\ceqtm{\snd{P}}{\snd{P'}}{B}$.

For any $\msubsts{\Psi_1}{\psi_1}{\Psi}$ and $\msubsts{\Psi_2}{\psi_2}{\Psi_1}$,
we know $\td{P}{\psi_1}\evals P_1$ and
$\vinper[\Psi_1]{P_1}{P_1}{\prd{\td{A}{\psi_1}}{\td{B}{\psi_1}}}$, so
$P_1 = \pair{M}{N}$ where
$\coftype[\Psi_1]{M}{\td{A}{\psi_1}}$ and
$\coftype[\Psi_1]{N}{\td{B}{\psi_1}}$.
Thus $\fst{\td{P}{\psi_1}}\steps^* \fst{P_1} = \fst{\pair{M}{N}} \steps M \evals
M_1$ and $\td{M_1}{\psi_2} \evals M_2$, where
$\inper[\Psi_2]{M_2}{\td{M}{\psi_2}}{A_{12}}$.

We also know
$\td{P}{\psi_1\psi_2}\evals P_{12}$ where
$\inper[\Psi_2]
{\td{P_1}{\psi_2} = \pair{\td{M}{\psi_2}}{\td{N}{\psi_2}}}
{P_{12}}
{\prd{\td{A}{\psi_1\psi_2}}{\td{B}{\psi_1\psi_2}}}$, so
$P_{12} = \pair{O}{Q}$ where
$\ceqtm[\Psi_2]{\td{M}{\psi_2}}{O}{\td{A}{\psi_1\psi_2}}$ and
$\ceqtm[\Psi_2]{\td{N}{\psi_2}}{Q}{\td{B}{\psi_1\psi_2}}$.
Then $\fst{\td{P}{\psi_1\psi_2}}\steps^* \fst{P_{12}}
= \fst{\pair{O}{Q}} \steps O$, and we want to show
$\inper[\Psi_2]{M_2}{O}{A_{12}}$.
It suffices to show $\inper[\Psi_2]{\td{M}{\psi_2}}{O}{A_{12}}$,
which follows directly from the above equality.

By a similar argument, $\fst{P'}$ also has coherent aspects. That the aspects of
$\fst{P}$ and $\fst{P'}$ are themselves $\vinper[\Psi_2]{-}{-}{A_{12}}$
follows from
$\ceqtm[\Psi_1]{M}{M'}{\td{A}{\psi_1}}$ where
$\td{P'}{\psi_1} \evals \pair{M'}{N'}$ and 
$\ceqtm[\Psi_2]{O}{O'}{\td{A}{\psi_1\psi_2}}$ where
$\td{P'}{\psi_1\psi_2} \evals \pair{O'}{Q'}$.
The argument for $\snd{-}$ is analogous.

\paragraph{Computation}
If $\cwftype{A}$, $\cwftype{B}$,
$\coftype{M}{A}$, and $\coftype{N}{B}$, then
$\ceqtm{\fst{\pair{M}{N}}}{M}{A}$ and
$\ceqtm{\snd{\pair{M}{N}}}{N}{B}$.

These follow by head expansion, since $\ceqtm{M}{M}{A}$ and for all $\psi$,
$\fst{\pair{\td{M}{\psi}}{\td{N}{\psi}}} \steps \td{M}{\psi}$, and the same for
$N$.

\paragraph{Eta}
If $\cwftype{A}$, $\cwftype{B}$, and
$\coftype{P}{\prd{A}{B}}$, then
$\ceqtm{P}{\pair{\fst{P}}{\snd{P}}}{\prd{A}{B}}$.

By the elimination and introduction rules for products, we already know that
$\coftype{\pair{\fst{P}}{\snd{P}}}{\prd{A}{B}}$. Thus by
\cref{lem:coftype-ceqtm} it suffices to show that for any $\msubst{\psi}$, 
$\inper[\Psi']{\td{P}{\psi}}{\pair{\fst{\td{P}{\psi}}}{\snd{\td{P}{\psi}}}}
{\prd{\td{A}{\psi}}{\td{B}{\psi}}}$.
By $\coftype{P}{\prd{A}{B}}$, we know that $\td{P}{\psi}\evals\pair{M}{N}$ where
$\coftype[\Psi']{M}{\td{A}{\psi}}$ and $\coftype[\Psi']{N}{\td{B}{\psi}}$.
Therefore we must show
$\vinper[\Psi']{\pair{M}{N}}{\pair{\fst{\td{P}{\psi}}}{\snd{\td{P}{\psi}}}}
{\prd{\td{A}{\psi}}{\td{B}{\psi}}}$, which requires showing that
$\ceqtm[\Psi']{M}{\fst{\td{P}{\psi}}}{\td{A}{\psi}}$ and
$\ceqtm[\Psi']{N}{\snd{\td{P}{\psi}}}{\td{B}{\psi}}$.

Again by \cref{lem:coftype-ceqtm}, it suffices to show that for any
$\msubsts{\Psi''}{\psi'}{\Psi'}$, 
$\inper[\Psi'']{\td{M}{\psi'}}{\fst{\td{P}{\psi\psi'}}}{A'}$ and
$\inper[\Psi'']{\td{N}{\psi'}}{\snd{\td{P}{\psi\psi'}}}{B'}$
where $\td{A}{\psi\psi'}\evals A'$ and $\td{B}{\psi\psi'}\evals B'$.
By $\coftype{P}{\prd{A}{B}}$, we know that
$\td{P}{\psi\psi'}\evals\pair{M'}{N'}$; by coherence of aspects,
$\inper[\Psi'']{\pair{M'}{N'}}{\pair{\td{M}{\psi'}}{\td{N}{\psi'}}}
{\prd{\td{A}{\psi\psi'}}{\td{B}{\psi\psi'}}}$, and thus
$\ceqtm[\Psi'']{M'}{\td{M}{\psi'}}{\td{A}{\psi\psi'}}$ and
$\ceqtm[\Psi'']{N'}{\td{N}{\psi'}}{\td{B}{\psi\psi'}}$.
But then 
$\fst{\td{P}{\psi\psi'}} \steps^* \fst{\pair{M'}{N'}} \steps M'$ and
$\snd{\td{P}{\psi\psi'}} \steps^* N'$, and the relations
$\inper[\Psi'']{\td{M}{\psi'}}{M'}{A'}$ and
$\inper[\Psi'']{\td{N}{\psi'}}{N'}{B'}$
follow from the corresponding $\eq$ equalities.

\paragraph{Kan}
If $\cwftype{A}$ and $\cwftype{B}$, then $\cpretype{\prd{A}{B}}$ is Kan.

The first Kan condition asserts that for any
$\msubsts{(\Psi',x)}{\psi}{\Psi}$, if
\begin{enumerate}
\item $\ceqtm[\Psi',x]{M}{O}{\prd{\td{A}{\psi}}{\td{B}{\psi}}}$,
\item $\ceqtm[\Psi',y]{\dsubst{N^\e}{\e}{x}}{\dsubst{P^\e}{\e}{x}}
{\prd{\dsubst{\td{A}{\psi}}{\e}{x}}{\dsubst{\td{B}{\psi}}{\e}{x}}}$ for
$\e=0,1$, and
\item $\ceqtm[\Psi']{\dsubst{\dsubst{N^\e}{r}{y}}{\e}{x}}{\dsubst{M}{\e}{x}}
{\prd{\dsubst{\td{A}{\psi}}{\e}{x}}{\dsubst{\td{B}{\psi}}{\e}{x}}}$
for $\e=0,1$,
\end{enumerate}
then $\ceqtm[\Psi',x]
{\hcomgeneric{x}{\prd{\td{A}{\psi}}{\td{B}{\psi}}}}
{\hcom{x}{\prd{\td{A}{\psi}}{\td{B}{\psi}}}{r}{r'}{O}{y.P^0,y.P^1}}
{\prd{\td{A}{\psi}}{\td{B}{\psi}}}$.

By head expansion on both sides, it suffices to show that
\begin{gather*}
\pair
{\hcom{x}{\td{A}{\psi}}{r}{r'}{\fst{M}}{y.\fst{N^0},y.\fst{N^1}}}
{\hcom{x}{\td{B}{\psi}}{r}{r'}{\snd{M}}{y.\snd{N^0},y.\snd{N^1}}},
\\
\pair
{\hcom{x}{\td{A}{\psi}}{r}{r'}{\fst{O}}{y.\fst{P^0},y.\fst{P^1}}}
{\hcom{x}{\td{B}{\psi}}{r}{r'}{\snd{O}}{y.\snd{P^0},y.\snd{P^1}}}
\end{gather*}
are $\ceqtm[\Psi',x]{-}{-}{\prd{\td{A}{\psi}}{\td{B}{\psi}}}$.
By the introduction rule for products, it suffices to show that the components
of these pairs are $\ceqtm[\Psi',x]{-}{-}{\td{A}{\psi}}$ and
$\ceqtm[\Psi',x]{-}{-}{\td{B}{\psi}}$ respectively.
But these follow from the first Kan conditions of $\cwftype{A}$ and
$\cwftype{B}$, with the elimination rules for products applied to the hypotheses
of this Kan condition (using transitivity of $\eq$ to get the adjacency
condition for $O,P^\e$).

The second Kan condition asserts that for any
$\msubsts{(\Psi',x)}{\psi}{\Psi}$, if
\begin{enumerate}
\item $\coftype[\Psi',x]{M}{\prd{\td{A}{\psi}}{\td{B}{\psi}}}$,
\item $\coftype[\Psi',y]{\dsubst{N^\e}{\e}{x}}
{\prd{\dsubst{\td{A}{\psi}}{\e}{x}}{\dsubst{\td{B}{\psi}}{\e}{x}}}$ for
$\e=0,1$, and
\item $\ceqtm[\Psi']{\dsubst{\dsubst{N^\e}{r}{y}}{\e}{x}}{\dsubst{M}{\e}{x}}
{\prd{\dsubst{\td{A}{\psi}}{\e}{x}}{\dsubst{\td{B}{\psi}}{\e}{x}}}$
for $\e=0,1$,
\end{enumerate}
then $\ceqtm[\Psi',x]{\hcom{x}{\prd{\td{A}{\psi}}{\td{B}{\psi}}}{r}{r}{M}{y.N^0,y.N^1}}{M}{\prd{\td{A}{\psi}}{\td{B}{\psi}}}$.

By head expansion, it suffices to show
\begin{gather*}
\pair{\hcom{x}{\td{A}{\psi}}{r}{r}{\fst{M}}{y.\fst{N^0},y.\fst{N^1}}}
      {\hcom{x}{\td{B}{\psi}}{r}{r}{\snd{M}}{y.\snd{N^0},y.\snd{N^1}}} \\
\ceqtm[\Psi',x]{}{M}{\prd{\td{A}{\psi}}{\td{B}{\psi}}}
\end{gather*}
By the introduction and elimination rules for products and the second 
Kan conditions of $\cwftype{A}$ and $\cwftype{B}$, the pair above is
$\ceqtm[\Psi',x]{-}{\pair{\fst{M}}{\snd{M}}}{\prd{\td{A}{\psi}}{\td{B}{\psi}}}$.
The result follows from the eta rule for products.

The third Kan condition asserts that for any $\msubst{\psi}$, if
\begin{enumerate}
\item \coftype[\Psi']{M}{\prd{\td{A}{\psi}}{\td{B}{\psi}}},
\item \coftype[\Psi',y]{N^\e}{\prd{\td{A}{\psi}}{\td{B}{\psi}}}, and
\item
\ceqtm[\Psi']{\dsubst{N^\e}{r}{y}}{M}{\prd{\td{A}{\psi}}{\td{B}{\psi}}},
\end{enumerate}
then $\ceqtm[\Psi']{\hcomgeneric{\e}{\prd{\td{A}{\psi}}{\td{B}{\psi}}}}{\dsubst{N^\e}{r'}{y}}{\prd{\td{A}{\psi}}{\td{B}{\psi}}}$.

The proof is the same as for the second Kan condition, above, appealing instead
to the third Kan conditions of $\cwftype{A}$ and $\cwftype{B}$.

The fourth Kan condition asserts that for any
$\msubsts{(\Psi',x)}{\psi}{\Psi}$, if
$\ceqtm[\Psi']{M}{N}
{\prd{\dsubst{\td{A}{\psi}}{r}{x}}{\dsubst{\td{B}{\psi}}{r}{x}}}$, then
$\ceqtm[\Psi']{{\coe{x.\prd{\td{A}{\psi}}{\td{B}{\psi}}}{r}{r'}{M}}}{{\coe{x.\prd{\td{A}{\psi}}{\td{B}{\psi}}}{r}{r'}{N}}}
{\prd{\dsubst{\td{A}{\psi}}{r'}{x}}{\dsubst{\td{B}{\psi}}{r'}{x}}}$.

By head expansion on both sides, it suffices to show
\begin{gather*}
\pair{\coe{x.\td{A}{\psi}}{r}{r'}{\fst{M}}}
     {\coe{x.\td{B}{\psi}}{r}{r'}{\snd{M}}} \\
\ceqtm[\Psi']{}
{\pair{\coe{x.\td{A}{\psi}}{r}{r'}{\fst{N}}}
      {\coe{x.\td{B}{\psi}}{r}{r'}{\snd{N}}}}
{\prd{\dsubst{\td{A}{\psi}}{r'}{x}}{\dsubst{\td{B}{\psi}}{r'}{x}}}
\end{gather*}
By the introduction rule, it suffices to show
the components of these pairs are 
$\ceqtm[\Psi']{-}{-}{\dsubst{\td{A}{\psi}}{r'}{x}}$ and
$\ceqtm[\Psi']{-}{-}{\dsubst{\td{B}{\psi}}{r'}{x}}$ respectively.
But these follow from the elimination rule for products and the 
fourth Kan conditions of $\cwftype{A}$ and $\cwftype{B}$.

\paragraph{Cubical}
If $\cwftype{A}$, $\cwftype{B}$, $\msubst{\psi}$, and
$\vinper[\Psi']{P}{P'}{\prd{\td{A}{\psi}}{\td{B}{\psi}}}$, then
$\ceqtm[\Psi']{P}{P'}{\prd{\td{A}{\psi}}{\td{B}{\psi}}}$.

Then $P = \pair{M}{N}$, $P' = \pair{M'}{N'}$,
$\ceqtm[\Psi']{M}{M'}{A}$ and $\ceqtm[\Psi']{N}{N'}{B}$,
and the result follows from the introduction rule for products.

%%%%%%%%%%%%%%%%%%%%%%%%%%%%%%%%%%%%%%%%%%%%%%%%%%%%%%%%%%%%%%%%%%%%%%%%%%%%%%%%
\subsection{Functions}

When $\cwftype{A}$ and $\cwftype{B}$ we define $\per{\arr{A}{B}}$ as follows:
\[ \vinper{\lam{a}{M}}{\lam{a}{M'}}{\arr{A}{B}} \]
when $\eqtm{\oft{a}{A}}{M}{M'}{B}$.

\paragraph{Pretype}
If $\cwftype{A}$ and $\cwftype{B}$ then $\cpretype{\arr{A}{B}}$.

For any $\msubsts{\Psi_1}{\psi_1}{\Psi}$ and $\msubsts{\Psi_2}{\psi_2}{\Psi_1}$,
$\wfval[\Psi_1]{\arr{\td{A}{\psi_1}}{\td{B}{\psi_1}}}$ and
$\wfval[\Psi_2]{\arr{\td{A}{\psi_1\psi_2}}{\td{B}{\psi_1\psi_2}}}$. Since
$\cwftype[\Psi_2]{\td{A}{\psi_1\psi_2}}$ and
$\cwftype[\Psi_2]{\td{B}{\psi_1\psi_2}}$, we have
$\per[\Psi_2]{\arr{\td{A}{\psi_1\psi_2}}{\td{B}{\psi_1\psi_2}}}$.

\paragraph{Introduction}
If $\cwftype{A}$, $\cwftype{B}$, and $\eqtm{\oft aA}{M}{M'}{B}$, then
$\ceqtm{\lam{a}{M}}{\lam{a}{M'}}{\arr{A}{B}}$.

Each side has coherent aspects up to syntactic equality, since
$\wfval[\Psi']{\lam{a}{\td{M}{\psi}}}$ for all $\msubst{\psi}$.
Thus it suffices to show
$\vinper[\Psi_2]
{\lam{a}{\td{M}{\psi_1\psi_2}}}
{\lam{a}{\td{M'}{\psi_1\psi_2}}}
{\arr{\td{A}{\psi_1\psi_2}}{\td{B}{\psi_1\psi_2}}}$,
which holds because
$\eqtm[\Psi_2]{\oft{a}{\td{A}{\psi_1\psi_2}}}
{\td{M}{\psi_1\psi_2}}{\td{M'}{\psi_1\psi_2}}{\td{B}{\psi_1\psi_2}}$.

\paragraph{Elimination}
If $\cwftype{A}$, $\cwftype{B}$,
$\ceqtm{F}{F'}{\arr{A}{B}}$, and $\ceqtm{N}{N'}{A}$, then
$\ceqtm{\app{F}{N}}{\app{F'}{N'}}{B}$.

For any $\msubsts{\Psi_1}{\psi_1}{\Psi}$ and $\msubsts{\Psi_2}{\psi_2}{\Psi_1}$,
by $\coftype{F}{\arr{A}{B}}$ we know
$\td{F}{\psi_1}\evals F_1 = \lam{a}{M_1}$ and
$\oftype[\Psi_1]{\oft{a}{\td{A}{\psi}}}{M_1}{\td{B}{\psi}}$.
Thus $\app{\td{F}{\psi_1}}{\td{N}{\psi_1}} \steps^*
\app{\lam{a}{M_1}}{\td{N}{\psi_1}} \steps
\subst{M_1}{\td{N}{\psi_1}}{a}$. Since
$\coftype[\Psi_1]{\td{N}{\psi_1}}{\td{A}{\psi_1}}$, we have
$\coftype[\Psi_1]{\subst{M_1}{\td{N}{\psi_1}}{a}}{\td{B}{\psi_1}}$, so
$\subst{M_1}{\td{N}{\psi_1}}{a} \evals X_1$ and
$\inper[\Psi_2]{\td{X_1}{\psi_2}}
{\subst{\td{M_1}{\psi_2}}{\td{N}{\psi_1\psi_2}}{a}}{B_{12}}$.

We also know
$\td{F}{\psi_1\psi_2}\evals F_{12}$ where
$\inper[\Psi_2]
{\td{F_1}{\psi_2} = \lam{a}{\td{M_1}{\psi_2}}}
{F_{12}}
{\arr{\td{A}{\psi_1\psi_2}}{\td{B}{\psi_1\psi_2}}}$, so
$F_{12} = \lam{a}{M_{12}}$ and
$\eqtm[\Psi_2]{\oft{a}{\td{A}{\psi_1\psi_2}}}
{\td{M_1}{\psi_2}}{M_{12}}{\td{B}{\psi_1\psi_2}}$.
Then $\app{\td{F}{\psi_1\psi_2}}{\td{N}{\psi_1\psi_2}} \steps^*
\app{\lam{a}{M_{12}}}{\td{N}{\psi_1\psi_2}} \steps
\subst{M_{12}}{\td{N}{\psi_1\psi_2}}{a}$.
We want to show
$\inper[\Psi_2]{\td{X_1}{\psi_2}}
{\subst{M_{12}}{\td{N}{\psi_1\psi_2}}{a}}{B_{12}}$;
by the above $\eq$,
$\inper[\Psi_2]
{\subst{\td{M_1}{\psi_2}}{\td{N}{\psi_1\psi_2}}{a}}
{\subst{M_{12}}{\td{N}{\psi_1\psi_2}}{a}}{B_{12}}$,
so the result follows by transitivity.

By a symmetric argument, $\app{F'}{N'}$ also has coherent aspects.
To see that the aspects of $\app{F}{N}$ and $\app{F'}{N'}$ are themselves
$\vinper[\Psi_2]{-}{-}{B_{12}}$, observe that
$\td{F'}{\psi_1\psi_2} \evals \lam{a}{M'_{12}}$ such that
$\eqtm[\Psi_2]{\oft{a}{\td{A}{\psi_1\psi_2}}}
{M_{12}}{M'_{12}}{\td{B}{\psi_1\psi_2}}$, so
$\app{\td{F'}{\psi_1\psi_2}}{\td{N'}{\psi_1\psi_2}} \steps^*
\subst{M'_{12}}{\td{N'}{\psi_1\psi_2}}{a}$ and
$\ceqtm[\Psi_2]
{\subst{M_{12}}{\td{N}{\psi_1\psi_2}}{a}}
{\subst{M'_{12}}{\td{N'}{\psi_1\psi_2}}{a}}
{\td{B}{\psi_1\psi_2}}$.

\paragraph{Computation}
If $\cwftype{A}$, $\cwftype{B}$,
$\oftype{\oft aA}{M}{B}$, and $\coftype{N}{A}$, then
$\ceqtm{\app{\lam{a}{M}}{N}}{\subst{M}{N}{a}}{B}$.

That $\coftype{\subst{M}{N}{a}}{B}$ follows from the definition of
$\oftype{\oft aA}{M}{B}$, and the desired equality follows by head expansion.

\paragraph{Eta}
If $\cwftype{A}$, $\cwftype{B}$, and $\coftype{F}{\arr{A}{B}}$, then
$\ceqtm{F}{\lam{a}{\app{F}{a}}}{\arr{A}{B}}$.

We first prove that the right-hand side has this type, and then apply
\cref{lem:coftype-ceqtm}.

\begin{lemma}
If $\cwftype{A}$, $\cwftype{B}$, and $\coftype{F}{\arr{A}{B}}$, then
$\coftype{\lam{a}{\app{F}{a}}}{\arr{A}{B}}$.
\end{lemma}
\begin{proof}
By the introduction rule for functions, it suffices to show that for any
$\msubst{\psi}$ and $\ceqtm[\Psi']{N}{N'}{\td{A}{\psi}}$, 
$\ceqtm[\Psi']{\app{\td{F}{\psi}}{N}}{\app{\td{F}{\psi}}{N'}}{\arr{A}{B}}$. But
this follows from the elimination rule.
\end{proof}

By \cref{lem:coftype-ceqtm}, we must show that for any $\msubst{\psi}$,
$\inper[\Psi']{\td{F}{\psi}}{\lam{a}{\app{\td{F}{\psi}}{a}}}
{\arr{\td{A}{\psi}}{\td{B}{\psi}}}$.
We know $\td{F}{\psi} \evals \lam{a}{M}$ and
$\oftype[\Psi']{\oft{a}{\td{A}{\psi}}}{M}{\td{B}{\psi}}$, and must show
$\eqtm[\Psi']{\oft{a}{\td{A}{\psi}}}{M}{\app{\td{F}{\psi}}{a}}{\td{B}{\psi}}$.
That is, for any $\msubsts{\Psi''}{\psi'}{\Psi'}$ and
$\ceqtm[\Psi'']{N}{N'}{\td{A}{\psi\psi'}}$, 
$\ceqtm[\Psi'']{\subst{\td{M}{\psi'}}{N}{a}}
{\app{\td{F}{\psi\psi'}}{N'}}{\td{B}{\psi\psi'}}$.
We know both sides have this type, so again by \cref{lem:coftype-ceqtm},
it suffices to show that for any $\msubsts{\Psi'''}{\psi''}{\Psi''}$,
$\inper[\Psi''']{\subst{\td{M}{\psi'\psi''}}{\td{N}{\psi''}}{a}}
{\app{\td{F}{\psi\psi'\psi''}}{\td{N'}{\psi''}}}{B''}$
where $\td{B}{\psi\psi'\psi''}\evals B''$.

By $\coftype{F}{\arr{A}{B}}$, 
$\td{F}{\psi'\psi''}\evals \lam{a}{M''}$ and
$\vinper[\Psi''']{\lam{a}{\td{M}{\psi'\psi''}}}{\lam{a}{M''}}
{\arr{\td{A}{\psi\psi'\psi''}}{\td{B}{\psi\psi'\psi''}}}$ such that
$\eqtm[\Psi''']{\oft{a}{\td{A}{\psi\psi'\psi''}}}{\td{M}{\psi'\psi''}}{M''}
{\td{B}{\psi\psi'\psi''}}$.
Then $\app{\td{F}{\psi\psi'\psi''}}{\td{N'}{\psi''}} \steps^*
\app{\lam{a}{M''}}{\td{N'}{\psi''}} \steps
\subst{M''}{\td{N'}{\psi''}}{a}$, and from the above $\eq$ we deduce that
$\inper[\Psi''']{\subst{\td{M}{\psi'\psi''}}{\td{N}{\psi''}}{a}}
{\subst{M''}{\td{N'}{\psi''}}{a}}{B''}$.

\paragraph{Kan}
If $\cwftype{A}$ and $\cwftype{B}$, then $\cpretype{\arr{A}{B}}$ is Kan.

The first Kan condition asserts that for any
$\msubsts{(\Psi',x)}{\psi}{\Psi}$, if
\begin{enumerate}
\item $\ceqtm[\Psi',x]{M}{O}{\arr{\td{A}{\psi}}{\td{B}{\psi}}}$,
\item $\ceqtm[\Psi',y]{\dsubst{N^\e}{\e}{x}}{\dsubst{P^\e}{\e}{x}}
{\arr{\dsubst{\td{A}{\psi}}{\e}{x}}{\dsubst{\td{B}{\psi}}{\e}{x}}}$ for
$\e=0,1$, and
\item $\ceqtm[\Psi']{\dsubst{\dsubst{N^\e}{r}{y}}{\e}{x}}{\dsubst{M}{\e}{x}}
{\arr{\dsubst{\td{A}{\psi}}{\e}{x}}{\dsubst{\td{B}{\psi}}{\e}{x}}}$
for $\e=0,1$,
\end{enumerate}
then $\ceqtm[\Psi',x]
{\hcomgeneric{x}{\arr{\td{A}{\psi}}{\td{B}{\psi}}}}
{\hcom{x}{\arr{\td{A}{\psi}}{\td{B}{\psi}}}{r}{r'}{O}{y.P^0,y.P^1}}
{\arr{\td{A}{\psi}}{\td{B}{\psi}}}$.

By head expansion on both sides, it suffices to show that
\[\begin{aligned}
& \lam{a}{\hcom{x}{\td{B}{\psi}}{r}{r'}{\app{M}{a}}
{y.\app{N^0}{a},y.\app{N^1}{a}}} \\
\ceqtmtab[\Psi',x]{}
{\lam{a}{\hcom{x}{\td{B}{\psi}}{r}{r'}{\app{O}{a}}
{y.\app{P^0}{a},y.\app{P^1}{a}}}}
{\arr{\td{A}{\psi}}{\td{B}{\psi}}}
\end{aligned}\]
By the introduction rule for functions, it suffices to show
$\eqtm[\Psi',x]{\oft{a}{\td{A}{\psi}}}{-}{-}{\td{B}{\psi}}$
for these lambdas' bodies. That is, for any
$\msubsts{\Psi''}{\psi'}{(\Psi',x)}$ and
$\ceqtm[\Psi'']{Q}{Q'}{\td{A}{\psi\psi'}}$,
\[\begin{aligned}
& \hcom{\td{x}{\psi'}}{\td{B}{\psi\psi'}}{\td{r}{\psi'}}{\td{r'}{\psi'}}
{\app{\td{M}{\psi'}}{Q}}
{y.\app{\td{N^0}{\psi'}}{Q},y.\app{\td{N^1}{\psi'}}{Q}} \\
\ceqtmtab[\Psi'']{}
{\hcom{\td{x}{\psi'}}{\td{B}{\psi\psi'}}{\td{r}{\psi'}}{\td{r'}{\psi'}}
 {\app{\td{O}{\psi'}}{Q'}}
 {y.\app{\td{P^0}{\psi'}}{Q'},y.\app{\td{P^1}{\psi'}}{Q'}}}
{\td{B}{\psi\psi'}}
\end{aligned}\]

If $\td{x}{\psi'} = x'$ then $\Psi'' = (\Psi''',x')$ and the result follows from
the elimination rule for functions and the first Kan condition of $\cwftype{B}$.
Note that $Q$ might contain $x'$ and $\td{N^\e}{\psi'}$ might not make type
sense on arguments containing $x'$, because
$\coftype[\Psi''',x']{Q}{\td{A}{\psi\psi'}}$ and
$\coftype[\Psi''']{\dsubst{\td{N^\e}{\psi'}}{\e}{x'}}
{\arr{\dsubst{\td{A}{\psi\psi'}}{\e}{x'}}{\dsubst{\td{B}{\psi\psi'}}{\e}{x'}}}$.
But we only need
$\coftype[\Psi''']
{\app{\dsubst{\td{N^\e}{\psi'}}{\e}{x'}}{\dsubst{Q}{\e}{x'}}}
{\dsubst{\td{B}{\psi\psi'}}{\e}{x'}}$,
which follows from the elimination rule.

If $\td{x}{\psi'} = \e$ then by the elimination rule for functions,
the third Kan condition of $\cwftype{B}$,
and transitivity of $\eq$, it suffices to show
$\ceqtm[\Psi'']
{\app{\dsubst{\td{N^\e}{\psi'}}{\td{r'}{\psi'}}{y}}{Q}}
{\app{\dsubst{\td{P^\e}{\psi'}}{\td{r'}{\psi'}}{y}}{Q'}}
{\td{B}{\psi\psi'}}$,
which follows from our second hypothesis and the elimination rule for functions.

The second Kan condition asserts that for any
$\msubsts{(\Psi',x)}{\psi}{\Psi}$, if
\begin{enumerate}
\item $\coftype[\Psi',x]{M}{\arr{\td{A}{\psi}}{\td{B}{\psi}}}$,
\item $\coftype[\Psi',y]{\dsubst{N^\e}{\e}{x}}
{\arr{\dsubst{\td{A}{\psi}}{\e}{x}}{\dsubst{\td{B}{\psi}}{\e}{x}}}$ for
$\e=0,1$, and
\item $\ceqtm[\Psi']{\dsubst{\dsubst{N^\e}{r}{y}}{\e}{x}}{\dsubst{M}{\e}{x}}
{\arr{\dsubst{\td{A}{\psi}}{\e}{x}}{\dsubst{\td{B}{\psi}}{\e}{x}}}$
for $\e=0,1$,
\end{enumerate}
then $\ceqtm[\Psi',x]{\hcom{x}{\arr{\td{A}{\psi}}{\td{B}{\psi}}}{r}{r}{M}{y.N^0,y.N^1}}{M}{\arr{\td{A}{\psi}}{\td{B}{\psi}}}$.

By head expansion, it suffices to show
\[
\ceqtm[\Psi',x]
{\lam{a}{\hcom{x}{\td{B}{\psi}}{r}{r}
{\app{M}{a}}{y.\app{N^0}{a},y.\app{N^1}{a}}}}
{M}{\arr{\td{A}{\psi}}{\td{B}{\psi}}}
\]
By the eta and introduction rules for functions, it suffices to show that for
any $\msubsts{\Psi''}{\psi'}{\Psi'}$ and
$\ceqtm[\Psi'']{O}{O'}{\td{A}{\psi\psi'}}$,
\[
\ceqtm[\Psi'']
{\hcom{\td{x}{\psi'}}{\td{B}{\psi\psi'}}{\td{r}{\psi'}}{\td{r}{\psi'}}
{\app{\td{M}{\psi'}}{O}}{y.\app{\td{N^0}{\psi'}}{O},y.\app{\td{N^1}{\psi'}}{O}}}
{\app{\td{M}{\psi'}}{O'}}{\td{B}{\psi\psi'}}
\]
If $\td{x}{\psi'} = x'$ then this follows from the second Kan
condition of $\cwftype{B}$ and the elimination rule for functions.
If $\td{x}{\psi'} = \e$ then by the third Kan condition of $\cwftype{B}$, 
$\ceqtm[\Psi'']
{\hcomsym}
{\app{\dsubst{\td{N^\e}{\psi'}}{\td{r}{\psi'}}{y}}{O}}
{\td{B}{\psi\psi'}}$.
The result follows from the elimination rule and the fact that 
$\ceqtm[\Psi'']{\td{\dsubst{N^\e}{r}{y}}{\psi'}}{\td{M}{\psi'}}
{\arr{\td{A}{\psi\psi'}}{\td{B}{\psi\psi'}}}$.

The third Kan condition asserts that for any $\msubst{\psi}$, if
\begin{enumerate}
\item \coftype[\Psi']{M}{\arr{\td{A}{\psi}}{\td{B}{\psi}}},
\item \coftype[\Psi',y]{N^\e}{\arr{\td{A}{\psi}}{\td{B}{\psi}}}, and
\item
\ceqtm[\Psi']{\dsubst{N^\e}{r}{y}}{M}{\arr{\td{A}{\psi}}{\td{B}{\psi}}},
\end{enumerate}
then $\ceqtm[\Psi']{\hcomgeneric{\e}{\arr{\td{A}{\psi}}{\td{B}{\psi}}}}{\dsubst{N^\e}{r'}{y}}{\arr{\td{A}{\psi}}{\td{B}{\psi}}}$.

Again, by head expansion and the eta and introduction rules for functions, it
suffices to show that for any $\msubsts{\Psi''}{\psi'}{\Psi'}$ and
$\ceqtm[\Psi'']{O}{O'}{\td{A}{\psi\psi'}}$,
\[\begin{aligned}
&\hcom{\e}{\td{B}{\psi\psi'}}{\td{r}{\psi'}}{\td{r'}{\psi'}}
{\app{\td{M}{\psi'}}{O}}{y.\app{\td{N^0}{\psi'}}{O},y.\app{\td{N^1}{\psi'}}{O}}
\\
\ceqtmtab[\Psi'']{}
{\app{\td{\dsubst{N^\e}{r'}{y}}{\psi'}}{O'}}{\td{B}{\psi\psi'}}
\end{aligned}\]
This follows from the third Kan condition of $\cwftype{B}$ and the
elimination rule for functions.

The fourth Kan condition asserts that for any
$\msubsts{(\Psi',x)}{\psi}{\Psi}$, if
$\ceqtm[\Psi']{F}{F'}
{\arr{\dsubst{\td{A}{\psi}}{r}{x}}{\dsubst{\td{B}{\psi}}{r}{x}}}$, then
$\ceqtm[\Psi']{{\coe{x.\arr{\td{A}{\psi}}{\td{B}{\psi}}}{r}{r'}{F}}}{{\coe{x.\arr{\td{A}{\psi}}{\td{B}{\psi}}}{r}{r'}{F'}}}
{\arr{\dsubst{\td{A}{\psi}}{r'}{x}}{\dsubst{\td{B}{\psi}}{r'}{x}}}$.

By head expansion on both sides and the introduction rule for functions, it
suffices to show that for any $\msubsts{\Psi''}{\psi'}{\Psi'}$ and
$\ceqtm[\Psi'']{N}{N'}{\td{\dsubst{\td{A}{\psi}}{r'}{x}}{\psi'}}$,
\[
\ceqtm[\Psi']
{\coe{x.\td{B}{\psi\psi'}}{\td{r}{\psi'}}{\td{r'}{\psi'}}
 {\app{\td{F}{\psi'}}{\coe{x.\td{A}{\psi\psi'}}{\td{r'}{\psi'}}{\td{r}{\psi'}}{N}}}}
{\coe{x.\td{B}{\psi\psi'}}{\td{r}{\psi'}}{\td{r'}{\psi'}}
 {\app{\td{F'}{\psi'}}{\coe{x.\td{A}{\psi\psi'}}{\td{r'}{\psi'}}{\td{r}{\psi'}}{N'}}}}
{\td{\dsubst{\td{B}{\psi}}{r'}{x}}{\psi'}}
\]
This follows from the fourth Kan condition of $\cwftype{B}$ and the
elimination rule for functions.

\paragraph{Cubical}
If $\cwftype{A}$, $\cwftype{B}$, $\msubst{\psi}$, and
$\vinper[\Psi']{F}{F'}{\arr{\td{A}{\psi}}{\td{B}{\psi}}}$, then
$\ceqtm[\Psi']{F}{F'}{\arr{\td{A}{\psi}}{\td{B}{\psi}}}$.

Then $F = \lam{a}{M}$, $F' = \lam{a}{M'}$, and
$\eqtm[\Psi']{\oft{a}{\td{A}{\psi}}}{M}{M'}{\td{B}{\psi}}$,
and the result follows from the introduction rule for functions.

%%%%%%%%%%%%%%%%%%%%%%%%%%%%%%%%%%%%%%%%%%%%%%%%%%%%%%%%%%%%%%%%%%%%%%%%%%%%%%%%
\subsection{Not}

We define $\per[\Psi,x]{\notb{x}}$ as follows:
\[ \vinper[\Psi,x]{\notel{x}{M}}{\notel{x}{N}}{\notb{x}} \]
when $\ceqtm[\Psi,x]{M}{N}{\bool}$.

This type is somewhat unusual because it exists primarily to be coerced along
($\coe{x.\notb{x}}{r}{r'}{M}$), rather than to be introduced or eliminated in
the manner of function and product types. Accordingly, the bulk of this section
is dedicated to proving that $\cpretype[\Psi,x]{\notb{x}}$ is Kan.

The results in this section depend heavily on the following lemmas:
\begin{lemma}
If $\coftype{M}{\bool}$ then $\ceqtm{\notf{\notf{M}}}{M}{\bool}$.
\end{lemma}
\begin{proof}
Recalling that $\notf{M}$ is notation for $\ifb{\bool}{M}{\false}{\true}$,
we conclude from the introduction and elimination rules for booleans that
$\coftype{\notf{\notf{M}}}{\bool}$. Therefore each side has coherent aspects,
and it suffices to show that
$\inper[\Psi_2]{\notf{\notf{\td{M}{\psi_1\psi_2}}}}{\td{M}{\psi_1\psi_2}}
{\bool}$. We case on
$\inper[\Psi_2]{\td{M}{\psi_1\psi_2}}{\td{M}{\psi_1\psi_2}}{\bool}$:
\begin{enumerate}
\item $\td{M}{\psi_1\psi_2} \evals \true$.

Then $\notf{\notf{\td{M}{\psi_1\psi_2}}} \steps^*
\notf{\notf{\true}} \steps \notf{\false} \steps \true$, and
$\vinper[\Psi_2]{\true}{\true}{\bool}$.

\item $\td{M}{\psi_1\psi_2} \evals \false$.

Then $\notf{\notf{\td{M}{\psi_1\psi_2}}} \steps^*
\notf{\notf{\false}} \steps \notf{\true} \steps \false$, and
$\vinper[\Psi_2]{\false}{\false}{\bool}$.

\item $\td{M}{\psi_1\psi_2} \evals \hcom{x}{\bool}{r}{r'}{M'}{y.N^0,y.N^1}$
where $\Psi_2 = (\Psi',x)$, $r\neq r'$,
$\coftype[\Psi',x]{M'}{\bool}$,
$\coftype[\Psi',y]{\dsubst{N^\e}{\e}{x}}{\bool}$ for $\e=0,1$, and
$\ceqtm[\Psi']{\dsubst{\dsubst{N^\e}{r}{y}}{\e}{x}}{\dsubst{M'}{\e}{x}}{\bool}$
for $\e=0,1$.

Then
\[\begin{aligned}
\notf{\notf{\td{M}{\psi_1\psi_2}}} &\steps^*
\notf{\hcom{x}{\bool}{r}{r'}{\notf{M'}}{y.\notf{N^0},y.\notf{N^1}}}
\\ &\steps\hphantom{{}^*}
\hcom{x}{\bool}{r}{r'}{\notf{\notf{M'}}}
{y.\notf{\notf{N^0}},y.\notf{\notf{N^1}}}
\end{aligned}\]
which by the first Kan condition of $\bool$ is
$\inper[\Psi',x]{-}{\td{M}{\psi_1\psi_2}}{\bool}$ when
$\ceqtm[\Psi',x]{M'}{\notf{\notf{M'}}}{\bool}$ and
$\ceqtm[\Psi',y]{\dsubst{N^\e}{\e}{x}}
{\notf{\notf{\dsubst{N^\e}{\e}{x}}}}{\bool}$ for $\e=0,1$.
(The adjacency conditions for this $\hcomsym$ follow from the elimination rule
for booleans.) But these $\eq$ follow from the inductive hypothesis.
\qedhere
\end{enumerate}
\end{proof}

\begin{lemma}\label{lem:inper-not-swap}
If $\coftype{M}{\bool}$, $\coftype{N}{\bool}$, and for all $\msubst{\psi}$,
$\inper[\Psi']{\notf{\td{M}{\psi}}}{\td{N}{\psi}}{\bool}$, then
$\ceqtm{M}{\notf{N}}{\bool}$ (and in particular, $\inper{M}{\notf{N}}{\bool}$).
\end{lemma}
\begin{proof}
By \cref{lem:coftype-ceqtm}, $\ceqtm{\notf{M}}{N}{\bool}$. Then
$\ceqtm{\notf{\notf{M}}}{\notf{N}}{\bool}$, so $\ceqtm{M}{\notf{N}}{\bool}$.
\end{proof}

\paragraph{Pretype}
$\cpretype[\Psi,x]{\notb{x}}$ and $\ceqpretype{\notb{\e}}{\bool}$.

For the first part, there are three cases to consider.
For any $\msubsts{\Psi_1}{\psi_1}{\Psi}$ and $\msubsts{\Psi_2}{\psi_2}{\Psi_1}$,
\begin{enumerate}
\item If $\td{x}{\psi_1} = \e$ then
$\notb{\td{x}{\psi_1}}\evals\bool$,
$\td{\bool}{\psi_2}\evals\bool$, and
$\notb{\td{x}{\psi_1\psi_2}}\evals\bool$;

\item If $\td{x}{\psi_1} = x'$ and $\td{x'}{\psi_2} = \e$ then
$\notb{\td{x}{\psi_1}}\evals\notb{x'}$,
$\notb{\td{x'}{\psi_2}}\evals\bool$, and
$\notb{\td{x}{\psi_1\psi_2}}\evals\bool$; and

\item If $\td{x}{\psi_1} = x'$ and $\td{x'}{\psi_2} = x''$ then
$\notb{\td{x}{\psi_1}}\evals\notb{x'}$,
$\notb{\td{x'}{\psi_2}}\evals\notb{x''}$, and
$\notb{\td{x}{\psi_1\psi_2}}\evals\notb{x''}$.
\end{enumerate}
But $\per[\Psi_2]{\bool}$ and $\per[\Psi_2',x'']{\notb{x''}}$ where $\Psi_2 =
(\Psi_2',x'')$.

For the second part, $\notb{\td{\e}{\psi}} \evals \bool$ and
$\per[\Psi']{\bool}$ for any $\msubst{\psi}$.

\paragraph{Introduction}
If $\ceqtm{M}{N}{\bool}$, then $\ceqtm{\notel{r}{M}}{\notel{r}{N}}{\notb{r}}$.

Let $\msubsts{\Psi_1}{\psi_1}{\Psi}$ and $\msubsts{\Psi_2}{\psi_2}{\Psi_1}$.
\begin{enumerate}
\item If $\td{r}{\psi_1} = 0$ then
$\notel{0}{\td{M}{\psi_1}} \steps \notf{\td{M}{\psi_1}}$ and
$\notel{0}{\td{M}{\psi_1\psi_2}} \steps \notf{\td{M}{\psi_1\psi_2}}$. By
$\ceqtm[\Psi_1]{\notf{\td{M}{\psi_1}}}{\notf{\td{N}{\psi_1}}}{\bool}$ we know
$\notf{\td{M}{\psi_1}} \evals M_1'$, $\notf{\td{N}{\psi_1}} \evals N_1'$, and
$\inperfour[\Psi_2]
{\td{M_1'}{\psi_2}}{\td{N_1'}{\psi_2}}
{\notf{\td{M}{\psi_1\psi_2}}}{\notf{\td{N}{\psi_1\psi_2}}}{\bool}$, which is
what we wanted to show.

\item If $\td{r}{\psi_1} = 1$ then
$\notel{1}{\td{M}{\psi_1}} \steps \td{M}{\psi_1}$ and
$\notel{1}{\td{M}{\psi_1\psi_2}} \steps \td{M}{\psi_1\psi_2}$. Our assumption
directly implies $\td{M}{\psi_1} \evals M_1$, $\td{N}{\psi_1} \evals N_1$, and
$\inperfour[\Psi_2]
{\td{M_1}{\psi_2}}{\td{N_1}{\psi_2}}
{\td{M}{\psi_1\psi_2}}{\td{N}{\psi_1\psi_2}}{\bool}$.

\item If $\td{r}{\psi_1} = x$ and $\td{x}{\psi_2} = 0$ then
$\isval{\notel{x}{\td{M}{\psi_1}}}$ and
$\notel{0}{\td{M}{\psi_1\psi_2}} \steps \notf{\td{M}{\psi_1\psi_2}}$.
Then by
$\ceqtm[\Psi_2]{\notf{\td{M}{\psi_1\psi_2}}}{\notf{\td{N}{\psi_1\psi_2}}}{\bool}$
we conclude
$\inperfour[\Psi_2]
{\notf{\td{M}{\psi_1\psi_2}}}{\notf{\td{N}{\psi_1\psi_2}}}
{\notf{\td{M}{\psi_1\psi_2}}}{\notf{\td{N}{\psi_1\psi_2}}}{\bool}$.

\item If $\td{r}{\psi_1} = x$ and $\td{x}{\psi_2} = 1$ then
$\isval{\notel{x}{\td{M}{\psi_1}}}$ and
$\notel{1}{\td{M}{\psi_1\psi_2}} \steps \td{M}{\psi_1\psi_2}$.
By our assumption,
$\inperfour[\Psi_2]
{\td{M}{\psi_1\psi_2}}{\td{N}{\psi_1\psi_2}}
{\td{M}{\psi_1\psi_2}}{\td{N}{\psi_1\psi_2}}{\bool}$.

\item If $\td{r}{\psi_1} = x$ and $\td{x}{\psi_2} = x'$ then
$\isval{\notel{x}{\td{M}{\psi_1}}}$ and
$\isval{\notel{x'}{\td{M}{\psi_1\psi_2}}}$.
By 
$\ceqtm[\Psi_2]{\td{M}{\psi_1\psi_2}}{\td{N}{\psi_1\psi_2}}{\bool}$,
$\vinperfour[\Psi_2]
{\notel{x'}{\td{M}{\psi_1\psi_2}}}{\notel{x'}{\td{N}{\psi_1\psi_2}}}
{\notel{x'}{\td{M}{\psi_1\psi_2}}}{\notel{x'}{\td{N}{\psi_1\psi_2}}}
{\notb{x'}}$.
\end{enumerate}

\paragraph{Computation}
If $\coftype{M}{\bool}$, then
$\ceqtm{\coe{x.\notb{x}}{\e}{\e}{M}}{M}{\bool}$ and
$\ceqtm{\coe{x.\notb{x}}{\e}{\eb}{M}}{\notf{M}}{\bool}$.

These are immediate by head expansion.

\paragraph{Kan}
$\cpretype[\Psi,x]{\notb{x}}$ is Kan.

The operational semantics for $\hcomsym$ at $\notb{x}$ involve $\coesym$, so we
start by proving the fourth Kan condition, which asserts that for
any $\msubsts{(\Psi',x')}{\psi}{(\Psi,x)}$, if
$\ceqtm[\Psi']{M}{N}{\notb{\dsubst{\td{x}{\psi}}{r}{x'}}}$, then
$\ceqtm[\Psi']
{\coe{x'.\notb{\td{x}{\psi}}}{r}{r'}{M}}
{\coe{x'.\notb{\td{x}{\psi}}}{r}{r'}{N}}
{\notb{\dsubst{\td{x}{\psi}}{r'}{x'}}}$.

If $\td{x}{\psi} = \e$ then by head expansion and \cref{lem:ceqpretype-ceqtm}
it suffices to show
$\ceqtm[\Psi']{\coe{x'.\bool}{r}{r'}{M}}{\coe{x'.\bool}{r}{r'}{N}}{\bool}$,
which is the fourth Kan condition of $\bool$.
If $\td{x}{\psi} = y \neq x'$ then by head expansion it suffices to show
$\ceqtm[\Psi']{M}{N}{\notb{y}}$ when $\ceqtm[\Psi']{M}{N}{\notb{y}}$.
Otherwise, $\td{x}{\psi} = x'$, and we must show that if
$\ceqtm[\Psi']{M}{N}{\notb{r}}$ then
$\ceqtm[\Psi']
{\coe{x'.\notb{x'}}{r}{r'}{M}}
{\coe{x'.\notb{x'}}{r}{r'}{N}}
{\notb{r'}}$.
Establishing this requires a large case split; we focus on the unary version
because the binary one follows easily. Let
$\msubsts{\Psi_1}{\psi_1}{\Psi'}$ and $\msubsts{\Psi_2}{\psi_2}{\Psi_1}$.

\begin{enumerate}
\item
If $\td{r}{\psi_1} = \e$ and $\td{r'}{\psi_1} = \e$
then $\coe{x.\notb{x}}{\e}{\e}{\td{M}{\psi_1}} \steps
\td{M}{\psi_1} \evals M_1$,
and $\coe{x.\notb{x}}{\e}{\e}{\td{M}{\psi_1\psi_2}} \steps 
\td{M}{\psi_1\psi_2}$ where
$\inper[\Psi_2]{\td{M_1}{\psi_2}}{\td{M}{\psi_1\psi_2}}{\bool}$.

\item
If $\td{r}{\psi_1} = \e$ and $\td{r'}{\psi_1} = \eb$ then
$\coe{x.\notb{x}}{\e}{\eb}{\td{M}{\psi_1}} \steps
\notf{\td{M}{\psi_1}} \evals X_1$,
and $\coe{x.\notb{x}}{\e}{\eb}{\td{M}{\psi_1\psi_2}} \steps 
\notf{\td{M}{\psi_1\psi_2}}$. 
By \cref{lem:ceqpretype-ceqtm} we know $\coftype[\Psi_1]{\td{M}{\psi_1}}{\bool}$
so $\coftype[\Psi_1]{\notf{\td{M}{\psi_1}}}{\bool}$ and therefore
$\inper[\Psi_2]{\td{X_1}{\psi_2}}{\notf{\td{M}{\psi_1\psi_2}}}{\bool}$.

\item
If $\td{r}{\psi_1} = 1$ and $\td{r'}{\psi_1} = x$ then
$\coe{x.\notb{x}}{1}{x}{\td{M}{\psi_1}} \steps \notel{x}{\td{M}{\psi_1}}$.

\begin{enumerate}
\item
If $\td{x}{\psi_2} = 1$ then
$\notel{1}{\td{M}{\psi_1\psi_2}} \steps \td{M}{\psi_1\psi_2}$ and
$\coe{x.\notb{x}}{1}{1}{\td{M}{\psi_1\psi_2}} \steps \td{M}{\psi_1\psi_2}$,
where $\inper[\Psi_2]{\td{M}{\psi_1\psi_2}}{\td{M}{\psi_1\psi_2}}{\bool}$.

\item
If $\td{x}{\psi_2} = 0$ then
$\notel{0}{\td{M}{\psi_1\psi_2}} \steps \notf{\td{M}{\psi_1\psi_2}}$,
$\coe{x.\notb{x}}{1}{0}{\td{M}{\psi_1\psi_2}} \steps
\notf{\td{M}{\psi_1\psi_2}}$, and the result follows from
$\coftype[\Psi_2]{\td{M}{\psi_1\psi_2}}{\bool}$.

\item
If $\td{x}{\psi_2} = x'$ then
$\isval{\notel{x'}{\td{M}{\psi_1\psi_2}}}$ and
$\coe{x.\notb{x}}{1}{x'}{\td{M}{\psi_1\psi_2}} \steps
\notel{x'}{\td{M}{\psi_1\psi_2}}$, and so
$\vinper[\Psi_2]
{\notel{x'}{\td{M}{\psi_1\psi_2}}}
{\notel{x'}{\td{M}{\psi_1\psi_2}}}
{\notb{x'}}$
because $\coftype[\Psi_2]{\td{M}{\psi_1\psi_2}}{\bool}$.
\end{enumerate}

\item
If $\td{r}{\psi_1} = 0$ and $\td{r'}{\psi_1} = x$ then
$\coe{x.\notb{x}}{0}{x}{\td{M}{\psi_1}} \steps
\notel{x}{\notf{\td{M}{\psi_1}}}$.

\begin{enumerate}
\item
If $\td{x}{\psi_2} = 0$ then
$\notel{0}{\notf{\td{M}{\psi_1\psi_2}}} \steps
\notf{\notf{\td{M}{\psi_1\psi_2}}}$ and
$\coe{x.\notb{x}}{0}{0}{\td{M}{\psi_1\psi_2}} \steps \td{M}{\psi_1\psi_2}$.
By $\coftype[\Psi_2]{\td{M}{\psi_1\psi_2}}{\bool}$ we have
$\ceqtm[\Psi_2]
{\notf{\notf{\td{M}{\psi_1\psi_2}}}}
{\td{M}{\psi_1\psi_2}}{\bool}$ and in particular
$\inper[\Psi_2]
{\notf{\notf{\td{M}{\psi_1\psi_2}}}}
{\td{M}{\psi_1\psi_2}}{\bool}$.

\item
If $\td{x}{\psi_2} = 1$ then
$\notel{1}{\notf{\td{M}{\psi_1\psi_2}}} \steps
\notf{\td{M}{\psi_1\psi_2}}$,
$\coe{x.\notb{x}}{0}{1}{\td{M}{\psi_1\psi_2}} \steps
\notf{\td{M}{\psi_1\psi_2}}$, and the result follows from
$\coftype[\Psi_2]{\notf{\td{M}{\psi_1\psi_2}}}{\bool}$.

\item
If $\td{x}{\psi_2} = x'$ then
$\isval{\notel{x'}{\notf{\td{M}{\psi_1\psi_2}}}}$,
$\coe{x.\notb{x}}{0}{x'}{\td{M}{\psi_1\psi_2}} \steps
\notel{x'}{\notf{\td{M}{\psi_1\psi_2}}}$, and the result follows from
$\coftype[\Psi_2]{\notf{\td{M}{\psi_1\psi_2}}}{\bool}$.
\end{enumerate}

\item
If $\td{r}{\psi_1} = x$ and $\td{r'}{\psi_1} = 1$ then
$\coftype[\Psi_1]{\td{M}{\psi_1}}{\notb{x}}$ and so
$\td{M}{\psi_1} \evals \notel{x}{N}$ where $\coftype[\Psi_1]{N}{\bool}$.
Therefore
$\coe{x.\notb{x}}{x}{1}{\td{M}{\psi_1}} \steps^*
\coe{x.\notb{x}}{x}{1}{\notel{x}{N}} \steps^* N \evals N_0$.

\begin{enumerate}
\item
If $\td{x}{\psi_2} = 1$ then
$\coe{x.\notb{x}}{1}{1}{\td{M}{\psi_1\psi_2}} \steps \td{M}{\psi_1\psi_2}$.
By $\coftype[\Psi_1]{\td{M}{\psi_1}}{\notb{x}}$ we know that
$\inper[\Psi_2]{\notel{1}{\td{N}{\psi_2}} \steps
\td{N}{\psi_2}}{\td{M}{\psi_1\psi_2}}{\bool}$ and by
$\coftype[\Psi_1]{N}{\bool}$ we know
$\inper[\Psi_2]{\td{N}{\psi_2}}{\td{N_0}{\psi_2}}{\bool}$.
We conclude 
$\inper[\Psi_2]{\td{M}{\psi_1\psi_2}}{\td{N_0}{\psi_2}}{\bool}$ as desired.

\item
If $\td{x}{\psi_2} = 0$ then
$\coe{x.\notb{x}}{0}{1}{\td{M}{\psi_1\psi_2}} \steps
\notf{\td{M}{\psi_1\psi_2}}$, and we must show
$\inper[\Psi_2]{\notf{\td{M}{\psi_1\psi_2}}}{\td{N_0}{\psi_2}}{\bool}$.
By $\coftype[\Psi_1]{N}{\bool}$, it suffices to show
$\inper[\Psi_2]{\notf{\td{M}{\psi_1\psi_2}}}{\td{N}{\psi_2}}{\bool}$.
By $\coftype[\Psi_1]{\td{M}{\psi_1}}{\notb{x}}$
we know that for any $\msubsts{\Psi'}{\psi}{\Psi_2}$,
$\inper[\Psi']{\notel{0}{\td{N}{\psi_2\psi}} \steps
\notf{\td{N}{\psi_2\psi}}}{\td{M}{\psi_1\psi_2\psi}}{\bool}$,
and the result follows by \cref{lem:inper-not-swap}.

\item
If $\td{x}{\psi_2} = x'$ then
$\vinper[\Psi_2]
{\td{M}{\psi_1\psi_2} \evals \notel{x'}{N'}}
{\notel{x'}{\td{N}{\psi_2}}}{\notb{x'}}$ where
$\ceqtm[\Psi_2]{N'}{\td{N}{\psi_2}}{\bool}$, and
$\coe{x.\notb{x}}{x'}{1}{\td{M}{\psi_1\psi_2}} \steps^*
\coe{x.\notb{x}}{x'}{1}{\notel{x'}{N'}} \steps^* N'$. Then
$\inper[\Psi_2]{\td{N}{\psi_2}}{\td{N_0}{\psi_2}}{\bool}$ and
$\inper[\Psi_2]{N'}{\td{N}{\psi_2}}{\bool}$ so we have
$\inper[\Psi_2]{N'}{\td{N_0}{\psi_2}}{\bool}$.
\end{enumerate}

\item
If $\td{r}{\psi_1} = x$ and $\td{r'}{\psi_1} = 0$ then
$\coftype[\Psi_1]{\td{M}{\psi_1}}{\notb{x}}$ and so
$\td{M}{\psi_1} \evals \notel{x}{N}$ where $\coftype[\Psi_1]{N}{\bool}$.
Therefore
$\coe{x.\notb{x}}{x}{0}{\td{M}{\psi_1}} \steps^*
\coe{x.\notb{x}}{x}{0}{\notel{x}{N}} \steps^* \notf{N} \evals X_1$.

\begin{enumerate}
\item
If $\td{x}{\psi_2} = 0$ then
$\coe{x.\notb{x}}{0}{0}{\td{M}{\psi_1\psi_2}} \steps \td{M}{\psi_1\psi_2}$.
By $\coftype[\Psi_1]{\td{M}{\psi_1}}{\notb{x}}$ we know that
$\inper[\Psi_2]{\notel{0}{\td{N}{\psi_2}} \steps \notf{\td{N}{\psi_2}}}
{\td{M}{\psi_1\psi_2}}{\bool}$, and by
$\coftype[\Psi_1]{\notf{N}}{\bool}$ we know
$\inper[\Psi_2]{\td{X_1}{\psi_2}}{\notf{\td{N}{\psi_2}}}{\bool}$. Thus
$\inper[\Psi_2]{\td{M}{\psi_1\psi_2}}{\td{X_1}{\psi_2}}{\bool}$.

\item
If $\td{x}{\psi_2} = 1$ then
$\coe{x.\notb{x}}{1}{0}{\td{M}{\psi_1\psi_2}} \steps
\notf{\td{M}{\psi_1\psi_2}}$.
By $\coftype[\Psi_1]{\td{M}{\psi_1}}{\notb{x}}$ we know that for any
$\msubsts{\Psi'}{\psi}{\Psi_2}$,
$\inper[\Psi']{\notel{1}{\td{N}{\psi_2\psi}} \steps \td{N}{\psi_2\psi}}
{\td{M}{\psi_1\psi_2\psi}}{\bool}$, so by \cref{lem:coftype-ceqtm},
$\ceqtm[\Psi_2]{\td{N}{\psi_2}}{\td{M}{\psi_1\psi_2}}{\bool}$, and thus
$\inper[\Psi_2]{\notf{\td{N}{\psi_2}}}{\notf{\td{M}{\psi_1\psi_2}}}{\bool}$.
By $\coftype[\Psi_1]{\notf{N}}{\bool}$, we know
$\inper[\Psi_2]{\td{X_1}{\psi_2}}{\notf{\td{N}{\psi_2}}}{\bool}$. Therefore
$\inper[\Psi_2]{\td{X_1}{\psi_2}}{\notf{\td{M}{\psi_1\psi_2}}}{\bool}$.

\item
If $\td{x}{\psi_2} = x'$ then
$\vinper[\Psi_2]
{\td{M}{\psi_1\psi_2} \evals \notel{x'}{N'}}
{\notel{x'}{\td{N}{\psi_2}}}{\notb{x'}}$ where
$\ceqtm[\Psi_2]{N'}{\td{N}{\psi_2}}{\bool}$, and
$\coe{x.\notb{x}}{x'}{0}{\td{M}{\psi_1\psi_2}} \steps^*
\coe{x.\notb{x}}{x'}{0}{\notel{x'}{N'}} \steps^* \notf{N'}$. Then
$\inper[\Psi_2]{\notf{N'}}{\notf{\td{N}{\psi_2}}}{\bool}$, and by
$\coftype[\Psi_1]{\notf{N}}{\bool}$,
$\inper[\Psi_2]{\td{X_1}{\psi_2}}{\notf{\td{N}{\psi_2}}}{\bool}$, so
$\inper[\Psi_2]{\notf{N'}}{\td{X_1}{\psi_2}}{\bool}$.
\end{enumerate}

\item
If $\td{r}{\psi_1} = x$ and $\td{r'}{\psi_1} = y$ then
$\coftype[\Psi_1]{\td{M}{\psi_1}}{\notb{x}}$ and so
$\td{M}{\psi_1} \evals \notel{x}{N}$ where $\coftype[\Psi_1]{N}{\bool}$.
Therefore
$\coe{x.\notb{x}}{x}{y}{\td{M}{\psi_1}} \steps^*
\coe{x.\notb{x}}{x}{y}{\notel{x}{N}} \steps^* \notel{y}{N}$.

\begin{enumerate}
\item
If $\td{x}{\psi_2} = \e$ and $\td{y}{\psi_2} = \e$ then
$\coe{x.\notb{x}}{\e}{\e}{\td{M}{\psi_1\psi_2}} \steps
\td{M}{\psi_1\psi_2}$. By
$\coftype[\Psi_1]{\td{M}{\psi_1}}{\notb{x}}$, we have
$\inper[\Psi_2]{\notel{\e}{\td{N}{\psi_2}}}{\td{M}{\psi_1\psi_2}}{\bool}$ as
desired.

\item
If $\td{x}{\psi_2} = 0$ and $\td{y}{\psi_2} = 1$ then
$\coe{x.\notb{x}}{0}{1}{\td{M}{\psi_1\psi_2}} \steps
\notf{\td{M}{\psi_1\psi_2}}$, and
$\notel{1}{\td{N}{\psi_2}} \steps \td{N}{\psi_2}$. By
$\coftype[\Psi_1]{\td{M}{\psi_1}}{\notb{x}}$, for any
$\msubsts{\Psi'}{\psi}{\Psi_2}$,
$\inper[\Psi']{\notel{0}{\td{N}{\psi_2\psi}} \steps \notf{\td{N}{\psi_2\psi}}}
{\td{M}{\psi_1\psi_2\psi}}{\bool}$, so by \cref{lem:inper-not-swap} we have
$\inper[\Psi_2]{\td{N}{\psi_2}}{\notf{\td{M}{\psi_1\psi_2}}}{\bool}$ as desired.

\item
If $\td{x}{\psi_2} = 1$ and $\td{y}{\psi_2} = 0$ then
$\coe{x.\notb{x}}{1}{0}{\td{M}{\psi_1\psi_2}} \steps
\notf{\td{M}{\psi_1\psi_2}}$, and
$\notel{0}{\td{N}{\psi_2}} \steps \notf{\td{N}{\psi_2}}$. By
$\coftype[\Psi_1]{\td{M}{\psi_1}}{\notb{x}}$, for any
$\msubsts{\Psi'}{\psi}{\Psi_2}$,
$\inper[\Psi']{\notel{1}{\td{N}{\psi_2\psi}} \steps \td{N}{\psi_2\psi}}
{\td{M}{\psi_1\psi_2\psi}}{\bool}$, so by \cref{lem:coftype-ceqtm}
$\ceqtm[\Psi_2]{\td{N}{\psi_2}}{\td{M}{\psi_1\psi_2}}{\bool}$ and so
$\inper[\Psi_2]{\notf{\td{N}{\psi_2}}}{\notf{\td{M}{\psi_1\psi_2}}}{\bool}$.

\item
If $\td{x}{\psi_2} = 1$ and $\td{y}{\psi_2} = y'$ then
$\coe{x.\notb{x}}{1}{y'}{\td{M}{\psi_1\psi_2}} \steps
\notel{y'}{\td{M}{\psi_1\psi_2}}$, and $\isval{\notel{y'}{\td{N}{\psi_2}}}$. By
$\coftype[\Psi_1]{\td{M}{\psi_1}}{\notb{x}}$, for any
$\msubsts{\Psi'}{\psi}{\Psi_2}$,
$\inper[\Psi']{\notel{1}{\td{N}{\psi_2\psi}} \steps \td{N}{\psi_2\psi}}
{\td{M}{\psi_1\psi_2\psi}}{\bool}$, so by \cref{lem:coftype-ceqtm},
$\ceqtm[\Psi_2]{\td{N}{\psi_2}}{\td{M}{\psi_1\psi_2}}{\bool}$ and so
$\vinper[\Psi_2]{\notel{y'}{\td{M}{\psi_1\psi_2}}}{\notel{y'}{\td{N}{\psi_2}}}{\notb{y'}}$.

\item
If $\td{x}{\psi_2} = 0$ and $\td{y}{\psi_2} = y'$ then
$\coe{x.\notb{x}}{0}{y'}{\td{M}{\psi_1\psi_2}} \steps
\notel{y'}{\notf{\td{M}{\psi_1\psi_2}}}$, and
$\isval{\notel{y'}{\td{N}{\psi_2}}}$. By
$\coftype[\Psi_1]{\td{M}{\psi_1}}{\notb{x}}$, for any
$\msubsts{\Psi'}{\psi}{\Psi_2}$,
$\inper[\Psi']{\notel{0}{\td{N}{\psi_2\psi}} \steps \notf{\td{N}{\psi_2\psi}}}
{\td{M}{\psi_1\psi_2\psi}}{\bool}$, so by \cref{lem:inper-not-swap},
$\ceqtm[\Psi_2]{\td{N}{\psi_2}}{\notf{\td{M}{\psi_1\psi_2}}}{\bool}$ and so
$\vinper[\Psi_2]
{\notel{y'}{\notf{\td{M}{\psi_1\psi_2}}}}
{\notel{y'}{\td{N}{\psi_2}}}{\notb{y'}}$.

\item
If $\td{x}{\psi_2} = x'$ and $\td{y}{\psi_2} = 1$ then
$\vinper[\Psi_2]
{\td{M}{\psi_1\psi_2} \evals \notel{x'}{N'}}
{\notel{x'}{\td{N}{\psi_2}}}{\notb{x'}}$ where
$\ceqtm[\Psi_2]{N'}{\td{N}{\psi_2}}{\bool}$. Moreover,
$\coe{x.\notb{x}}{x'}{1}{\td{M}{\psi_1\psi_2}} \steps^*
\coe{x.\notb{x}}{x'}{1}{\notel{x'}{N'}} \steps^* N'$, and
$\notel{1}{\td{N}{\psi_2}} \steps \td{N}{\psi_2}$. Then
$\inper[\Psi_2]{N'}{\td{N}{\psi_2}}{\bool}$ as desired.

\item
If $\td{x}{\psi_2} = x'$ and $\td{y}{\psi_2} = 0$ then
$\vinper[\Psi_2]
{\td{M}{\psi_1\psi_2} \evals \notel{x'}{N'}}
{\notel{x'}{\td{N}{\psi_2}}}{\notb{x'}}$ where
$\ceqtm[\Psi_2]{N'}{\td{N}{\psi_2}}{\bool}$. Moreover,
$\coe{x.\notb{x}}{x'}{0}{\td{M}{\psi_1\psi_2}} \steps^*
\coe{x.\notb{x}}{x'}{0}{\notel{x'}{N'}} \steps^* \notf{N'}$, and
$\notel{0}{\td{N}{\psi_2}} \steps \notf{\td{N}{\psi_2}}$. Then
$\inper[\Psi_2]{\notf{N'}}{\notf{\td{N}{\psi_2}}}{\bool}$ as desired.

\item
If $\td{x}{\psi_2} = x'$ and $\td{y}{\psi_2} = y'$ then
$\vinper[\Psi_2]
{\td{M}{\psi_1\psi_2} \evals \notel{x'}{N'}}
{\notel{x'}{\td{N}{\psi_2}}}{\notb{x'}}$ where
$\ceqtm[\Psi_2]{N'}{\td{N}{\psi_2}}{\bool}$. Moreover,
$\coe{x.\notb{x}}{x'}{y'}{\td{M}{\psi_1\psi_2}} \steps^*
\coe{x.\notb{x}}{x'}{y'}{\notel{x'}{N'}} \steps \notel{y'}{N'}$,
$\isval{\notel{y'}{\td{N}{\psi_2}}}$, and
$\vinper[\Psi_2]{\notel{y'}{N'}}{\notel{y'}{\td{N}{\psi_2}}}{\bool}$ as desired.
\end{enumerate}
\end{enumerate}

This concludes the proof of the fourth Kan condition.
The proofs of the first three Kan conditions (regarding $\hcomsym$) rely on
that result, as well as two additional lemmas:

\begin{lemma}
$\cpretype{\notb{\e}}$ is Kan.
\end{lemma}
\begin{proof}
This follows directly from head expansion and the fact that $\cpretype{\bool}$
is Kan, because $\hcomsym$ and $\coesym$ first evaluate their type argument and
$\notb{\e} \steps \bool$.
\end{proof}

\begin{lemma}\label{lem:ceqtm-notel-coe}
If $\coftype{M}{\notb{r}}$, then
$\ceqtm{\notel{r}{\coe{x.\notb{x}}{r}{1}{M}}}{M}{\notb{r}}$.
\end{lemma}
\begin{proof}
The introduction rule and fourth Kan condition of $\notb{x}$ imply
$\coftype{\notel{r}{\coe{x.\notb{x}}{r}{1}{M}}}{\notb{r}}$. Therefore by
\cref{lem:coftype-ceqtm} it suffices to show that for any $\msubst{\psi}$,
$\inper[\Psi']
{\notel{\td{r}{\psi}}{\coe{x.\notb{x}}{\td{r}{\psi}}{1}{\td{M}{\psi}}}}
{\td{M}{\psi}}
{A_0}$ where $\notb{\td{r}{\psi}} \evals A_0$.

\begin{enumerate}
\item If $\td{r}{\psi} = 0$ then
$\notel{0}{\coe{x.\notb{x}}{0}{1}{\td{M}{\psi}}} \steps
\notf{\coe{x.\notb{x}}{0}{1}{\td{M}{\psi}}} \steps
\notf{\notf{\td{M}{\psi}}}$.
By \cref{lem:ceqpretype-ceqtm},
$\coftype[\Psi']{\td{M}{\psi}}{\bool}$, so 
$\ceqtm[\Psi']{\notf{\notf{\td{M}{\psi}}}}{\td{M}{\psi}}{\bool}$ and therefore
$\inper[\Psi']{\notf{\notf{\td{M}{\psi}}}}{\td{M}{\psi}}{\bool}$.

\item If $\td{r}{\psi} = 1$ then
$\notel{1}{\coe{x.\notb{x}}{1}{1}{\td{M}{\psi}}} \steps
\coe{x.\notb{x}}{1}{1}{\td{M}{\psi}} \steps \td{M}{\psi}$, and
$\inper[\Psi']{\td{M}{\psi}}{\td{M}{\psi}}{\bool}$.

\item If $\td{r}{\psi} = x$ then
$\isval{\notel{x}{\coe{x.\notb{x}}{x}{1}{\td{M}{\psi}}}}$. We know
$\inper[\Psi']{\td{M}{\psi}}{\td{M}{\psi}}{\notb{x}}$, so
$\td{M}{\psi} \evals \notel{x}{N}$ where
$\coftype[\Psi']{N}{\bool}$. To show
$\vinper[\Psi']
{\notel{x}{\coe{x.\notb{x}}{x}{1}{\td{M}{\psi}}}}
{\notel{x}{N}}{\notb{x}}$, we must show
$\ceqtm[\Psi']{\coe{x.\notb{x}}{x}{1}{\td{M}{\psi}}}{N}{\bool}$.
Again, by \cref{lem:coftype-ceqtm} it suffices to show that for any
$\msubsts{\Psi''}{\psi'}{\Psi'}$,
$\inper[\Psi'']
{\coe{x.\notb{x}}{\td{x}{\psi'}}{1}{\td{M}{\psi\psi'}}}
{\td{N}{\psi'}}{\bool}$.
\begin{enumerate}
\item If $\td{x}{\psi'} = 0$ then
$\coe{x.\notb{x}}{0}{1}{\td{M}{\psi\psi'}} \steps
\notf{\td{M}{\psi\psi'}}$. Because
$\coftype[\Psi']{\td{M}{\psi}}{\notb{x}}$ and
$\td{M}{\psi} \evals \notel{x}{N}$,
we know that for any $\msubsts{\Psi'''}{\psi''}{\Psi''}$,
$\inper[\Psi''']
{\notel{0}{\td{N}{\psi'\psi''}} \steps \notf{\td{N}{\psi'\psi''}}}
{\td{M}{\psi\psi'\psi''}}{\bool}$.
Therefore by \cref{lem:inper-not-swap} we have
$\inper[\Psi'']{\notf{\td{M}{\psi\psi'}}}{\td{N}{\psi'}}{\bool}$, which is what
we needed.

\item If $\td{x}{\psi'} = 1$ then
$\coe{x.\notb{x}}{1}{1}{\td{M}{\psi\psi'}} \steps \td{M}{\psi\psi'}$. By
$\coftype[\Psi']{\td{M}{\psi}}{\notb{x}}$ and
$\td{M}{\psi} \evals \notel{x}{N}$, we have
$\inper[\Psi'']
{\notel{1}{\td{N}{\psi'}} \steps \td{N}{\psi'}}
{\td{M}{\psi\psi'}}{\bool}$, which is what we needed.

\item If $\td{x}{\psi'} = x'$ then by
$\coftype[\Psi']{\td{M}{\psi}}{\notb{x}}$ and
$\td{M}{\psi} \evals \notel{x}{N}$, we have
$\inper[\Psi'']{\notel{x'}{\td{N}{\psi'}}}{\td{M}{\psi\psi'}}{\notb{x'}}$,
so $\td{M}{\psi\psi'} \evals \notel{x'}{N'}$ and
$\ceqtm[\Psi'']{\td{N}{\psi'}}{N'}{\bool}$. Therefore
$\coe{x.\notb{x}}{x'}{1}{\td{M}{\psi\psi'}} \steps^*
\coe{x.\notb{x}}{x'}{1}{\notel{x'}{N'}} \steps
\notel{1}{N'} \steps N'$ and we must show
$\inper[\Psi'']{N'}{\td{N}{\psi'}}{\notb{x'}}$, which follows from
$\ceqtm[\Psi'']{\td{N}{\psi'}}{N'}{\bool}$.
\qedhere
\end{enumerate}
\end{enumerate}
\end{proof}
In particular, if $\coftype{M}{\bool}$, then
$\ceqtm{\notf{\coe{x.\notb{x}}{0}{1}{M}}}{M}{\bool}$.

The first Kan condition asserts that
for any $\msubsts{(\Psi',x)}{\psi}{(\Psi,w)}$, if
\begin{enumerate}
\item \ceqtm[\Psi',x]{M}{O}{\notb{\td{w}{\psi}}},
\item \ceqtm[\Psi',y]{\dsubst{N^\e}{\e}{x}}{\dsubst{P^\e}{\e}{x}}{\notb{\dsubst{\td{w}{\psi}}{\e}{x}}} for
$\e=0,1$, and
\item \ceqtm[\Psi']{\dsubst{\dsubst{N^\e}{r}{y}}{\e}{x}}{\dsubst{M}{\e}{x}}{\notb{\dsubst{\td{w}{\psi}}{\e}{x}}} for $\e=0,1$,
\end{enumerate}
then $\ceqtm[\Psi',x]{\hcomgeneric{x}{\notb{\td{w}{\psi}}}}{\hcom{x}{\notb{\td{w}{\psi}}}{r}{r'}{O}{y.P^0,y.P^1}}{\notb{\td{w}{\psi}}}$.

Let $\msubsts{\Psi_1}{\psi_1}{(\Psi',x)}$ and
$\msubsts{\Psi_2}{\psi_2}{\Psi_1}$. We will again focus on the unary case.

\begin{enumerate}
\item If $\td{w}{\psi\psi_1} = \e$ then
$\coftype[\Psi_1]
{\hcom{\td{x}{\psi_1}}{\notb{\e}}{\td{r}{\psi_1}}{\td{r'}{\psi_1}}
{\td{M}{\psi_1}}{y.\td{N^0}{\psi_1},y.\td{N^0}{\psi_1}}}
{\notb{\e}}$
by the first Kan condition of $\notb{\e}$, or the third Kan condition if
$\td{x}{\psi_1} = \e'$. Therefore $\td{\hcomsym}{\psi_1} \evals X_1$ and
$\inper[\Psi_2]{\td{X_1}{\psi_2}}{\td{\hcomsym}{\psi_1\psi_2}}{\bool}$.

\item If $\td{w}{\psi\psi_1} = w'$ then $\td{\hcomsym}{\psi_1} \evals$
\[ 
\notel{w'}{
  \hcom{\td{x}{\psi_1}}{\bool}{\td{r}{\psi_1}}{\td{r'}{\psi_1}}
  {\coe{x.\notb{x}}{w'}{1}{\td{M}{\psi_1}}}
  {y.\coe{x.\notb{x}}{w'}{1}{\td{N^0}{\psi_1}},
   y.\coe{x.\notb{x}}{w'}{1}{\td{N^1}{\psi_1}}}}
\]
Let $H$ be the argument of the above $\notel{w'}{-}$. We must show
$\inper[\Psi_2]{\td{\hcomsym}{\psi_1\psi_2}}{\notel{\td{w'}{\psi_2}}{\td{H}{\psi_2}}}{A_{12}}$
where $\notb{\td{w'}{\psi_2}} \evals A_{12}$.

For all $\psi_2$,
$\coftype[\Psi_2]{\td{M}{\psi_1\psi_2}}{\notb{\td{w'}{\psi_2}}}$,
so by the fourth Kan condition of $\notb{x}$,
$\coftype[\Psi_2]{\coe{x.\notb{x}}{\td{w'}{\psi_2}}{1}{\td{M}{\psi_1\psi_2}}}{\bool}$.
If $\td{x}{\psi_1\psi_2} = \e$ then
$\td{\dsubst{}{\e}{x}}{\psi_1\psi_2} = \td{}{\psi_1\psi_2}$, so
$\coftype[\Psi_2,y]{\td{N^\e}{\psi_1\psi_2}}{\notb{\td{w'}{\psi_2}}}$ and thus
$\coftype[\Psi_2,y]{\coe{x.\notb{x}}{\td{w'}{\psi_2}}{1}{\td{N^\e}{\psi_1\psi_2}}}{\bool}$
and the adjacency condition holds as well.
In this case, by the third Kan condition of $\bool$,
$\ceqtm[\Psi_2]{\td{H}{\psi_2}}
{\coe{x.\notb{x}}{\td{w'}{\psi_2}}{1}{\td{\dsubst{N^\e}{r'}{y}}{\psi_1\psi_2}}}
{\bool}$.

If instead $\td{x}{\psi_1\psi_2} = x'$ then $\Psi_2 = (\Psi_2',x')$ and 
$\td{\dsubst{}{\e}{x}}{\psi_1\psi_2} = \dsubst{\td{}{\psi_1\psi_2}}{\e}{x'}$. Then
$\coftype[\Psi_2',y]{\dsubst{\td{N^\e}{\psi_1\psi_2}}{\e}{x'}}{\notb{\dsubst{\td{w'}{\psi_2}}{\e}{x'}}}$,
so $\coftype[\Psi_2',y]
{\coe{x.\notb{x}}{\dsubst{\td{w'}{\psi_2}}{\e}{x'}}{1}{\dsubst{\td{N^\e}{\psi_1\psi_2}}{\e}{x'}}}
{\notb{\dsubst{\td{w'}{\psi_2}}{\e}{x'}}}$
and the adjacency condition holds as well.
In this case, by the first Kan condition of $\bool$,
$\coftype[\Psi_2]{\td{H}{\psi_2}}{\bool}$.

\begin{enumerate}
\item If $\td{w'}{\psi_2} = 0$ then each side steps once to
\[
\inper[\Psi_2]
{\hcom{\td{x}{\psi_1\psi_2}}{\bool}{\td{r}{\psi_1\psi_2}}{\td{r'}{\psi_1\psi_2}}
{\td{M}{\psi_1\psi_2}}{y.\td{N^0}{\psi_1\psi_2},y.\td{N^1}{\psi_1\psi_2}}}
{\notf{\td{H}{\psi_2}}}{\bool}
\]
If $\td{x}{\psi_1\psi_2} = \e$ then 
by the third Kan condition of $\bool$, it
suffices to show
$\inper[\Psi_2]
{\td{\dsubst{N^\e}{r'}{y}}{\psi_1\psi_2}}
{\notf{\coe{x.\notb{x}}{0}{1}{\td{\dsubst{N^\e}{r'}{y}}{\psi_1\psi_2}}}}
{\bool}$,
which follows from \cref{lem:ceqtm-notel-coe}.

If $\td{x}{\psi_1\psi_2} = x'$ and $\td{r}{\psi_1\psi_2}=\td{r'}{\psi_1\psi_2}$
then by the second Kan condition of $\bool$, it suffices to show
$\inper[\Psi_2]
{\td{M}{\psi_1\psi_2}}
{\notf{\coe{x.\notb{x}}{0}{1}{\td{M}{\psi_1\psi_2}}}}
{\bool}$,
which follows from \cref{lem:ceqtm-notel-coe}.

Otherwise $\td{x}{\psi_1\psi_2} = x'$ and $\td{r}{\psi_1\psi_2} \neq
\td{r'}{\psi_1\psi_2}$, and the right-hand side steps twice to
\[
\hcom{\td{x}{\psi_1\psi_2}}{\bool}{\td{r}{\psi_1\psi_2}}{\td{r'}{\psi_1\psi_2}}
{\notf{\coe{x.\notb{x}}{0}{1}{\td{M}{\psi_1\psi_2}}}}
{y.\notf{\coe{x.\notb{x}}{0}{1}{\td{N^0}{\psi_1\psi_2}}},\dots}
\]
which $\eq$ the left-hand side by \cref{lem:ceqtm-notel-coe} and the first Kan
condition of $\bool$.

\item If $\td{w'}{\psi_2} = 1$ then each side steps once to
\[
\inper[\Psi_2]
{\hcom{\td{x}{\psi_1\psi_2}}{\bool}{\td{r}{\psi_1\psi_2}}{\td{r'}{\psi_1\psi_2}}
{\td{M}{\psi_1\psi_2}}{y.\td{N^0}{\psi_1\psi_2},y.\td{N^1}{\psi_1\psi_2}}}
{\td{H}{\psi_2}}{\bool}
\]
If $\td{x}{\psi_1\psi_2} = \e$ then by the third Kan condition of $\bool$, it
suffices to show
$\inper[\Psi_2]{\td{M}{\psi_1\psi_2}}{\coe{x.\notb{x}}{1}{1}{\td{M}{\psi_1\psi_2}}}{\bool}$,
which follows from the computation rule for $\notb{x}$.

If $\td{x}{\psi_1\psi_2} = x'$ then the result follows from the first Kan
condition of $\bool$ and the computation rule for $\notb{x}$.

\item If $\td{w'}{\psi_2} = w''$ then each side steps once to
\[
\vinper[\Psi_2]{\notel{w''}{\td{H}{\psi_2}}}{\notel{w''}{\td{H}{\psi_2}}}{\notb{w''}}
\]
which follows from $\coftype[\Psi_2]{\td{H}{\psi_2}}{\bool}$.
\end{enumerate}
\end{enumerate}

The second Kan condition asserts that for any
$\msubsts{(\Psi',x)}{\psi}{(\Psi,w)}$, if
\begin{enumerate}
\item \coftype[\Psi',x]{M}{\notb{\td{w}{\psi}}},
\item \coftype[\Psi',y]{\dsubst{N^\e}{\e}{x}}{\notb{\dsubst{\td{w}{\psi}}{\e}{x}}} for $\e=0,1$, and
\item \ceqtm[\Psi']{\dsubst{\dsubst{N^\e}{r}{y}}{\e}{x}}{\dsubst{M}{\e}{x}}{\notb{\dsubst{\td{w}{\psi}}{\e}{x}}} for $\e=0,1$,
\end{enumerate}
then
$\ceqtm[\Psi',x]{\hcom{x}{\notb{\td{w}{\psi}}}{r}{r}{M}{y.N^0,y.N^1}}{M}{\notb{\td{w}{\psi}}}$.

The first Kan condition implies that the left-hand side is
$\coftype[\Psi',x]{-}{\notb{\td{w}{\psi}}}$, so by \cref{lem:coftype-ceqtm}, it
suffices to show that for all $\msubsts{\Psi''}{\psi'}{(\Psi',x)}$,
\[
\inper[\Psi'']
{\hcom{\td{x}{\psi'}}{\notb{\td{w}{\psi\psi'}}}{\td{r}{\psi'}}{\td{r}{\psi'}}{\td{M}{\psi'}}{y.\td{N^0}{\psi'},y.\td{N^1}{\psi'}}}
{\td{M}{\psi'}}
{A_{12}}
\]
where $\notb{\td{w}{\psi\psi'}} \evals A_{12}$.
\begin{enumerate}
\item If $\td{w}{\psi\psi'} = \e$ then
if $\td{x}{\psi'} = x'$ the second Kan condition of $\notb{\e}$ implies that
the left-hand side is $\ceqtm[\Psi'']{-}{\td{M}{\psi'}}{\notb{\e}}$,
and therefore also $\inper[\Psi'']{-}{\td{M}{\psi'}}{\bool}$.
Otherwise $\td{x}{\psi'} = \e'$ and by the third Kan condition of $\notb{\e}$,
the left-hand side is
$\ceqtm[\Psi'']{-}{\td{\dsubst{N^{\e'}}{r}{y}}{\psi'}}{\notb{\e}}$. 
The result follows because $\td{\dsubst{}{\e'}{x}}{\psi'} = \td{}{\psi'}$ and so
$\ceqtm[\Psi'']{\td{\dsubst{N^{\e'}}{r}{y}}{\psi'}}{\td{M}{\psi'}}{\notb{\e}}$.

\item If $\td{w}{\psi\psi'} = w'$ then the left side steps to
\[ 
\notel{w'}{
  \hcom{\td{x}{\psi'}}{\bool}{\td{r}{\psi'}}{\td{r}{\psi'}}
  {\coe{x.\notb{x}}{w'}{1}{\td{M}{\psi'}}}
  {y.\coe{x.\notb{x}}{w'}{1}{\td{N^0}{\psi'}},
   y.\coe{x.\notb{x}}{w'}{1}{\td{N^1}{\psi'}}}}
\]
If $\td{x}{\psi'} = x'$ the second Kan condition of $\bool$ (along with the
introduction rule and fourth Kan condition of $\notb{x}$) implies this is
$\ceqtm[\Psi'']{-}
{\notel{w'}{\coe{x.\notb{x}}{w'}{1}{\td{M}{\psi'}}}}
{\notb{w'}}$
which by \cref{lem:ceqtm-notel-coe} is
$\ceqtm[\Psi'']{-}{\td{M}{\psi'}}{\notb{w'}}$.
Otherwise $\td{x}{\psi'} = \e'$ and by the third Kan condition of $\bool$ (and
the introduction rule and fourth Kan condition of $\notb{x}$) this is
$\ceqtm[\Psi'']{-}
{\notel{w'}{\coe{x.\notb{x}}{w'}{1}{\td{\dsubst{N^{\e'}}{r}{y}}{\psi'}}}}
{\notb{w'}}$, which by
\cref{lem:ceqtm-notel-coe} is
$\ceqtm[\Psi'']{-}{\td{\dsubst{N^{\e'}}{r}{y}}{\psi'}}{\notb{w'}}$ and by
$\td{\dsubst{}{\e'}{x}}{\psi'} = \td{}{\psi'}$ is
$\ceqtm[\Psi'']{-}{\td{M}{\psi'}}{\notb{w'}}$.
\end{enumerate}

The third Kan condition asserts that for any
$\msubsts{\Psi'}{\psi}{(\Psi,w)}$, if
\begin{enumerate}
\item \coftype[\Psi']{M}{\notb{\td{w}{\psi}}},
\item \coftype[\Psi',y]{N^\e}{\notb{\td{w}{\psi}}}, and
\item
\ceqtm[\Psi']{\dsubst{N^\e}{r}{y}}{M}{\notb{\td{w}{\psi}}},
\end{enumerate}
then $\ceqtm[\Psi']{\hcomgeneric{\e}{\notb{\td{w}{\psi}}}}{\dsubst{N^\e}{r'}{y}}{\notb{\td{w}{\psi}}}$.

Let $\msubsts{\Psi_1}{\psi_1}{\Psi'}$ and
$\msubsts{\Psi_2}{\psi_2}{\Psi_1}$. We must show
$\td{\hcomsym}{\psi_1} \evals X_1$, 
$\inper[\Psi_2]{\td{X_1}{\psi_2}}
{\td{\dsubst{N^\e}{r'}{y}}{\psi_1\psi_2}}{A_{12}}$, and
$\inper[\Psi_2]{\td{\hcomsym}{\psi_1\psi_2}}
{\td{\dsubst{N^\e}{r'}{y}}{\psi_1\psi_2}}{A_{12}}$ where
$\notb{\td{w}{\psi\psi_1\psi_2}} \evals A_{12}$.

\begin{enumerate}
\item If $\td{w}{\psi\psi_1} = \e'$ then
$\ceqtm[\Psi_1]{\td{\hcomsym}{\psi_1}}
{\td{\dsubst{N^\e}{r'}{y}}{\psi_1}}
{\notb{\e'}}$ by the third Kan condition of $\notb{\e'}$, and the
result follows.

\item If $\td{w}{\psi\psi_1} = w'$ then
$\td{\hcomsym}{\psi_1} \steps$
\[ 
\notel{w'}{
  \hcom{\e}{\bool}{\td{r}{\psi_1}}{\td{r'}{\psi_1}}
  {\coe{x.\notb{x}}{w'}{1}{\td{M}{\psi_1}}}
  {y.\coe{x.\notb{x}}{w'}{1}{\td{N^0}{\psi_1}},
   y.\coe{x.\notb{x}}{w'}{1}{\td{N^1}{\psi_1}}}}
\]
By the third Kan condition of $\bool$, this is
$\ceqtm[\Psi_1]{-}
{\notel{w'}{\coe{x.\notb{x}}{w'}{1}{\td{\dsubst{N^\e}{r'}{y}}{\psi_1}}}}
{\notb{w'}}$, which by \cref{lem:ceqtm-notel-coe} is
$\ceqtm[\Psi_1]{-}{\td{\dsubst{N^\e}{r'}{y}}{\psi_1}}{\notb{w'}}$.
Therefore $\td{\hcomsym}{\psi_1} \evals X_1$ and
$\inper[\Psi_2]{\td{X_1}{\psi_2}}
{\td{\dsubst{N^\e}{r'}{y}}{\psi_1\psi_2}}{A_{12}}$ where
$\notb{\td{w'}{\psi_2}} \evals A_{12}$.
\begin{enumerate}
\item If $\td{w'}{\psi_2} = \e'$ then
$\ceqtm[\Psi_2]{\td{\hcomsym}{\psi_1\psi_2}}
{\td{\dsubst{N^\e}{r'}{y}}{\psi_1\psi_2}}
{\notb{\e'}}$ by the third Kan condition of $\notb{\e'}$, so
$\inper[\Psi_2]{\td{\hcomsym}{\psi_1\psi_2}}
{\td{\dsubst{N^\e}{r'}{y}}{\psi_1\psi_2}}{\bool}$.

\item If $\td{w'}{\psi_2} = w''$ then
$\td{\hcomsym}{\psi_1\psi_2} \steps$
\[ 
\notel{w''}{
  \hcom{\e}{\bool}{\td{r}{\psi_1\psi_2}}{\td{r'}{\psi_1\psi_2}}
  {\coe{x.\notb{x}}{w''}{1}{\td{M}{\psi_1\psi_2}}}
  {y.\coe{x.\notb{x}}{w''}{1}{\td{N^0}{\psi_1\psi_2}},\dots}}
\]
By the third Kan condition of $\bool$, this is
$\ceqtm[\Psi_2]{-}
{\notel{w''}{\coe{x.\notb{x}}{w''}{1}{\td{\dsubst{N^\e}{r'}{y}}{\psi_1\psi_2}}}}
{\notb{w''}}$, which by \cref{lem:ceqtm-notel-coe} is
$\ceqtm[\Psi_2]{-}{\td{\dsubst{N^\e}{r'}{y}}{\psi_1\psi_2}}{\notb{w''}}$, so
$\inper[\Psi_2]{\td{\hcomsym}{\psi_1\psi_2}}
{\td{\dsubst{N^\e}{r'}{y}}{\psi_1\psi_2}}{\notb{w''}}$.
\end{enumerate}
\end{enumerate}

\paragraph{Cubical}
Show for any $\msubsts{\Psi'}{\psi}{(\Psi,x)}$ and
$\vinper[\Psi']{M}{N}{A_0}$ (where $\notb{\td{x}{\psi}}\evals A_0$) then
$\ceqtm[\Psi']{M}{N}{\notb{\td{x}{\psi}}}$.

If $\td{x}{\psi} = \e$ then $A_0 = \bool$ and we have
$\ceqtm[\Psi']{M}{N}{\bool}$ because $\cpretype[\Psi']{\bool}$ is cubical. By
\cref{lem:ceqpretype-ceqtm}, this implies $\ceqtm[\Psi']{M}{N}{\notb{\e}}$.
If $\td{x}{\psi} = x'$ then this follows from the introduction rule for
$\notb{x'}$.

%% file: summary.tex
\section{Summary}

In this section we summarize the results of \cref{sec:types} in rule notation.
These rules are not intended to define a conventional proof theory.
However, if one were to inductively define a proof theory with these rules (and
structural rules such as hypothesis, weakening, etc.), the result would indeed
be sound for our computational semantics, in the sense that the conclusion of
each rule is true given that the premises are true. From this perspective, our
computational semantics are a model of higher type theory validating the
following \emph{canonicity theorem}:

\begin{theorem}[Canonicity]
If $\oftype[\cdot]{\cdot}{M}{\bool}$ then either
$M \evals \true$ or $M \evals \false$.
\end{theorem}
\begin{proof}
By the definition of $\coftype{M}{A}$, it follows that
$M \evals M_0$ such that $\vinper[\cdot]{M_0}{M_0}{\bool}$, which implies that
$M_0 = \true$ or $M_0 = \false$.
\end{proof}

Our choice of rules is inspired by the formal cubical type theories given
by~\citet{cohen2016cubical,licata2014cubical} so as to make clear that
our computational semantics are a valid interpretation of those rules.
However, these semantics may be used to justify concepts,
such as strict types, that are not currently considered in the formal
setting.  Moreover, there is no strong reason to limit consideration
to inductively defined proof theories.  The role of a proof theory is
to provide access to the truth, in particular to support
mechanization.  But there are methods of accessing the truth, such as
decision procedures for arithmetic, that do not fit into the
conventional setup for proof theory.  (This point was stressed for the
the NuPRL type theory~\citep{constableetalnuprl}; we are merely
reiterating it here.)

\begin{remark}
Some complexities in the rules below have been suppressed for the sake of
clarity. The introduction and elimination rules omit respect for equality. Also,
all of the rules for function and product types should contain the hypotheses
$\cwftype{A}$ and $\cwftype{B}$.
\end{remark}

\begin{remark}
The theorems in \cref{sec:types} are stated only for closed terms. The
corresponding generalizations to open-term sequents, below, follow from
\cref{lem:eqtm-tsubst}, the fact that the introduction and elimination rules
respect equality (proven in \cref{sec:types}), and the fact that all
substitutions commute with term formers.
\end{remark}

\paragraph{Function types}

\[
\infer
  {\cwftype{\arr{A}{B}}}
  {\cwftype{A} & \cwftype{B}}
\]

\[
\infer
  {\oftype\G{\lam{a}{M}}{\arr{A}{B}}}
  {\oftype{\Gamma,\oft aA}{M}{B}}
\]

\[
\infer
  {\oftype\G{\app{M}{N}}{B}}
  {\oftype\G{M}{\arr{A}{B}} &
   \oftype\G{N}{A}}
\]

\[
\infer
  {\eqtm\G{\app{\lam{a}{M}}{N}}{\subst{M}{N}{a}}{B}}
  {\oftype{\Gamma,\oft aA}{M}{B} &
   \oftype\G{N}{A}}
\]

\[
\infer
  {\eqtm\G{M}{\lam{a}{\app{M}{a}}}{\arr{A}{B}}}
  {\oftype\G{M}{\arr{A}{B}}}
\]

\paragraph{Product types}

\[
\infer
  {\cwftype{\prd{A}{B}}}
  {\cwftype{A} & \cwftype{B}}
\]

\[
\infer
  {\oftype\G{\pair{M}{N}}{\prd{A}{B}}}
  {\oftype\G{M}{A} &
   \oftype\G{N}{B}}
\]

\[
\infer
  {\oftype\G{\fst{P}}{A}}
  {\oftype\G{P}{\prd{A}{B}}}
\qquad
\infer
  {\oftype\G{\snd{P}}{B}}
  {\oftype\G{P}{\prd{A}{B}}}
\]

\[
\infer
  {\eqtm\G{\fst{\pair{M}{N}}}{M}{A}}
  {\oftype\G{M}{A} &
   \oftype\G{N}{B}}
\qquad
\infer
  {\eqtm\G{\snd{\pair{M}{N}}}{N}{B}}
  {\oftype\G{M}{A} &
   \oftype\G{N}{B}}
\]

\[
\infer
  {\eqtm{\G}{P}{\pair{\fst{P}}{\snd{P}}}{\prd{A}{B}}}
  {\oftype{\G}{P}{\prd{A}{B}}}
\]

\paragraph{Booleans}

\[
\infer{\cwftype{\bool}}{}
\]

\[
\infer{\oftype\G{\true}{\bool}}{}
\qquad
\infer{\oftype\G{\false}{\bool}}{}
\]

\[
\infer
  {\oftype\G{\ifb{A}{M}{T}{F}}{A}}
  {\cwftype{A} &
   \oftype\G{M}{\bool} &
   \oftype\G{T}{A} &
   \oftype\G{F}{A}}
\]

\[
\infer
  {\eqtm\G{\ifb{A}{\true}{T}{F}}{T}{A}}
  {\cwftype{A} &
   \oftype\G{T}{A} &
   \oftype\G{F}{A}}
\]

\[
\infer
  {\eqtm\G{\ifb{A}{\false}{T}{F}}{F}{A}}
  {\cwftype{A} &
   \oftype\G{T}{A} &
   \oftype\G{F}{A}}
\]

\paragraph{Circle}

\[
\infer{\cwftype{\C}}{}
\]

\[
\infer
  {\oftype\G{\base}{\C}}
  {}
\qquad
\infer
  {\oftype\G{\lp{r}}{\C}}
  {\wfdim{r}}
\qquad
\infer
  {\eqtm\G{\lp{\e}}{\base}{\C}}
  {}
\]

\[
\infer
  {\oftype\G{\Celim{A}{M}{P}{x.L}}{A}}
  {\cwftype{A} &
   \oftype\G{M}{\C} &
   \oftype[\Psi,x]\G{L}{A} &
   (\forall\e)\ \eqtm\G{\dsubst{L}{\e}{x}}{P}{A}}
\]

\[
\infer
  {\eqtm\G{\Celim{A}{\base}{P}{x.L}}{P}{A}}
  {\cwftype{A} &
   \oftype[\Psi,x]\G{L}{A} &
   (\forall\e)\ \eqtm\G{\dsubst{L}{\e}{x}}{P}{A}}
\]

\[
\infer
  {\eqtm\G{\Celim{A}{\lp{r}}{P}{x.L}}{\dsubst{L}{r}{x}}{A}}
  {\cwftype{A} &
   \wfdim{r} &
   \oftype[\Psi,x]\G{L}{A} &
   (\forall\e)\ \eqtm\G{\dsubst{L}{\e}{x}}{P}{A}}
\]

\paragraph{Hcom}

\[
\infer
  {\oftype[\Psi,x]\G{\hcomgeneric{x}{A}}{A}}
  {\begin{array}{ll}
   &\cwftype[\Psi,x]{A} \\
   &\wfdim[\Psi,x]{r,r'} \\
   &\oftype[\Psi,x]\G{M}{A} \\
   (\forall\e) &\oftype[\Psi,y]{\dsubst{\G}{\e}{x}}{\dsubst{N^\e}{\e}{x}}{\dsubst{A}{\e}{x}} \\
   (\forall\e) &\eqtm{\dsubst{\G}{\e}{x}}{\dsubst{\dsubst{N^\e}{r}{y}}{\e}{x}}{\dsubst{M}{\e}{x}}{\dsubst{A}{\e}{x}}
   \end{array}}
\]

\[
\infer
  {\eqtm[\Psi,x]\G{\hcom{x}{A}{r}{r}{M}{y.N^0,y.N^1}}{M}{A}}
  {\begin{array}{ll}
   &\cwftype[\Psi,x]{A} \\
   &\wfdim{r} \\
   &\oftype[\Psi,x]\G{M}{A} \\
   (\forall\e) &\oftype[\Psi,y]{\dsubst{\G}{\e}{x}}{\dsubst{N^\e}{\e}{x}}{\dsubst{A}{\e}{x}} \\
   (\forall\e) &\eqtm{\dsubst{\G}{\e}{x}}{\dsubst{\dsubst{N^\e}{r}{y}}{\e}{x}}{\dsubst{M}{\e}{x}}{\dsubst{A}{\e}{x}}
   \end{array}}
\]

\[
\infer
  {\eqtm\G{\hcomgeneric{\e}{A}}{\dsubst{N^\e}{r'}{y}}{A}}
  {\begin{array}{l}
   \cwftype{A} \\
   \wfdim{r,r'} \\
   \oftype\G{M}{A} \\
   \oftype[\Psi,y]\G{N^\e}{A} \\
   \eqtm\G{\dsubst{N^\e}{r}{y}}{M}{A}
   \end{array}}
\]

\paragraph{Coe}

\[
\infer
  {\oftype\G{\coegeneric{x.A}}{\dsubst{A}{r'}{x}}}
  {\cwftype[\Psi,x]{A} &
   \wfdim{r,r'} &
   \oftype\G{M}{\dsubst{A}{r}{x}}}
\]

\paragraph{Not}

\[
\infer{\cwftype[\Psi,x]{\notb{x}}}{}
\qquad
\infer{\ceqtype{\notb{\e}}{\bool}}{}
\]

\[
\infer
  {\eqtm\G{\coe{x.\notb{x}}{\e}{\e}{M}}{M}{\bool}}
  {\oftype\G{M}{\bool}}
\qquad
\infer
  {\eqtm\G{\coe{x.\notb{x}}{\e}{\eb}{M}}{\notf{M}}{\bool}}
  {\oftype\G{M}{\bool}}
\]

%%% Local Variables:
%%% mode: latex
%%% TeX-master: "heo"
%%% End: